\documentclass[a4paper,reqno,11pt]{amsart} %Loads amsthm, amsmath, amsfonts + a4paper and right-numbered equations

\usepackage[T1]{fontenc} 
\usepackage[utf8]{inputenc}
\usepackage[english]{babel} 
\usepackage{newpxtext} %Palatino

\usepackage{scalerel}

\usepackage{amssymb,textcomp,mathrsfs,stmaryrd,braket,commath,relsize}
\usepackage{bm} %To use many alphabets
	\SetSymbolFont{stmry}{bold}{U}{stmry}{m}{n} %Avoids meaningless warning for bold math symbols
\usepackage[new]{old-arrows} %For long hook-arrows
    
\usepackage{mathtools}
\usepackage{thm-restate}

\theoremstyle{plain} %Math enviroments I
    \newtheorem{theorem}{Theorem}[section]
    \newtheorem*{theorem*}{Theorem}

    \newtheorem{proposition}[theorem]{Proposition}
    \newtheorem*{proposition*}{Proposition}
    
	\newtheorem{corollary}[theorem]{Corollary}
    \newtheorem*{corollary*}{Corollary}
    
	\newtheorem{lemma}[theorem]{Lemma}
    \newtheorem*{lemma*}{Lemma}
    
	\newtheorem{conjecture}[theorem]{Conjecture}
    \newtheorem*{conjecture*}{Conjecture}
    
\theoremstyle{definition} %Math enviroments II
    \newtheorem{definition}[theorem]{Definition}
    \newtheorem*{definition*}{Definition}

    \newtheorem{notation}[theorem]{Notation}
    \newtheorem*{notation*}{Notation}

\theoremstyle{remark} %Math enviroments III
    \newtheorem{remark}[theorem]{Remark}
    \newtheorem*{remark*}{Remark}
	
	\newtheorem{example}[theorem]{Example}
	\newtheorem*{example*}{Example}	
	
\numberwithin{equation}{section}
\numberwithin{figure}{section}

\usepackage{etoolbox} 
    \newcommand{\addQEDstyle}[2]{\AtBeginEnvironment{#1}{\pushQED{\qed}\renewcommand{\qedsymbol}{#2}}
    \AtEndEnvironment{#1}{\popQED}} %Symbol at the end of environment: call as \addQEDstyle{environmentname}{symbolname}
    \addQEDstyle{remark}{$\circ$}
    \addQEDstyle{remark*}{$\circ$} %Triangles at the end of remarks
	\addQEDstyle{example}{$\circ$}
    \addQEDstyle{example*}{$\circ$} %Triangles at the end of examples

\usepackage{tikz}
	\usetikzlibrary{cd} %Commutative diagrams
	
\AtBeginDocument{\def\MR#1{}} %No MR number in bibliography
\apptocmd{\sloppy}{\hbadness 10000\relax}{}{} %For badboxes in .bbl (uses etoolbox)

\makeatletter %Print 2020AMSsubjclass (instead of 2010)
	\@namedef{subjclassname@2020}{
		\textup{2020} Mathematics Subject Classification}
\makeatother

\usepackage{ytableau}
\usepackage[percent]{overpic}
\usepackage{enumitem}
    
\allowdisplaybreaks

\usepackage[colorlinks,psdextra,pdfencoding=auto]{hyperref}
\hypersetup{colorlinks,linkcolor=red,filecolor=green,urlcolor=blue,citecolor=blue} 

\begin{document}

%%%%%%%%%%%%%%%%%%%%%%%%%%%%%%%%%%%%%%%%%%%%%%%%%%%%%%%%%%%%%%%%%%%%%%%%%%%%%%%%%%%%%%%%%%%%%%%%

\global\long\def\bR{\mathbb{R}}
\global\long\def\bRpos{\mathbb{R}_{> 0}}
\global\long\def\bRnn{\mathbb{R}_{\geq 0}}
\global\long\def\bZ{\mathbb{Z}}
\global\long\def\bN{\mathbb{N}}
\global\long\def\bZpos{\mathbb{Z}_{>0}}
\global\long\def\bZneg{\mathbb{Z}_{<0}}
\global\long\def\bZnn{\mathbb{Z}_{\geq 0}}
\global\long\def\bQ{\mathbb{Q}}
\global\long\def\bC{\mathbb{C}}
\global\long\def\bH{\mathbb{H}}
\global\long\def\bD{\mathbb{D}}

 \global\long\def\sF{\mathcal{F}}
 \global\long\def\sZ{\mathcal{Z}}
 \global\long\def\sD{\mathcal{D}}
 \global\long\def\sG{\mathcal{G}}
 \global\long\def\sC{\mathcal{C}}
 \global\long\def\sL{\mathcal{L}}
 \global\long\def\sA{\mathcal{A}}
 \global\long\def\sE{\mathcal{E}}
 \global\long\def\sR{\mathcal{R}}
 \global\long\def\sS{\mathcal{S}}
 \global\long\def\sP{\mathcal{P}}
 \global\long\def\sM{\mathcal{M}}
 \global\long\def\sV{\mathcal{V}}

\global\long\def\ii{\mathfrak{i}}

\newcommand{\SymGrp}{\mathfrak{S}}
\newcommand{\LP}{\mathsf{LP}}
\newcommand{\LS}{\mathsf{L}}
\newcommand{\DP}{\mathsf{DP}}
\newcommand{\TL}{\mathsf{TL}}
\newcommand{\Hecke}{\mathsf{H}}
\newcommand{\Uqsltwo}{{\mathsf{U}_q}}

\newcommand{\np}{d}
\newcommand{\Summed}{n}
\newcommand{\multii}{\varsigma}

\newcommand{\NB}{\textnormal{NB}^\lambda}
\newcommand{\SYT}{\textnormal{SYT}^\lambda}
\newcommand{\RSYT}{\textnormal{RSYT}^\lambda_\multii}
\newcommand{\CSYT}{\textnormal{CSYT}^\lambda_\multii}
\newcommand{\CSYTBar}{\textnormal{CSYT}^{\Bar{\lambda}}_\multii}
\newcommand{\Fill}{\textnormal{Fill}^\lambda_\multii}
\newcommand{\Fillof}[2]{{\textnormal{Fill}\sub{#1}^{#2}}}

\newcommand{\SYTof}[1]{{\textnormal{SYT}\super{#1}}}
\newcommand{\CSYTof}[1]{{\textnormal{CSYT}\super{#1}_\multii}}
\newcommand{\RSYTof}[1]{{\textnormal{RSYT}\super{#1}_\multii}}

\newcommand{\Spn}{\textnormal{span}}
\newcommand{\Det}{\textnormal{Det}}
\newcommand{\cconf}{\textnormal{Adm}}
\newcommand{\Stab}[2]{{\textnormal{Stab}_{#1}({#2})}}

\newcommand{\ord}{\textnormal{ord}}

\newcommand{\idpt}{\mathrm{p}_\multii}%{p_\multii}
\newcommand{\idptof}[1]{{\mathrm{p}\sub{#1}}}%{{p_{#1}}}
\newcommand{\sym}{\mathrm{s}_\multii}%{s_\multii}
\newcommand{\symTL}{\bar{\mathrm{s}}_\multii}%{\check{\mathrm{s}}_\multii}%{s_\multii}

\newcommand{\Rows}{\mathfrak{R}^\lambda}%{\textnormal{R}^\lambda}
\newcommand{\Columns}{\mathfrak{C}^\lambda}%{\textnormal{C}^\lambda}

\newcommand{\End}{\textnormal{\upshape End}}

\global\long\def\SolSp{\mathcal{S}}%{\mathcal{C}}

\global\long\def\PartF{\mathcal{Z}}
\global\long\def\CobloF{\mathcal{U}}
\global\long\def\chamber{\mathfrak{X}}

\global\long\def\Schur{S}

\newcommand{\re}{\textnormal{\upshape Re}}
\newcommand{\im}{\textnormal{\upshape Im}}

\global\long\def\SLE{\textrm{SLE}}

\global\long\def\PR{\mathbb{P}}

%Derivatives notation
\global\long\def\ud{\mathrm{d}}
\global\long\def\der#1{\frac{\ud}{\ud#1}}
\global\long\def\pder#1{\frac{\partial}{\partial#1}}
\global\long\def\pdder#1{\frac{\partial^{2}}{\partial#1^{2}}}
\global\long\def\pddder#1{\frac{\partial^{3}}{\partial#1^{3}}}

\newcommand{\rainbow}[1]{{\star({#1})}} % to be replaced with a better rainbow!

\newcommand{\sign}{\mathrm{sgn}}

\newcommand{\GreenK}{\mathsf{G}}
\newcommand{\Walks}{\mathscr{W}}

\newcommand{\GreenKH}{\mathscr{G}}
\newcommand{\PoissonKH}{\mathscr{P}}
\newcommand{\ExcKH}{\mathscr{K}}

\newcommand{\DPleq}{\preceq} % partial order of DPs
\newcommand{\DPgeq}{\succeq} % partial order of DPs
\newcommand{\CItilingsof}{\mathcal{C}}

\global\long\def\Mob{\phi}
\global\long\def\domain{\Lambda}
%\global\long\def\linkpatt{\omega}

\global\long\def\summ{q}%{p}%{\sigma}
\global\long\def\mult{\mathrm{m}}
\newcommand{\conn}{\vartheta_{\mathrm{UST}}}

\newcommand{\event}{\mathrm{Conn}}

\newcommand{\Vir}{\mathrm{Vir}}

\global\long\def\link#1#2{\raisebox{.5ex}{\hspace{-1mm}
\scalebox{.75}{$\linkInEquation{\boldsymbol{#1}}{\boldsymbol{#2}}$}}}

\global\long\def\defect#1{\raisebox{.5ex}{\hspace{-1mm}
\scalebox{.75}{$\defectInEquation{\boldsymbol{#1}}$}}}

\global\long\def\quote#1{$``${{#1}}$"$}

\global\long\def\bs{\boldsymbol}

\newcommand{\super}[1]{^{\scaleobj{0.85}{(#1)}}}
\newcommand{\sub}[1]{_{\scaleobj{0.85}{(#1)}}}
\newcommand{\superscr}[1]{^{\scaleobj{0.85}{#1}}}
\newcommand{\subscr}[1]{_{\scaleobj{0.85}{#1}}}

\newcommand{\vol}{\mathrm{vol}}
\newcommand{\lfunct}{S_L^0}

%%%%%%%%%%%%%%%%%%%%%%%%%%%

\title{Fused Specht Polynomials and $c=1$ Degenerate Conformal Blocks}

\vspace{2.5cm}

\begin{center}
{%\LARGE 
\huge
\bf \scshape{
Fused Specht Polynomials \\[.5em] and $c=1$ Degenerate Conformal Blocks
}}
\end{center}

\vspace{0.75cm}

\begin{center}
{\Large \scshape Augustin Lafay}{\footnotesize\footnotemark[1]} \qquad
{\Large \scshape Eveliina Peltola}{\footnotesize\footnotemark[1]\textsuperscript{\&}\footnotemark[2]} \qquad  
{\Large \scshape Julien Roussillon}{\footnotesize\footnotemark[1]} \\
\, {\footnotesize{\protect\url{augustin.lafay@aalto.fi}}} \quad \;
{\footnotesize{\protect\url{eveliina.peltola@aalto.fi}}} \quad \,
{\footnotesize{\protect\url{julien.roussillon@aalto.fi}}} 
\end{center}

\footnotetext[1]{Department of Mathematics and Systems Analysis, 
P.O. Box 11100, FI-00076, Aalto University, Finland.}
\footnotetext[2]{Institute for Applied Mathematics, University of Bonn, Endenicher Allee 60, D-53115 Bonn, Germany.}
	
\setcounter{footnote}{0}

\vspace{0.75cm}

\begin{center}
\begin{minipage}{0.85\textwidth} 
{\scshape Abstract.}
We introduce a class of polynomials that we call fused Specht polynomials and use them to characterize 
irreducible representations of the fused Hecke algebra with parameter $q=-1$ in the space of polynomials. 
We apply the fused Specht polynomials to construct a basis for a space of holomorphic (chiral) conformal blocks 
with central charge $c=1$ which are degenerate at each point. 
In conformal field theory, this corresponds to all primary fields having conformal weight in the Kac table. 
The associated correlation functions 
are expected to give rise to conformally invariant boundary conditions for the Gaussian free field, which has also been verified in special cases.  
\end{minipage}
\end{center}

\vspace{0.65cm}
%\vspace{0.75cm}
%\newpage

{\hypersetup{linkcolor=black}
\setcounter{tocdepth}{2}
\tableofcontents}

\newpage

\bigskip{}
\section{Introduction}
This article is essentially divided into two parts, each of which is of independent interest. 
The first part is combinatorial and only assumes basic background in representation theory.
It concerns irreducible representations of the \emph{fused} (or valenced) \emph{Hecke algebra}, 
whose building blocks are fused generalizations of the classical Specht polynomials. 
The second part concerns \emph{degenerate conformal blocks} in a $c=1$ conformal field theory (CFT), which we explicitly build from the fused Specht polynomials. 
(For readers interested in CFT or random geometry, 
the algebraic results from the first part can be taken as a black box.)

We begin with motivation for our results from topology/representation theory point of view on the one hand, 
and from CFT/random geometry point of view on the other hand.

\smallskip

The representation theory of the symmetric group $\SymGrp_\Summed$ is a very classical subject (initiated by Frobenius, Schur, Young, and Specht), 
with ubiquitous applications to various areas in mathematics and physics. 
It was observed in the 1930s 
that the combinatorics of \emph{Young tableaux} plays a prominent role in the classification of irreducible representations of $\SymGrp_\Summed$~\cite{Young:On_quantitative_substitutional_analysis_IV, Fulton-Harris:Representation_theory}. 
A particular class of those, yielding a complete set of irreducible representations, is termed \emph{Specht modules} and spanned by \emph{Specht polynomials}~\cite{Specht:Die_irreduziblen_Darstellungen_der_symmetrischen_Gruppe, Peel:Specht_modules_and_symmetric_groups}.

One of the basic questions in representation theory is the decomposition of a given representation into irreducible components. 
A structurally beautiful result
(termed Schur-Weyl duality)  
relates the representation theory of the symmetric group to that of the special linear group $\mathrm{SL}(2,\bC)$ and its Lie algebra $\mathfrak{sl}(2,\bC)$~\cite{Schur:SW_duality, Weyl:Classical_groups, Fulton-Harris:Representation_theory}.
It concerns a tensor product 
of defining representations $\bC^2$ of $\mathfrak{sl}(2,\bC)$, and implies in particular that the centralizer algebra of $\mathfrak{sl}(2,\bC)$ on $(\bC^2)^{\otimes n}$ equals a quotient of the symmetric group algebra $\bC[\SymGrp_\Summed]$.  
In the case of tensor products of higher-dimensional representations, 
one encounters \emph{fused} (or valenced) versions of the symmetric group algebra $\bC[\SymGrp_\Summed]$. 
More precisely, the centralizer algebra of $\mathfrak{sl}(2,\bC)$ on its tensor product representation $\bC^{s_1+1} \otimes \cdots \otimes \bC^{s_\np+1}$, where $\multii=(s_1,\ldots,s_\np)$ encode the \emph{valences} of the representation, 
is isomorphic to a specific quotient of the \emph{fused Hecke algebra} 
(viz.~the algebra of \quote{fused permutations})~\cite{Flores-Peltola:Higher_spin_QSW,Crampe-Poulain-d-Andecy:Fused_braids_and_centralisers_of_tensor_representations_of_Uq_gln}.
This quotient is also known as the \emph{valenced Temperley-Lieb algebra}~\cite{Temperley-Lieb:Relations_between_percolation_and_colouring_problem, Flores-Peltola:Standard_modules_radicals_and_the_valenced_TL_algebra, Flores-Peltola:Higher_spin_QSW}.

In topology, the Temperley-Lieb algebra can be used to construct the Jones polynomial of a link~\cite{Jones:Polynomial_invariant_for_knots_via_von_Neumann_algebras}, and its valenced version the \quote{colored} Jones polynomial~\cite{Kashaev:Link_invariant_from_quantum_dilogarithm,Kashaev:The_hyperbolic_volume_of_knots_from_the_quantum_dilogarithm,Murakami-Murakami:The_colored_Jones_polynomials_and_the_simplicial_volume_of_a_knot}.
Hecke algebras can be used to construct further generalizations, such as the HOMFLY-PT polynomial~\cite{HOMFLY:New_polynomial_invariant_of_knots_and_links, Przytycki-Traczyk:Conway_algebras_and_skein_equivalence_of_links}. 
In applications to mathematical physics, one can build solutions of the Yang-Baxter equation from the Hecke algebra, 
which is intimately related to quantum groups (or quasitriangular Hopf algebras). 
We will not need to discuss the Yang-Baxter equation in the present work.
Let us briefly mention, however, that \quote{quantum} variants of the Schur-Weyl duality 
relate representations of quantum groups $U_q(\mathfrak{sl}(2,\bC))$ 
to representations of (quotients of) the Hecke algebra $\Hecke_\Summed(q)$, where $q \in \bC \setminus \{0\}$ 
is a deformation parameter~\cite{Jimbo:q_analog_of_UqglN_Hecke_algebra_and_YBE, Dipper-James:Q_Schur_algebra, Martin:On_Schur-Weyl_duality_An_Hecke_algebras_and_quantum_slN_on_CN_tensor_nplus1} and $\Summed \in \bZpos$. 
In the present article, we shall be concerned with the case of $q = -1$ 
(analogous to the classical case of $q = 1$)\footnote{The Hecke algebra $\Hecke_\Summed(\pm 1)$ is isomorphic to the group algebra $\bC[\SymGrp_\Summed]$ of the symmetric group,
and $U_{\pm 1}(\mathfrak{sl}(2,\bC))$ is understood as just the classical universal enveloping algebra $U(\mathfrak{sl}(2,\bC))$.
The quantum groups come up in the case where the deformation parameter is $q \in \bC \setminus \{0, \pm1\}$.
Nevertheless, because (motivated by CFT) we will speak of \quote{fusion} in the present work, 
which also has a direct analogue in the $q$-deformed case, 
we shall adopt the terminology of \quote{(fused) Hecke algebra} (or \quote{(valenced) Hecke algebra}) and the \quote{(valenced) Temperley-Lieb algebra} when we discuss the representation theory of the case of $q=-1$ as well.}. 
We will build irreducible representations of the fused Hecke algebra $\Hecke_\multii := \Hecke_\multii(-1)$ with $q=-1$,  
by introducing a class of polynomials that we call \emph{fused Specht polynomials} (Theorem~\ref{thm:theoremA}).

\medskip

\emph{Conformal field theory} %(CFT) 
has become a rich and important field of study in the mathematical physics community in the recent decades, 
both because of its relation with critical lattice models in statistical physics and random geometry 
(see~\cite{DMS:CFT, Smirnov:Towards_conformal_invariance_of_2D_lattice_models,Peltola:Towards_CFT_for_SLEs, GKR:Compactified_imaginary_Liouville} and references therein), 
and for its intricate connections to algebraic geometry and supersymmetric gauge theories
(see~\cite{AGT:Liouville_correlation_functions_from_4D_gauge_theories, Nekrasov-Shatashvili:Quantization_of_integrable_systems_and_four_dimensional_gauge_theories, Teschner:Quantization_of_moduli_spaces_etc} and references therein). 
In certain CFTs, combinatorial methods and special functions play an important role (cf.~\cite{AFLT:On_combinatorial_expansion_of_the_conformal_blocks_arising_from_AGT_conjecture, Bershtein-Foda:AGT_Burge_pairs_and_minimal_models, ILT:Isomonodromic_tau-functions_from_Liouville_conformal_blocks}), 
as will also be the case in the present work. 
Indeed, we shall find new expressions for conformal blocks in a CFT with central charge $c=1$ in terms of special functions, 
building on the aforementioned (fused) Specht polynomials (cf.~Theorem~\ref{thm:theoremBSA}).

\emph{Conformal blocks} 
provide fundamental building blocks of correlation functions of a CFT.  
In two dimensions, the conformal symmetry imposes infinitely many constraints to the system (encoded into representations of the Virasoro algebra)~\cite{BPZ:Infinite_conformal_symmetry_of_critical_fluctuations_in_2D,DMS:CFT} 
and thereby the structure of the correlation functions is believed to be completely determined by the two- and three-point functions together with the \emph{fusion rules} (or \quote{spectrum}), 
which describe the asymptotics of the correlation functions, and with the \emph{central charge} $c$, a parameter encoding the \quote{conformal anomaly.} 
In this approach, often termed \quote{conformal bootstrap,} or BPZ's algebraic approach, 
it is in principle sufficient to understand the correlation functions of the \emph{primary fields} and the underlying Virasoro algebra representation 
--- the former correspond to highest-weight vectors in Virasoro highest-weight modules, 
and the latter then yields the algebraic structure of the rest of the theory. 
Moreover, in applications one in fact most frequently encounters precisely the correlation functions of primary fields. 
In this article, we shall focus on correlation functions of primary fields in a certain $c=1$ CFT, comprising so-called \quote{degenerate fields,} 
relevant to random geometry applications.  

Upon expanding the correlation functions in terms of a Frobenius type expansion (operator product expansion (OPE) determined by the fusion rules), 
choices of different intermediate Virasoro modules yield different correlation functions.
Particular choices are expected to give distinguished bases of correlation functions (thus singled out by their OPEs), 
and all correlation functions then to be expanded in such bases. 
Certain distinguished bases of correlation functions have been related to geometric observables in scaling limits of critical lattice models: 
solving \emph{crossing probabilities}  (cf.~\cite{Cardy:Critical_percolation_in_finite_geometries,Smirnov:Critical_percolation_in_the_plane,FSKZ:A_formula_for_crossing_probabilities_of_critical_systems_inside_polygons,Peltola-Wu:Crossing_probabilities_of_multiple_Ising_interfaces}), 
or describing \emph{boundary condition changing operators}
(cf.~\cite{Cardy:Conformal_invariance_and_surface_critical_behavior,FSKZ:A_formula_for_crossing_probabilities_of_critical_systems_inside_polygons,Peltola-Wu:Global_and_local_multiple_SLEs_and_connection_probabilities_for_level_lines_of_GFF,FPW:Connection_probabilities_of_multiple_FK_Ising_interfaces}),
also related to \emph{Schramm-Loewner evolution curves}, SLE$(\kappa)$ 
(cf.~\cite{BBK:Multiple_SLEs_and_statistical_mechanics_martingales,Dubedat:Euler_integrals_for_commuting_SLEs,Kytola-Peltola:Pure_partition_functions_of_multiple_SLEs,Peltola:Basis_for_solutions_of_BSA_PDEs_with_particular_asymptotic_properties}. 
In that context, the OPE structure also admits a probabilistic meaning in the corresponding model, 
and is crucial in deriving rigorous scaling limit results (see~\cite{Peltola:Towards_CFT_for_SLEs} for a survey).

The correlation functions of primary fields are expected to be conformally covariant functions, 
and their behavior under conformal transformations is entirely characterized by their conformal weights.
Interestingly enough, a special class of primary fields called \emph{degenerate fields} often appear in applications to boundary effects in statistical physics models (as in the aforementioned references).
Their correlation functions 
should furthermore satisfy certain linear homogeneous partial differential equations, \emph{BPZ PDEs},
which emerge from the fact that Virasoro Verma modules corresponding to degenerate fields contain singular vectors, i.e., 
vectors which generate a nontrivial submodule~\cite{BPZ:Infinite_conformal_symmetry_in_2D_QFT}. 
Fe{\u\i}gin~\&~Fuchs classified all such modules~\cite{Feigin-Fuchs:Verma_modules_over_Virasoro_book, Iohara-Koga:Representation_theory_of_Virasoro}, 
yielding a two-parameter family of relevant conformal weights. 
It is conventional to parameterize them as $h_{r,t}(\theta)$ in terms of $r,t \in \bZpos$, and $\theta \in \bC \setminus \{0\}$:
\begin{align*}
h_{r,t}(\theta) := 
\frac{(r^2-1)}{4} \, \theta + \frac{(t^2-1)}{4} \, \theta^{-1} 
+ \frac{(1-rt)}{2} 
\qquad \textnormal{and} \qquad
c(\theta) = 13 - 6( \theta + \theta^{-1} ) 
\end{align*}
(this is also called the \quote{Kac table}~\cite{Kac:Highest_weight_representations_of_infinite_dimensional_Lie_algebras,
Schottenloher:Mathematical_introduction_to_CFT}).

For SLE$(\kappa)$ applications, one takes $\theta = \kappa / 4$, in which case 
$c = \frac{(3\kappa-8)(6-\kappa)}{2\kappa}$ and $h_{1,2}  = \frac{6-\kappa}{2\kappa}$, for example. 
Note that $c = 1$ if and only if $\kappa=4$, and in this case, we have 
\begin{align} \label{eq:conf_weights0}
h_{r,t} = \frac{(t-r)^2}{4} = h_{t,r} = h_{1,|t-r|+1} ,
\qquad r,t \in \bZpos ,
\end{align}
so it then suffices to consider the collection (indexed by $s=t-1$ for convenience)
\begin{align}\label{eq:conf_weights}
\{ h_{1,s+1} \; | \; s \in \bZnn\} 
= \big\{  \tfrac{s^2}{4} \; | \; s \in \bZnn \big\} 
= \big\{ 0,\tfrac{1}{4},1,\tfrac{9}{4},4,\tfrac{25}{4},9,\tfrac{49}{4},16,\tfrac{81}{4} , \ldots \big\}.
\end{align}

In Sections~\ref{sec:section3}-\ref{sec:section4}, 
we construct a basis for a space of conformal blocks in a CFT with central charge $c=1$ and conformal weights in the Kac table~\eqref{eq:conf_weights}.
We prove that the associated correlation functions 
are linearly independent (Proposition~\ref{prop:basisformathcalC})
and span a solution space of a special class of BPZ PDEs, 
also known as \quote{Beno{\^i}t~\&~Saint-Aubin equations}~\cite{BSA:Degenerate_CFTs_and_explicit_expressions_for_some_null_vectors} (Theorem~\ref{thm:theoremBSA}).
Such conformal blocks are expected to give rise to a family of conformally invariant boundary conditions for 
the Gaussian free field (GFF)\footnote{The GFF also describes the scaling limit of the height function of the double-dimer model~\cite{Kenyon:Dominos_and_the_Gaussian_free_field},
and certain correlation functions in the $c=1$ CFT considered in the present article give formulas for connection probabilities in this model~\cite{Kenyon-Wilson:Double_dimer_pairings_and_skew_Young_diagrams,Peltola-Wu:Global_and_local_multiple_SLEs_and_connection_probabilities_for_level_lines_of_GFF}.
See also the recent~\cite{Lafay-Le-Roussillon:Degenerate_conformal_blocks_for_the_W3_algebra_and_Specht_polynomials} for the case of triple-dimers.},  
which can also be verified in special cases~\cite{Peltola-Wu:Global_and_local_multiple_SLEs_and_connection_probabilities_for_level_lines_of_GFF,Liu-Wu:Scaling_limits_of_crossing_probabilities_in_metric_graph_GFF}. 
We also plan to return to this in future work.

Interestingly (and surprisingly to us), the conformal block basis which we introduce in the present work
(and which plays an important role in applications to statistical physics and random geometry) 
does \emph{not} correspond to the so-called \quote{comb basis,} which is often used especially in the physics 
literature~\cite{DMS:CFT,KKP:Conformal_blocks_q_combinatorics_and_quantum_group_symmetry}. 
(We provide a counterexample in Remark~\ref{rem:remarkcomb} via asymptotics of a certain basis element.) 
The comb basis should arise instead as a limit $c \nearrow 1$ of the conformal block basis defined 
in~\cite{KKP:Conformal_blocks_q_combinatorics_and_quantum_group_symmetry} for irrational central charges,
and a valenced/fused generalization thereof (analogous to but different as in~\cite{Peltola:Basis_for_solutions_of_BSA_PDEs_with_particular_asymptotic_properties}).
Alternatively, the comb basis can be constructed from our basis.

In~\cite{KLPR:Planar_UST_branches_and_degenerate_boundary_correlations}, with A.~Karrila 
we consider analogous functions 
for a CFT with central charge $c=-2$, describing the scaling limit of boundary-touching 
branches in a uniform spanning tree model.
In particular, the explicit determinantal functions discussed in~\cite[Thm.~B.1]{KLPR:Planar_UST_branches_and_degenerate_boundary_correlations}
are the $c=-2$ (and $\kappa = 2$) analogues of the conformal block basis functions considered in the present work 
(having $c=1$ and $\kappa = 4$).  
A special case of these are the so-called \quote{Fomin determinants} 
(see~\cite{Fomin:LERW_and_total_positivity} and~\cite[Sect.~3.4]{KKP:Boundary_correlations_in_planar_LERW_and_UST})
which come up as partition functions for non-intersecting random walks (loop-erased walks). 
 
\smallskip

\textbf{Short description of our results}.
Throughout, we fix \emph{valences} $\multii=(s_1,\ldots,s_\np)$, 
where $s_i \in \bZpos$ for all $i \in \{1,\ldots,\np\}$, and such that $s_1+\cdots+s_\np = \Summed$.
(These are called \quote{integer compositions} of $\Summed$ in combinatorics literature.)  
The symmetric group $\SymGrp_\Summed$ acts naturally on $\{1,2,\ldots,\Summed\}$ by permutation, 
and roughly, the composition $\multii$ represents tuples of indices that should be stable under this action,
yielding variants of the symmetric group.

Let $\bC[\SymGrp_\Summed]$ be the symmetric group algebra.
The \quote{colored symmetric group} 
$\SymGrp_{s_1} \times \cdots \times \SymGrp_{s_\np}$
is a subgroup of $\SymGrp_{\Summed}$ giving rise to the $\multii$-\emph{antisymmetrizer} idempotent $\idpt$ defined in~\eqref{eq:idempotent}, 
obtained by antisymmetrizing groups of consecutive letters according to the valences $\multii$. 
By the idempotent property $\idpt^2 = \idpt$, the following conjugated set is an associative algebra with unit $\idpt$,
termed the \emph{fused Hecke algebra}~\cite{Crampe-Poulain-d-Andecy:Fused_braids_and_centralisers_of_tensor_representations_of_Uq_gln} (with deformation parameter $q=-1$):
\begin{align} \label{eq:Hecke}
\Hecke_\multii := \idpt \bC[\SymGrp_\Summed] \idpt = \{\idpt \, a \, \idpt \;|\; a \in \bC[\SymGrp_\Summed]\} .
\end{align}
In Section~\ref{sec:sectionhecke}, we investigate irreducible representations of $\Hecke_\multii$ in the space of polynomials. 
In fact, $\Hecke_\multii$ is a semisimple algebra and its simple modules\footnote{Recall that a \emph{simple module} is a nonzero vector space $V$ carrying an irreducible representation, i.e., such that $V$ does not have any nontrivial submodules (subspaces other than $\{0\}$ and $V$ carrying a subrepresentation).}
can be expressed in terms of Young diagrams satisfying certain properties 
(see Theorem~\ref{thm:proppVlambda} and~\cite[Thm.~6.5]{Crampe-Poulain-d-Andecy:Fused_braids_and_centralisers_of_tensor_representations_of_Uq_gln}).

Recall that irreducible representations of $\SymGrp_\Summed$ in the space of polynomials 
can be described in terms of \emph{Specht polynomials} \cite{Specht:Die_irreduziblen_Darstellungen_der_symmetrischen_Gruppe, 
Peel:Specht_modules_and_symmetric_groups}. 
They are labeled by \emph{standard} Young tableaux and are given by products of Vandermonde determinants. 
One of our main contributions of Section~\ref{sec:sectionhecke} is to introduce a class of polynomials 
labeled by \emph{semi-standard} Young tableaux that we call \emph{fused Specht polynomials}, 
which we define as certain limits of linear combinations of Specht polynomials (up to a normalization factor) 
motivated by fusion in CFT for applications in both CFT and in statistical physics 
--- see Definition~\ref{def:fusedspecht}.
We also present an explicit formula for the fused Specht polynomials in Proposition~\ref{prop:combinatorialformula}.

The main result of Section~\ref{sec:sectionhecke} is Theorem~\ref{thm:theoremA}, 
which pertains to a characterization of the irreducible representations of $\Hecke_\multii$ in terms of the fused Specht polynomials. 
Our proof of Theorem~\ref{thm:theoremA} relies on a combinatorial argument (Lemma~\ref{lem:lemmainj1}) 
and is valid only for Young diagrams with two columns 
(which is sufficient for our applications)
--- however, we believe that the claim extends to Young diagrams of any shape (Conjecture~\ref{conj:theoremA}).

Sections~\ref{sec:section3} and~\ref{sec:section4} constitute the second part of this article. 
The central object of interest is a certain space $\SolSp_\multii$ of functions. 
Any element in $\SolSp_\multii$ satisfies, in particular, a system of $\np$ BPZ type 
(in this case, Beno{\^i}t~\&~Saint-Aubin, BSA) partial differential equations
with $c=1$, and a certain covariance property under M\"obius transformations. 
In other words, functions in the space $\SolSp_\multii$ can be regarded as correlation functions in a $c=1$ CFT
with degenerate fields of weights in the Kac table~(\ref{eq:conf_weights0},~\ref{eq:conf_weights}) 
$(h_{1,s_1+1}, \ldots , h_{1,s_\np+1})$, labeled by the valences $\multii$.

The simplest nontrivial case occurs when $\multii=(1,\ldots,1)$. 
In this case, $\np=2N$ is even and all the PDEs are of second order, 
and a certain important basis for $\SolSp\sub{1,\ldots,1}$ called \emph{conformal block basis} was constructed in~\cite{Peltola-Wu:Global_and_local_multiple_SLEs_and_connection_probabilities_for_level_lines_of_GFF}. 
We revisit this result in Proposition~\ref{prop:lemma3p3} by rewriting the basis elements in terms of Specht polynomials 
associated with standard Young tableaux with two columns. 
We show in Corollary~\ref{cor:TL_rep} that $\SolSp\sub{1,\ldots,1}$ is isomorphic to a standard module (without defects) 
of the Temperley-Lieb algebra\footnote{Here, the loop \quote{fugacity} parameter $\nu := -q - q^{-1} \in \bC$ equals $2$ for $q=-1$.} 
$\TL_{2N} = \TL_{2N}(\nu) = \TL_{2N}(2)$.

The main contribution of Section~\ref{sec:section3} is to extend this to the case of arbitrary $\multii$:
we construct a basis of $\SolSp_\multii$ that we also call \quote{conformal block basis.} 
We show that the basis elements can be written in terms of fused Specht polynomials associated with 
semi-standard Young tableaux with two columns (Proposition~\ref{prop:basisformathcalC}). 
We then show (Proposition~\ref{prop:repvTL}) that $\SolSp_\multii$ is isomorphic to 
a standard module of the valenced Temperley-Lieb algebra~\cite{Flores-Peltola:Standard_modules_radicals_and_the_valenced_TL_algebra,Flores-Peltola_Generators_projectors_and_the_JW_algebra}. 
We also verify the M\"obius covariance property of the conformal block basis elements (Proposition~\ref{prop:wardidentity}), 
state the BPZ equations (Theorem~\ref{thm:theoremBSA}) and outline how we can verify them. 
However, the complete proof of Theorem~\ref{thm:theoremBSA} requires significantly more efforts and is the sole objective of Section~\ref{sec:section4}.

In Section~\ref{sec:section3}, we also show that special cases of our conformal block basis functions 
indeed equal the ones used in applications to the Gaussian free field (GFF).
The special case where $\multii = (1,\ldots,1)$ is the content of~\cite[Sect.~5-6]{Peltola-Wu:Global_and_local_multiple_SLEs_and_connection_probabilities_for_level_lines_of_GFF}, 
where crossing probability formulas for the GFF with alternating boundary data were proven, 
and the case of more general boundary data was pointed out 
(and proven later in~\cite[Thm.~4.1]{Liu-Wu:Scaling_limits_of_crossing_probabilities_in_metric_graph_GFF}). 
The special case where $\multii = (2,\ldots,2)$ was studied by Liu \& Wu~\cite{Liu-Wu:Scaling_limits_of_crossing_probabilities_in_metric_graph_GFF},
who proved crossing probability formulas for the GFF with generalized alternating boundary data.
In particular, they introduced three functions in \cite[Eq.~(5.15,~5.16,~5.17)]{Liu-Wu:Scaling_limits_of_crossing_probabilities_in_metric_graph_GFF}.
We check in Remark~\ref{rem:Dyck path gen2} that these indeed agree with the three elements of 
the conformal block basis of $\SolSp\sub{2,2,2,2}$. 
One could similarly carry out the analysis for the more general blocks with arbitrary $\multii$.
We hence obtain a complete set of conformal block basis functions applicable to crossing events for the GFF.

In Section~\ref{sec:section4}, we turn to the BPZ equations. 
Systematic verification of these equations does not seem amenable via a direct 
computation\footnote{An alternative approach could be provided by generalizing the elementary computation 
performed in~\cite[Sect.~5.2]{KKP:Boundary_correlations_in_planar_LERW_and_UST}, but this seems very complicated in general.}.  
Therefore, we proceed by a recursive approach bootstrapping from the already known case of 2nd order PDEs~\cite[Lem.~6.4]{Peltola-Wu:Global_and_local_multiple_SLEs_and_connection_probabilities_for_level_lines_of_GFF} 
via asymptotics and a combination of tools from algebra and complex geometry.
We follow Dub\'edat's approach~\cite{Dubedat:SLE_and_Virasoro_representations_localization, 
Dubedat:SLE_and_Virasoro_representations_fusion} (which unfortunately only applies with irrational central charges), 
utilizing the underlying Virasoro algebra structure. The proof is rather non-trivial, and we shall explain the strategy 
in more detail in the beginning of Section~\ref{sec:section4}. 
The key new input needed is representation-theoretic:
we extend~\cite[Lem.~1]{Dubedat:SLE_and_Virasoro_representations_fusion} 
to the case of $c=1$, where the Virasoro structure is slightly more intricate (see Lemma~\ref{lem:analoglemma1Dubedat}). 

\medskip

\textbf{Acknowledgments.}

We thank Rick Kenyon, Ian Le, and Hao Wu for inspiring conversations, 
and Alex Karrila for useful comments on the first version of this manuscript. 
We are also very grateful to the anonymous referees for their input, which helped us to greatly improve the manuscript. 

\begin{itemize}[leftmargin=1em] 
\item A.L. is supported by the Academy of Finland grant number 340461 \quote{Conformal invariance in planar random geometry.}

\item J.R. is supported by the Academy of Finland Centre of Excellence Programme grant number 346315 \quote{Finnish centre of excellence in Randomness and STructures} (FiRST).

\item This material is part of a project that has received funding from the  European Research Council (ERC) under the European Union's Horizon 2020 research and innovation programme (101042460): 
ERC Starting grant \quote{Interplay of structures in conformal and universal random geometry} (ISCoURaGe) 
and from the Academy of Finland grant number 340461 \quote{Conformal invariance in planar random geometry}.
E.P.~is also supported by 
the Academy of Finland Centre of Excellence Programme grant number 346315 \quote{Finnish centre of excellence in Randomness and STructures} (FiRST) 
and by the Deutsche Forschungsgemeinschaft (DFG, German Research Foundation) under Germany's Excellence Strategy EXC-2047/1-390685813, 
as well as the DFG collaborative research centre \quote{The mathematics of emerging effects} CRC-1060/211504053.

%E.P. is affiliated at Aalto University, Espoo, Finland and at the University of Bonn, Germany.
\end{itemize}

\bigskip{}
\section{Fused Specht polynomials and the fused Hecke algebra with parameter $q=-1$} 
\label{sec:sectionhecke}
Throughout, we let $\Summed \in \bZpos$ be an integer and $\lambda \vdash \Summed$ a partition of $\Summed$, that is, 
$\lambda=(\lambda_1,\lambda_2,\ldots ,\lambda_l)$ such that $\lambda_1 \geq \lambda_2 \geq\cdots \geq\lambda_l \geq 0$ and $\lambda_1+\lambda_2+\cdots +\lambda_l = \Summed$. 
The size of the partition $\lambda$ is denoted by $|\lambda| = \Summed$. 
Let $\bC[\SymGrp_\Summed]$ be the symmetric group algebra, 
generated by the transpositions $\tau_i = (i,i+1) \in \SymGrp_\Summed$ for $i\in\{1,\ldots ,\Summed-1\} =:\llbracket 1, \Summed-1\rrbracket$ with relations
\begin{equation*}
\begin{aligned}
\tau_i^2 = \; & 1 , && \textnormal{for } i\in\llbracket 1, \Summed-1\rrbracket , \\
\tau_i\tau_{i+1}\tau_i = \; & \tau_{i+1}\tau_i\tau_{i+1} , && \textnormal{for } i\in\llbracket 1, \Summed-2\rrbracket , \\
\tau_i\tau_j = \; & \tau_j\tau_i , && \textnormal{for } |j-i|>1.
\end{aligned}
\end{equation*}
This section is devoted to investigating the irreducible representations of $\bC[\SymGrp_\Summed]$ 
and its special subalgebra, the fused Hecke algebra~\eqref{eq:Hecke}, in the space of polynomials.
In the key Theorem~\ref{thm:theoremA}, we consider irreducible representations in terms of 
the fused Specht polynomials, which we introduce as limiting expressions from 
the classical Specht polynomials (Definition~\ref{def:fusedspecht}). 
One of the key ingredients to prove Theorem~\ref{thm:theoremA} is an explicit combinatorial formula for the fused Specht polynomials, Proposition~\ref{prop:combinatorialformula}, which is of independent interest. 
Theorem~\ref{thm:theoremA}, in turn, shall be used in CFT applications later.

\subsection{Specht polynomials and irreducible modules for the symmetric group}
\label{subsec:Specht}

We begin by fixing terminology. 
A \emph{Young diagram} of shape $\lambda$ is a finite collection of boxes arranged in $l$ left-justified rows with row lengths being, from top to bottom, $\lambda_1,\ldots ,\lambda_l$. 
A \emph{numbering} of a Young diagram is obtained by placing the numbers $1,\ldots ,\Summed$ in the $\Summed$ boxes of the Young diagram. 
A \emph{standard Young tableau} is a numbering which is strictly increasing across each row and down each column. 
The sets of numberings and of standard Young tableaux of shape $\lambda$ will be denoted $\NB$ and $\SYT$, respectively.
Observe that $\SYT \subset \NB$.

The group $\SymGrp_\Summed$ acts on $\NB$ by letter permutations; 
the action of $\sigma \in \SymGrp_\Summed$ on a numbering $N \in \NB$ is denoted $\sigma.N$. 
For $N \in \NB$, let $\Rows(N)$ (resp.~$\Columns(N)$) 
be the subgroup of $\SymGrp_\Summed$ which preserves the set of entries of each of its rows (resp.~columns). 
A \emph{tabloid} $\{N\}$ is an equivalence class of numberings defined by $\{N'\} = \{N\}$ 
if and only if $N'=\sigma.N$ for some $\sigma \in \Rows(N)$. 
The $\bC$-vector space spanned by tabloids of shape $\lambda$,
\begin{align*}
M^\lambda := \Spn_{\bC} \{ \{N\} \; | \; N \in \NB\} ,
\end{align*}
carries a natural $\SymGrp_\Summed$-action denoted by $\sigma.\{N\} := \{\sigma.N\}$. 
\emph{Simple modules} of $\SymGrp_\Summed$ 
(i.e., nontrivial modules for which the representation, is irreducible) 
are subspaces of $M^\lambda$, and can be realized in various ways. 
In what follows, we recall two different but equivalent (well known) realizations --- 
in terms of polytabloids (Section~\ref{subsubsec:polytabloid})
and polynomials (Section~\ref{subsubsec:polynomial}).

\subsubsection{Polytabloid basis}
\label{subsubsec:polytabloid}

For each numbering $N \in \NB$, the \emph{column antisymmetrizer}
\begin{align*} 
\epsilon_N := \sum_{\sigma \in \Columns(N)} \sign(\sigma) \, \sigma 
\end{align*}
defines the associated \emph{polytabloid} $v_N := \epsilon_N.\{N\} = \{\epsilon_N.N\} \in M^\lambda$. 
Note that $\{N\} = \{N'\}$ does not necessarily imply that $v_N$ and $v_{N'}$ would be equal, 
since the actions of row and column permutations (the subgroups $\Rows(N)$ and $\Columns(N)$) do not commute in general.

\begin{lemma}{\cite{Specht:Die_irreduziblen_Darstellungen_der_symmetrischen_Gruppe}} 
\label{lem:polytabloids}
A complete set of pairwise non-isomorphic simple modules of the algebra $\bC[\SymGrp_\Summed]$ is given by 
$\{V^\lambda \; | \; \lambda \vdash \Summed \}$, where 
$V^\lambda \subset M^\lambda$ is the $\bC$-vector space spanned by the polytabloids, 
\begin{align*} 
V^\lambda := \Spn_{\bC} \{ v_N \; | \; N \in \NB \}
= \Spn_{\bC} \{ v_T \; | \; T \in \SYT \} ,
\end{align*}
where the \emph{polytabloid basis} $\{ v_T \; | \; T \in \SYT \}$ is a linearly independent collection. 
\end{lemma}

Note that $\rho_\lambda(\sigma)(v_N):= \sigma.v_N = v_{\sigma.N}$, 
for $\sigma \in \SymGrp_\Summed$ and $N \in \NB$, 
which implies that $(V^\lambda,\rho_\lambda)$ has the structure of a (left) $\SymGrp_\Summed$-module.
Its linear extension then gives a representation $\rho_\lambda \colon \bC[\SymGrp_\Summed] \to \End(V^\lambda)$.
The pair $(V^\lambda,\rho_\lambda)$ is called a \emph{Specht module}~\cite{Specht:Die_irreduziblen_Darstellungen_der_symmetrischen_Gruppe}. 
See~\cite[Chap.~7]{Fulton:Young_tableaux_With_applications_to_representation_theory_and_geometry} for a detailed account on Specht modules and the proof of Lemma~\ref{lem:polytabloids}.

\subsubsection{Polynomial basis} 
\label{subsubsec:polynomial}

Throughout, let $\{x_i \; | \; i \in \bZpos\}$ be a collection of formal variables. 
We write $\bs{x}_{i_1,\ldots ,i_r} := (x_{i_1},\ldots ,x_{i_r})$. 
The \emph{Vandermonde determinant} is the antisymmetric function
\begin{align} \label{eq:Vandermonde determinant}
\Delta(\bs{x}_{i_1,\ldots ,i_r}) := \prod_{1\leq j < k \leq r} (x_{i_j}-x_{i_k}) 
\end{align}
(with the convention that $\Delta(x) = \Delta(\bs{x}_{i_1}) \equiv 1$ for $r=1$).   

\begin{definition}
The \emph{Specht polynomial} associated with $N \in \NB$ is the polynomial
\begin{align} \label{eq:spechtfactorized}
\mathcal P_N = \mathcal P_N(x_{1},\ldots ,x_{\Summed})
:= \prod_c \Delta(\bs{x}_{N_{.,c}}),
\end{align}
where $c$ runs through the columns of $N$ 
and $N_{.,c}$ is the ordered set of entries in the $c$-th column of $N$ listed from bottom to top.
For instance, we have
\begin{align*}
\mathcal P_{\;\ytableausetup{smalltableaux}\ytableaushort{12,34}}
= \; & \Delta(\bs{x}_{3,1}) \Delta(\bs{x}_{4,2})
= (x_3-x_1)(x_4-x_2)
\\
\mathcal P_{\;\ytableausetup{smalltableaux}\ytableaushort{154,63,2}}
= \; & \Delta(\bs{x}_{2,6,1}) \Delta(\bs{x}_{3,5}) \Delta(\bs{x}_{4})
= (x_2-x_6)(x_2-x_1)(x_6-x_1)(x_3-x_5) .
\end{align*}
\end{definition}

The symmetric group $\SymGrp_\Summed$ acts on the polynomial algebra $\bC[x_1,\ldots ,x_\Summed]$ 
by permutation of the variables. 
In fact, Peel showed in~\cite[Thm.~1.1]{Peel:Specht_modules_and_symmetric_groups} that the space
\begin{align} \label{eq:defPlambda}
P^\lambda := \Spn_{\bC}\{\mathcal P_N \; | \; N \in \NB \} 
= \Spn_{\bC}\{\mathcal P_T \; | \; T \in \SYT\}
\end{align}
is a simple $\SymGrp_\Summed$-module with basis $\{\mathcal P_T \; | \; T \in \SYT\}$ consisting of Specht polynomials.

\begin{lemma}\label{lem:isomorphism}
The following map is an isomorphism of simple $\bC[\SymGrp_\Summed]$-modules:
\begin{align*}
\phi \colon \; & V^\lambda \to P^\lambda \\
\; & v_N \mapsto \phi(v_N) := \mathcal P_N.
\end{align*}
\end{lemma}

\begin{proof}[Proof summary]
Consider first the homomorphism $\phi \colon M^\lambda \to \bC[x_1,\ldots ,x_\Summed]$ 
of $\SymGrp_\Summed$-modules defined by the natural extension of $\phi(\{N\}) := m_N$ in terms of the monomials
\begin{align*}
m_N = m_N(x_{1},\ldots ,x_{\Summed}) := \prod_{i=1}^\Summed x_i^{r^N(i)-1},
\end{align*}
where $r^N(i)$ denotes the row number of the entry \quote{$i$} in $N$, counting row numbers from top to bottom. 
For instance, we have 
\begin{align*}
m_{\;\ytableausetup{smalltableaux}\ytableaushort{12,34}} = x_1^0 x_2^0 x_3^1 x_4^1 = x_3x_4 .
\end{align*}  
By~\cite[Thm.~9]{HLV:Polytabloid_bases_and_Specht_polynomials}, 
the Specht polynomial~\eqref{eq:spechtfactorized} equals the image of the polytabloid $v_N$: 
\begin{align} \label{eq:spechtmonomials}
\mathcal P_N = \phi(v_N) := \epsilon_N.m_N 
= \sum_{\sigma \in \Columns(N)} \sign(\sigma) \prod_{i=1}^\Summed x_{\sigma(i)}^{r^N(i)-1} , \qquad N \in \NB .
\end{align}
For instance, we have 
\begin{align*}
\mathcal P_{\;\ytableausetup{smalltableaux}\ytableaushort{12,34}} 
= (x_3-x_1)(x_4-x_2) = x_3x_4 - x_1x_4 - x_3x_2 + x_1x_2 
= \epsilon_{\;\ytableausetup{smalltableaux}\ytableaushort{12,34}} \, m_{\;\ytableausetup{smalltableaux}\ytableaushort{12,34}} .
\end{align*}  
Hence, the restriction of $\phi$ to $V^\lambda$ (cf.~Lemma~\ref{lem:polytabloids}) yields the sought isomorphism.
\end{proof}

For any $N \in \NB$, Equation~\eqref{eq:spechtfactorized} expresses the Specht polynomial $\mathcal P_N$ as a factorized polynomial, whereas Equation~\eqref{eq:spechtmonomials} expresses it as a linear combination of monomials.

\subsection{Irreducible modules for the fused Hecke algebra}
\label{subsec:fused Hecke algebra}

Fix an integer composition $\multii=(s_1,\ldots,s_\np) \in \bZpos^\np$ such that $s_1+\cdots+s_\np = \Summed$ (valences).  
The \quote{colored symmetric group} 
$\SymGrp_{s_1} \times \cdots \times \SymGrp_{s_\np}$
is a subgroup of $\SymGrp_{\Summed}$ giving rise to the $\multii$-\emph{antisymmetrizer} idempotent 
\begin{align}\label{eq:idempotent}
\idpt := \frac{1}{s_1!\cdots s_\np!} \prod_{k=1}^\np \sum_{\sigma \in \SymGrp_{s_k}} \sign(\sigma) \, \sigma 
\; \in \; \bC[\SymGrp_{s_1} \times \cdots \times \SymGrp_{s_\np}] 
\; \subset \; \bC[\SymGrp_\Summed] ,
\end{align}
which is used to define the \emph{fused Hecke algebra}~\cite{Crampe-Poulain-d-Andecy:Fused_braids_and_centralisers_of_tensor_representations_of_Uq_gln} (with deformation parameter $q=-1$),
\begin{align*}
\Hecke_\multii = \Hecke_\multii(-1) := \idpt \bC[\SymGrp_\Summed] \idpt = \{\idpt \, a \, \idpt \;|\; a \in \bC[\SymGrp_\Summed]\} .
\end{align*}
Note that the algebra $\Hecke_\multii$ has unit $\idpt$, so in particular, it is not a unital subalgebra of $\bC[\SymGrp_\Summed]$. 

\subsubsection{Fused Hecke algebras for $q=\pm 1$}

The fused Hecke algebra at $q=1$, also-called the algebra of \emph{fused permutations} in \cite{Crampe-Poulain-d-Andecy:Fused_braids_and_centralisers_of_tensor_representations_of_Uq_gln}, 
is defined as $\Hecke_\multii(1) := \sym \bC[\SymGrp_\Summed] \sym$ with unit $\sym$, where
\begin{align}\label{eq:defsymmetrizer}
\sym := \frac{1}{s_1!\cdots  s_\np!} \prod_{k=1}^\np \sum_{\sigma \in \SymGrp_{s_k}} \sigma 
\; \in \; \bC[\SymGrp_{s_1} \times \cdots \times \SymGrp_{s_\np}] 
\; \subset \; \bC[\SymGrp_\Summed] 
\end{align}
is the $\multii$-\emph{symmetrizer} idempotent. 
The two fused Hecke algebras $\Hecke_\multii = \Hecke_\multii(-1)$ for $q=-1$ and $\Hecke_\multii(1)$ for $q=1$ are related in the following manner --- in particular, they are isomorphic. 
There exists an involutive automorphism $\omega$ of $\bC[\SymGrp_\Summed]$ defined via 
\begin{align} \label{eq:defomega}
\omega \colon \sigma \mapsto \sign(\sigma) \, \sigma , \qquad \sigma \in \SymGrp_\Summed ,
\end{align}
extending linearly to $\bC[\SymGrp_\Summed]$. Since $\omega(\idpt)=\sym$, we see that
$\Hecke_\multii \cong \sym \bC[\SymGrp_\Summed] \sym = \Hecke(1)$, 
where the isomorphism and its inverse are given by
\begin{align*}
\idpt \, a \, \idpt \mapsto \sym \, \omega(a) \, \sym , \qquad 
\sym \, a \, \sym \mapsto \idpt \, \omega(a) \, \idpt, \qquad 
a \in \bC[\SymGrp_\Summed] . 
\end{align*}

\begin{remark}
\label{rem:remarkomega}
Let $\Bar{\lambda}$ denote the transpose of the partition $\lambda$, 
whose columns are given by the rows of $\lambda$, 
and let $(V^\lambda,\rho_\lambda)$ be a Specht module. 
Then, $(V^\lambda,\rho_\lambda\circ \omega)$ yields a $\bC[\SymGrp_\Summed]$-module isomorphic to $V^{\Bar{\lambda}}$,  see~\cite[Chap.~7]{Fulton:Young_tableaux_With_applications_to_representation_theory_and_geometry}. 
In particular, this implies that, as vector spaces, 
\begin{align}\label{eq:isomorphism1}
\idpt (V^\lambda) \cong \sym(V^{\Bar{\lambda}}) .
\end{align}
We emphasize that under this isomorphism, the basis of polytabloids in $V^\lambda$ is \emph{not} 
mapped to the basis of polytabloids in $V^{\Bar{\lambda}}$, 
but instead, to the basis of so-called \quote{dual polytabloids.} 
\end{remark}

\begin{remark} \label{rem:dual polytabloids}
The space of \emph{dual tabloids} is defined as equivalence classes of numberings, 
$\check{M}^\lambda := \Spn_{\bC} \{ [N] \; | \; N \in \NB\}$, 
modulo $[N'] = \sign(\sigma) [N]$ if and only if $N'=\sigma.N$ for some $\sigma \in \Columns(N)$. 
The symmetric group $\SymGrp_\Summed$ acts on $\check{M}^\lambda$ by $\sigma.[N] = [\sigma.N]$.
Alternatively to Lemma~\ref{lem:polytabloids}, simple modules of $\SymGrp_\Summed$ can be characterized as subspaces of $\check{M}^\lambda$ and are constructed as follows. 
For each numbering $N \in \NB$, the \emph{row symmetrizer} 
\begin{align*} 
\check{\epsilon}_N := \sum_{\sigma \in \Rows(N)} \sigma 
\end{align*}
defines the associated \emph{dual polytabloid} 
$\check{v}_N := \check{\epsilon}_N . [N]$.  
Then, we have
\begin{align*}
V^\lambda \simeq \Spn_{\bC}\{\check{v}_N \; | \; N \in \NB\}.
\end{align*}
\end{remark}

The following lemma is proven, e.g., in~\cite[App.~A.1]{Crampe-Poulain-d-Andecy:Fused_braids_and_centralisers_of_tensor_representations_of_Uq_gln}.

\begin{lemma} \label{lem:lemmaPAP} 
Let $A$ be a finite-dimensional semisimple associative algebra and $p \in A$ an idempotent element (i.e.,~$p^2=p$). 
Then, the algebra $pAp$ with unit $p$ is finite-dimensional and semisimple. 
Moreover, if $\{R^\lambda \; | \; \lambda \in I\}$ is a complete set of pairwise non-isomorphic simple $A$-modules, 
then 
\begin{align*}
\{p(R^\lambda) \; | \; \lambda \in I, \; p(R^\lambda) \neq \{0\}\}
\end{align*}
is a complete set of pairwise non-isomorphic simple $pAp$-modules. 
\end{lemma}

Lemma~\ref{lem:lemmaPAP} implies in particular that both the subspaces
\begin{align}
\label{eq:defpVlamba} 
\idpt( V^\lambda) := \Spn_{\bC}\{\idpt.v_T \; | \; T \in \SYT\}, \\ 
\label{eq:defpPlambda} 
\idpt( P^\lambda) := \Spn_{\bC}\{\idpt.\mathcal P_T \; | \; T \in \SYT\}
\end{align}
are either $\{0\}$ or irreducible modules for the algebra $\Hecke_\multii$. 
Theorem~\ref{thm:proppVlambda} and Corollary~\ref{cor:proppPlambda} characterize these as complete sets of pairwise non-isomorphic simple $\Hecke_\multii$-modules.

\subsubsection{Row-strict Young tableaux}

Fix valences $\multii=(s_1,\ldots,s_\np)$. 
For $\lambda \vdash \Summed$, a (Young) \emph{filling} assigns a positive integer to each box of $\lambda$.
Let $\Fill$ be the set of fillings of Young diagrams of shape $\lambda \vdash \Summed$ 
where each number $k$ appears $s_k$ times, for $k \in \llbracket 1, \np\rrbracket$. 
We say that $\multii$ is the \emph{content}, or weight, of a filling in $\Fill$. 
In particular, we have $\NB = \Fillof{1^\Summed}{\lambda} = \Fillof{1,\ldots,1}{\lambda}$.

A \emph{row-strict} Young tableau is a filling whose entries are weakly increasing down each column and strictly increasing along each row. 
Similarly, a \emph{column-strict} Young tableau is a filling whose numbers are weakly increasing along each row and strictly increasing down each column. 
Let $\RSYT$ and $\CSYT$ be the set of row-strict and column-strict Young tableaux of shape $\lambda$ and content $\multii$, respectively. 
The column-strict ones are often called \emph{semistandard}. 
Observe that $|\RSYT| = |\CSYTBar|$, where $\Bar{\lambda}$ is the transpose of the partition $\lambda$.

There is a condition that $\lambda$ and $\multii$ need to satisfy in order for $|\RSYT|$ to be non-zero. 
Namely, let $\multii^\ord$ be the composition $\multii$ rearranged in decreasing order, i.e.,~a \emph{partition}. 
We say that two partitions $\lambda$ and $\mu$ satisfy the \emph{dominance ordering relation} $\lambda \geq \mu$ 
if and only if 
\begin{align*}
\lambda_1 + \cdots  + \lambda_i \geq \mu_1 + \cdots  + \mu_i , \qquad \textnormal{for all $i$} ,
\end{align*}
where we possibly extend the sequences by zeros.

\begin{lemma} \label{lem:kostkacondition}
We have $|\CSYT| \neq 0 \iff \lambda \geq \multii^\ord$ and similarly, 
$|\RSYT| \neq 0 \iff \Bar{\lambda} \geq \multii^\ord$.
\end{lemma}

\begin{proof}
The first statement follows immediately from~\cite[Lem.~6.3]{Crampe-Poulain-d-Andecy:Fused_braids_and_centralisers_of_tensor_representations_of_Uq_gln} and 
the second statement follows immediately from the first one, since $|\RSYT| = |\CSYTBar|$.
\end{proof}

The next result identifies the complete set of irreducible representations of $\Hecke_\multii$. 

\begin{corollary} \label{cor:proppPlambda}
The collection 
$\{ \idpt(P^\lambda) \; | \; \lambda \in I_\multii \}$, 
where $I_\multii:= \{\lambda \vdash \Summed \; | \; \Bar{\lambda} \geq \multii^\ord\}$, 
is a complete set of pairwise non-isomorphic simple $\Hecke_\multii$-modules. 
\end{corollary}

\begin{proof}
Combining Lemmas~\ref{lem:polytabloids},~\ref{lem:isomorphism}~\&~\ref{lem:lemmaPAP}, 
we see that a complete set of pairwise non-isomorphic simple $\Hecke_\multii$-modules is given by 
$\mathcal N := \{ \idpt(P^\lambda) \; | \; \lambda \vdash \Summed, \; \idpt(P^\lambda) \neq \{0\} \}$. 
The claim then follows, since $\dim \idpt( P^\lambda) = |\RSYT| \neq 0$ if and only if $\lambda \in I_\multii$ by Lemma~\ref{lem:kostkacondition}.
\end{proof}

The sum-of-squares formula now yields the dimension of the semisimple\footnote{The fused Hecke algebra $\Hecke_\multii$ is semisimple by Lemma~\ref{lem:lemmaPAP}.} 
algebra $\Hecke_\multii$:  
\begin{align} \label{eq:sum-of-squares}
\dim (\Hecke_\multii) 
\; = \; \sum_{\lambda \vdash \Summed} |\RSYT|^2 
\; = \; \sum_{\lambda \in I_\multii} |\RSYT|^2 .
\end{align}

\subsubsection{The subspaces $\idpt( V^\lambda)$}

We now return to the characterization of $\idpt( V^\lambda)$.

\begin{definition} \label{def:tildeT}
Let $F \in \Fill$ be a filling of shape $\lambda$ with content $\multii$. 
We associate to $F$ a numbering $\tilde{F} \in \NB$ injectively as follows. 
First, we relabel each entry \quote{$k$} of $F$ by 
\begin{align} \label{eq:defqk}
\summ_k := 1+ \sum_{j=1}^{k-1} s_j, \qquad k \in \llbracket 1, \np\rrbracket .
\end{align} 
This gives a new filling $F'$. 
Second, we construct a word $w$ by reading the entries of $F'$ from top to bottom, column by column from left to right; 
we call this \emph{column reading}. 
Third, we construct a new numbering $\tilde{F}$ by relabeling the entry \quote{$l$} of $F'$ by $l+u$, 
where $u$ is the number of times the letter $l$ has previously appeared in $w$. 
This defines $\tilde{F} \in \NB$.
\end{definition}

For example, with $\lambda=(3,3,1)$ and $\multii=(2,1,3,1)$, for $F \in \Fill$ and $T \in \RSYT$, we have
\begin{align*}
& F \; = \; {\ytableausetup{nosmalltableaux}\begin{ytableau} 
1 & 3 & 4 \\
3 & 3 & 1 \\
2
\end{ytableau}} 
\qquad \implies \qquad
F' \; = \;  {\ytableausetup{nosmalltableaux}\begin{ytableau} 
1 & 4 & 7 \\
4 & 4 & 1 \\
3
\end{ytableau}} 
\qquad \implies \qquad 
\tilde{F} \; = \;  {\ytableausetup{nosmalltableaux}\begin{ytableau}
1 & 5 & 7 \\
4 & 6 & 2 \\
3
\end{ytableau}}, 
\end{align*}
and
\begin{align*}
& T \; = \; {\ytableausetup{nosmalltableaux}\begin{ytableau} 
1 & 2 & 3 \\
1 & 3 & 4 \\
3
\end{ytableau}} 
\qquad \implies \qquad 
T' \; = \;  {\ytableausetup{nosmalltableaux}\begin{ytableau} 
1 & 3 & 4 \\
1 & 4 & 7 \\
4
\end{ytableau}} 
\qquad \implies \qquad 
\tilde{T} \; = \;  {\ytableausetup{nosmalltableaux}\begin{ytableau}
1 & 3 & 6 \\
2 & 5 & 7 \\
4
\end{ytableau}}.
\end{align*}

\begin{lemma} \label{lem:tildeT}
If $T \in \RSYT$, then $\tilde{T} \in \SYT$. 
\end{lemma}

\begin{proof}
It is a simple combinatorial exercise to verify from Definition~\ref{def:tildeT} 
that $T$ being row-strict implies that $\tilde{T}$ 
is strictly increasing across each row and down each column. 
\end{proof}

For $T \in \RSYT$, we define the following vector in $V^\lambda$:
\begin{align} \label{eq:defw_t}
w_T := \idpt.v_{\tilde{T}} \; \in \; V^\lambda.
\end{align}
It is, a priori, a linear combination containing polytabloids of tableaux which are not necessarily standard.
Nevertheless, $w_T$ can always be expressed as a linear combination of the basis elements $\{v_S \; | \; S \in \SYT \}$ (Lemma~\ref{lem:polytabloids}); see also Equation~\eqref{eq: upper triangular}. 
For example, %we have 
\begin{align*}
w_{\;\ytableausetup{smalltableaux}\ytableaushort{13,14,2}}
\; = \; \tfrac12 \, v_{\;\ytableausetup{smalltableaux}\ytableaushort{14,25,3}}
- \tfrac12 \,v_{\;\ytableausetup{smalltableaux}\ytableaushort{24,15,3}} 
\; = \; v_{\;\ytableausetup{smalltableaux}\ytableaushort{14,25,3}} , 
\qquad \qquad 
w_{\;\ytableausetup{smalltableaux}\ytableaushort{13,24,3}}
\; = \; \tfrac12 \, v_{\;\ytableausetup{smalltableaux}\ytableaushort{14,25,3}}
- \tfrac12 \,v_{\;\ytableausetup{smalltableaux}\ytableaushort{13,25,4}} .
\end{align*}

\begin{proposition} \label{prop:linind1}
The set $\{w_T \; | \; T \in \RSYT\}$ defined by Equation~\eqref{eq:defw_t} is a basis for $\idpt(V^\lambda)$.
\end{proposition}

\begin{proof}
Denote by $\preceq$ the total order on the set $\NB$ of tableaux given by the lexicographic order 
on the words obtained by column reading. 
Note that for each $\sigma \in \SymGrp_{s_1} \times \cdots \times \SymGrp_{s_\np}$
and $T \in \RSYT$, we have $\sigma.\tilde{T} \succeq \tilde{T}$ (for the tableau as in Lemma~\ref{lem:tildeT}),
with strict inequality when $\sigma$ is not the identity.
Moreover, it follows from Definition~\ref{def:tildeT} that either $\sigma.\tilde{T} \in \SYT$, 
or it becomes standard by permuting numbers within its columns only. 
If $\sigma \in \SymGrp_{s_1} \times \cdots  \times \SymGrp_{s_\np}$ such that $v_{\sigma.\tilde{T}}$ is proportional to $v_{\tilde{T}}$, 
then $\sigma \in \Columns(\tilde{T})$ and $v_{\sigma.\tilde{T}} = \sign(\sigma) v_{\tilde{T}}$. 
Hence, by expanding $w_T$ using the definitions~(\ref{eq:defw_t},~\ref{eq:idempotent}), 
we see that the coefficient of $v_{\sigma.\tilde{T}}$ equals $\sign(\sigma) a$ for some $a>0$. 
Therefore, we find that 
\begin{align} \label{eq: upper triangular}
w_T = c_{\tilde{T}} \, v_{\tilde{T}} + \sum_{\substack{S\in \SYT \\ S \succ \tilde{T}}} c_S \, v_S , 
\qquad c_{\tilde{T}},c_S \in \bR, \quad c_{\tilde{T}} > 0 ,
\end{align}
and in particular, $w_T\neq 0$
since $\{ v_S \; | \; S \in \SYT \}$ is linearly independent (Lemma~\ref{lem:polytabloids}).  
Moreover, since the map sending $T \mapsto \tilde{T}$ from Definition \ref{def:tildeT} is injective, 
we see that each element in $\{w_T\; | \; T \in \RSYT\}$ is obtained by an upper-triangular transformation from 
$\{ v_S \; | \; S \in \SYT \}$. This implies that the set $\{w_T\; | \; T \in \RSYT\}$ is linearly independent.

Lastly, using the isomorphism~\eqref{eq:isomorphism1} and the fact 
(e.g., from~\cite[Thm.~6.5]{Crampe-Poulain-d-Andecy:Fused_braids_and_centralisers_of_tensor_representations_of_Uq_gln}) 
that $\dim(\sym(V^{\Bar{\lambda}})) = |\textnormal{CSYT}^{\Bar \lambda}_\multii|$, we have
$\dim \idpt(V^\lambda) = \dim \sym(V^{\Bar{\lambda}}) = |\textnormal{CSYT}^{\Bar \lambda}_\multii| = |\RSYT|$.
\end{proof}

The next result identifies the complete set of irreducible representations of $\Hecke_\multii$ in terms of polytabloids. 
It essentially follows 
from the proof of~\cite[Thm.~6.5]{Crampe-Poulain-d-Andecy:Fused_braids_and_centralisers_of_tensor_representations_of_Uq_gln}\footnote{The result~\cite[Thm.~6.5]{Crampe-Poulain-d-Andecy:Fused_braids_and_centralisers_of_tensor_representations_of_Uq_gln} 
states in particular that a complete set of pairwise non-isomorphic simple modules of 
the algebra $\Hecke_\multii(1) := \sym \bC[\SymGrp_\Summed] \sym$ of fused permutations 
is given by $\sym(V^\lambda)$ for $\lambda \geq \multii^\ord$.}.

\begin{theorem} \label{thm:proppVlambda}
The collection 
$\{ \idpt(V^\lambda) \; | \; \lambda \in I_\multii \}$, 
where $I_\multii:= \{\lambda \vdash \Summed \; | \; \Bar{\lambda} \geq \multii^\ord\}$, 
is a complete set of pairwise non-isomorphic simple $\Hecke_\multii$-modules. 
\end{theorem}

\begin{proof}
On the one hand, combining Lemmas~\ref{lem:polytabloids}~\&~\ref{lem:lemmaPAP}, 
we see that a complete set of pairwise non-isomorphic simple $\Hecke_\multii$-modules is given by 
$\mathcal M := \{ \idpt(V^\lambda) \; | \; \lambda \vdash \Summed, \; \idpt(V^\lambda) \neq \{0\} \}$. 
On the other hand, we have 
$\idpt(V^\lambda) = \Spn_{\bC} \{w_T \; | \; T \in \RSYT\}$ by Proposition~\ref{prop:linind1},
and Lemma~\ref{lem:kostkacondition} shows that $|\RSYT| \neq 0$ if and only if $\lambda \in I_\multii$.
Hence, $\mathcal M = \{ \idpt(V^\lambda) \; | \; \lambda \in I_\multii \}$. 
\end{proof}

\subsection{Fused Specht polynomials}
\label{subsec:fusedSpecht}

Next, we will show how the $\Hecke_\multii$-modules $\idpt(P^\lambda)$ in~\eqref{eq:defpPlambda} 
can be characterized in terms of fused Specht polynomials (Definition~\ref{def:fusedspecht}~\&~Theorem~\ref{thm:theoremA})
when $\lambda$ is a Young diagram with two columns (we believe that this result also holds in general, but the proof eludes us, see Conjecture~\ref{conj:theoremA}). 
Observe that, by definition, any element of $\idpt(P^\lambda)$ is a totally antisymmetric polynomial with respect to its variables $x_{\summ_k},\ldots ,x_{\summ_{k+1}-1}$ for all $k \in \llbracket 1, \np\rrbracket$ 
(with the indices $q_k$ defined in~\eqref{eq:defqk}). 
Hence, any element of $\idpt(P^\lambda)$ is divisible by a product of Vandermonde determinants. 
This observation leads us to the definition of the fused Specht polynomials. % (Definition~\ref{def:fusedspecht}). 
To facilitate notation, we denote 
\begin{align*}
\mathfrak{D}_\multii := \big\{ (x_1,\ldots, x_\Summed) \in \bC^\Summed \; | \;  
x_{\summ_k}=x_{\summ_k+1}=\cdots=x_{\summ_{k+1}-1}\textnormal{ for all } k \in \llbracket 1, \np\rrbracket \big\}
\; \subset \; \bC^\Summed ,
\end{align*}
and for a function $f \colon U \to \bC$ defined on a domain $U\subset \bC^\Summed$ which can be continuously extended 
to a subset of $\mathfrak{D}_\multii$, we shall write
\begin{align} \label{eq: eval notation}
[f]_\textnormal{eval} \colon \bC^\np \to \bC
\end{align}
for the function obtained from 
$f(x_1,\ldots,x_\Summed)$ by the evaluations of variables (projection) 
$x_{\summ_k}=x_{\summ_k+1}=\cdots=x_{\summ_{k+1}-1}$ for all $k \in \llbracket 1, \np\rrbracket$. 
We abuse notation and denote the variables of both $f$ and $[f]_\textnormal{eval}$ 
by $(x_1,\ldots, x_\Summed) \in \bC^\Summed$ and $(x_1,\ldots ,x_\np) \in \bC^\np$, respectively. 
We define 
\begin{align} \label{eq:defmappsi}
\begin{split}
\psi \colon \; & \idpt. \bC[x_1,\ldots ,x_\Summed] \to \bC[x_1,\ldots ,x_\np] \\
\; & \idpt.f \; \mapsto \; \left[ \frac{\idpt.f}{\prod_{k=1}^\np \prod_{\summ_k\leq i<j<\summ_{k+1}} (x_j-x_i)} \right]_{\textnormal{eval}} .
\end{split}
\end{align}

\begin{definition} \label{def:fusedspecht}
For each $F \in \Fill$, we define the \emph{fused Specht polynomial} $\mathcal F_F \colon \bC^\np \to \bC$ 
as 
\begin{align} \label{eq:deffusedspecht}
\mathcal F_F 
:= \left[ \frac{\idpt.\mathcal P_{\tilde{F}}}{\prod_{k=1}^\np \prod_{\summ_k\leq i<j<\summ_{k+1}} (x_j-x_i)} \right]_{\textnormal{eval}}
= \psi(\idpt.\mathcal P_{\tilde{F}}) ,
\end{align}
where $\tilde{F} \in \NB$ is obtained from $F$ as in Definition~\ref{def:tildeT}.
\end{definition}

\begin{example}
Consider the following $F \in \Fill$ and its associated numbering $\tilde{F} \in \NB$:
\begin{align*}
F = {\ytableausetup{nosmalltableaux}
\begin{ytableau} 
1 & 2 \\
1 & 3 \\
2 & 4
\end{ytableau}}, 
\qquad \qquad \tilde{F} = {\ytableausetup{nosmalltableaux}
\begin{ytableau} 
1 & 4 \\
2 & 5 \\
3 & 6
\end{ytableau}} \; .
\end{align*}
Adopting the abuse of notation to denote variables of both sides as \quote{$x$}, Definition \ref{def:fusedspecht} gives
\begin{align*}
\mathcal F_F(x_1,x_2,x_3,x_4) 
= 
\left[ \frac{\idptof{2,2,1,1} . \mathcal P_{\tilde{F}}(x_1,\ldots,x_6)}{(x_2-x_1)(x_4-x_3)} \right]_\textnormal{eval} 
= 
\left[ \frac{\idptof{2,2,1,1} . \mathcal P_{\tilde{F}}(x_1,\ldots,x_6)}{(x_2-x_1)(x_4-x_3)} \right]_{\substack{x_1, \, x_2 \, \mapsto x_1, \\ x_3,\, x_4 \, \mapsto x_2 \\ x_5 \, \mapsto x_3, \\ x_6 \, \mapsto x_4}},
\end{align*}
where $\idptof{2,2,1,1}$ antisymmetrizes $\mathcal P_{\tilde{F}}$ with respect to $\{x_1,x_2\}$ and $\{x_3,x_4\}$. Since $\mathcal P_{\tilde{F}}$ is already antisymmetric with respect to $\{x_1,x_2\}$, we obtain
\begin{align*}
\mathcal F_F(x_1,x_2,x_3,x_4) = 
\left[\frac{\mathcal P_{\tilde{F}}(x_1,\ldots,x_6) - (x_3 \leftrightarrow x_4)}{2(x_2-x_1)(x_4-x_3)} \right]_{\substack{x_1, \, x_2 \, \mapsto x_1, \\ x_3,\, x_4 \, \mapsto x_2 \\ x_5 \, \mapsto x_3, \\ x_6 \, \mapsto x_4}},
\end{align*}
and a straightforward computation then leads to
\begin{align*}
\mathcal F_F(x_1,x_2,x_3,x_4) =  -\frac{(x_1-x_2) (x_3-x_4)}{2} \, \big(x_1 (2 x_2-x_3-x_4)-x_2 (x_3+x_4)+2 x_3 x_4\big).
\end{align*}
\end{example}

The simplest class of fused Specht polynomials arises when the tableau has one column: 

\begin{proposition} \label{prop:fusedspechtonecolumn}
Fix $\lambda = (1^\Summed)$ and valences $\multii=(s_1,\ldots,s_\np) \in \bZpos^\np$ such that $s_1+\cdots+s_\np = \Summed$. 
Let $T \in \textnormal{RSYT}\super{1^\Summed}_\multii$. Then, we have
\begin{align} \label{eqfusedspechtonecolumn}
\mathcal F_T = \mathcal F_T(x_1,\ldots ,x_\np) = \prod_{1\leq i < j \leq \np} (x_j-x_i)^{s_i s_j}.
\end{align}
\end{proposition}

\begin{proof}
By Definition~\ref{def:fusedspecht} and using the fact that the Specht polynomial for a standard Young tableau with one column is the Vandermonde determinant, we have
\begin{align*}
\mathcal F_T 
\; = \; \left[\frac{\prod_{1 \leq i < j \leq \Summed} (x_j-x_i)}{\prod_{k=1}^d \prod_{\summ_k \leq i < j < \summ_{k+1}} (x_j-x_i)}\right]_\textnormal{eval} 
= \; \left[ \prod_{1 \leq i < j \leq d} \prod_{l=0}^{s_j-1} \prod_{m=0}^{s_i-1} (x_{\summ_j+l}-x_{\summ_i+m}) \right]_\textnormal{eval}.
\end{align*}  
The evaluation of this leads to~\eqref{eqfusedspechtonecolumn}.
\end{proof}

We now state the main theorem of this section, which gives an isomorphism of 
the two $\Hecke_\multii$-modules 
$\idpt( P^\lambda) := \Spn_{\bC}\{\idpt.\mathcal P_S \; | \; S \in \SYT\}$ in~\eqref{eq:defpPlambda} 
and $\Spn_{\bC}\{\mathcal F_F \; | \; F \in \Fill\}$ defined via~\eqref{eq:deffusedspecht}, 
in the case where $\lambda$ is a Young diagram with two columns. 

\begin{theorem} \label{thm:theoremA}
Let $\lambda$ be a Young diagram with two columns. 
The map $\psi$ in~\eqref{eq:defmappsi} defines 
an isomorphism of $\Hecke_\multii$-modules from $\idpt( P^\lambda)$ to 
$\Spn_{\bC}\{\mathcal F_F \; | \; F \in \Fill\} = \Spn_{\bC} \{\mathcal F_T \; | \; T \in \RSYT\}$,
where the latter space 
obtains an $\Hecke_\multii$-action as
\begin{align*}
(\idpt \, a \, \idpt) . \mathcal F_T 
= \left[ \frac{(\idpt \, a \, \idpt) .  \mathcal P_{\tilde{T}}}{\prod_{k=1}^\np \prod_{\summ_k\leq i<j<\summ_{k+1}} (x_j-x_i)} \right]_{\textnormal{eval}} , \qquad \idpt \, a \, \idpt \in \Hecke_\multii ,
\end{align*}
for $T \in \RSYT$ and $\mathcal F_T$ being the fused Specht polynomial from Definition~\ref{def:fusedspecht},
\begin{align*} 
\mathcal F_T 
:= \left[ \frac{\idpt.\mathcal P_{\tilde{T}}}{\prod_{k=1}^\np \prod_{\summ_k\leq i<j<\summ_{k+1}} (x_j-x_i)} \right]_{\textnormal{eval}}
= \psi(\idpt.\mathcal P_{\tilde{T}})  .
\end{align*}
\end{theorem}

\begin{proof}
The key will be to prove that the set $\{\mathcal F_T \; | \; T \in \RSYT\}$ is linearly independent, 
when $\lambda$ has exactly two columns (Proposition~\ref{prop:proplinearindepfusedspecht}).
Given this, we can finish the proof as follows.

On the one hand, because the map $\psi$ is a surjection onto $\Spn_{\bC}\{\mathcal F_F \; | \; F \in \Fill\}$ 
from the linear span of $\idpt.\mathcal P_{\tilde{F}}$, where $\tilde{F} \in \NB$ is obtained from $F \in \Fill$ as in Definition~\ref{def:tildeT}, 
and the space $\idpt( P^\lambda)$ defined by~(\ref{eq:defPlambda},~\ref{eq:defpPlambda}) 
is either $\{0\}$ or an irreducible $\Hecke_\multii$-module (by Lemma~\ref{lem:lemmaPAP}),
we obtain from the sum-of-squares formula (SOS)~\eqref{eq:sum-of-squares} that
\begin{align*}
\sum_{\lambda \vdash \Summed} ( \dim (\Spn_{\bC}\{\mathcal F_F \; | \; F \in \Fill\}) )^2 
\; \leq \; \sum_{\lambda \vdash \Summed} ( \dim (\idpt( P^\lambda)) )^2 
\; = \; \dim (\Hecke_\multii) .
\end{align*}
On the other hand, since the linearly independent collection $\{\mathcal F_T \; | \; T \in \RSYT\}$ spans a subset of 
$\Spn_{\bC}\{\mathcal F_F \; | \; F \in \Fill\}$ of dimension $|\RSYT|$, 
the SOS~\eqref{eq:sum-of-squares} also gives
\begin{align*}
\dim (\Hecke_\multii) 
\; = \; \sum_{\lambda \vdash \Summed} |\RSYT|^2 
\; \leq \; \sum_{\lambda \vdash \Summed} ( \dim (\Spn_{\bC}\{\mathcal F_F \; | \; F \in \Fill\}) )^2 .
\end{align*}
Combining these facts together, we conclude that
\begin{align*}
\dim (\idpt( P^\lambda)) = \dim (\Spn_{\bC}\{\mathcal F_F \; | \; F \in \Fill\}) = |\RSYT| ,
\end{align*}
so $\Spn_{\bC}\{\mathcal F_F \; | \; F \in \Fill\} = \Spn_{\bC} \{\mathcal F_T \; | \; T \in \RSYT\}$ 
and $\psi$ defines a linear isomorphism from $\idpt( P^\lambda)$ onto this space. 
In particular, it induces an isomorphism of $\Hecke_\multii$-modules. 
\end{proof}

It thus remains to prove that the set $\{\mathcal F_T \; | \; T \in \RSYT\}$ is linearly independent (Proposition~\ref{prop:proplinearindepfusedspecht}). 
One of the key ingredients for the proof will be to find a combinatorial formula for the fused Specht 
polynomials (Proposition~\ref{prop:combinatorialformula}). 
Unfortunately, the arguments leading to the linear independence of $\{\mathcal F_T \; | \; T \in \RSYT\}$ 
and thus to Theorem~\ref{thm:theoremA} are valid only for Young diagrams with two columns.
However the combinatorial formula will hold for any shape. 
Thus, we believe that Theorem~\ref{thm:theoremA} also holds more generally: 

\begin{conjecture} \label{conj:theoremA}
Theorem~\ref{thm:theoremA} holds for Young diagrams of any shape.
\end{conjecture}

\begin{remark}
If Conjecture~\ref{conj:theoremA} holds, then the collection $\{ \Spn_{\bC} \{\mathcal F_T \; | \; T \in \RSYT\} \; | \; \lambda \in I_\multii \}$
is a complete set of pairwise non-isomorphic simple $\Hecke_\multii$-modules. 
\end{remark}

\subsubsection{Combinatorial formula for the fused Specht polynomials}

Consider the group
\begin{align*}
\mathfrak{Q}_\lambda := \SymGrp_{\Bar{\lambda_1}} \times \SymGrp_{\Bar{\lambda_2}} \times \cdots  \times \SymGrp_{\Bar{\lambda_l}} 
\; \subset \; \SymGrp_\Summed
\end{align*} 
where $\Bar{\lambda} = (\Bar{\lambda_1},\ldots ,\Bar{\lambda_l})$ (in particular, $\sum_i \Bar{\lambda_i} = \Summed$). 
Note that $\mathfrak{Q}_\lambda$ acts on $\Fill$ by permuting entries of a filling such that each factor $\SymGrp_{\Bar{\lambda_i}}$ permutes entries in the $i$th column. 
In what follows, we shall denote this action by \quote{$\star$}, in order to to avoid confusion with the action \quote{$.$} 
of $\SymGrp_\Summed$ on numberings defined in Section \ref{subsec:Specht}. 
\begin{example}
For instance, consider
\begin{align*}
T \; = \; {\ytableausetup{nosmalltableaux}\begin{ytableau} 
1 & 3 \\
1 & 3 \\
2 & 2
\end{ytableau}} 
\end{align*}
In this case, we have 
$\mathfrak{Q}_\lambda = \SymGrp_3 \times \SymGrp_3$. 
For instance, the permutation $\sigma = (13) \times \textnormal{Id} \in \mathfrak{Q}_\lambda$ 
exchanges the two entries lying in the first row, first column and third row, first column:
\begin{align*}
\sigma \star T \; = \; {\ytableausetup{nosmalltableaux}\begin{ytableau} 
2 & 3 \\
1 & 3 \\
1 & 2
\end{ytableau}} 
\end{align*}
Similarly, $\textnormal{Id} \times (12)$ leaves $T$ unchanged because it permutes two identical entries \quote{$3$}.
\end{example}

We denote by $\mathfrak{Q}_\lambda \star F$ the orbit of $F \in \Fill$ under the action of $\mathfrak{Q}_\lambda$. 
We also denote by $\Stab{\mathfrak{Q}_\lambda}{F}= \{ \sigma \in \mathfrak{Q}_\lambda \;|\; \sigma \star F = F\} \subset \mathfrak{Q}_\lambda $ the stabilizer of $F$ in $\mathfrak{Q}_\lambda$. 
\begin{remark}
For a numbering $N \in \NB$, the orbit $\mathfrak{Q}_\lambda \star N$ corresponds to $\Columns(N).N$, 
where $\Columns(N)$ is the column-stabilizer subgroup defined in Section~\ref{subsec:Specht}.
\end{remark}

Before proceeding, we fix some notation to be used throughout the rest of this section. 

\begin{notation}\label{not:combinatorial}
Let $W_F \subset \mathfrak{Q}_\lambda \star F$ be the subset of fillings in the orbit of $F$ 
which have at least two boxes containing the same entry in the same row. 
For $U \in (\mathfrak{Q}_\lambda \star F) \setminus W_F$, let $\sigma_{F;U}$ 
be the shortest permutation in $\mathfrak{Q}_\lambda$ such that $\sigma_{F;U} \star F = U$. 
We denote by $(r_i^U(k))_{i=1}^{s_k}$ the sequence of row numbers of boxes of $U$ containing the entry \quote{$k$}, 
ordered by column-reading $U$. 

Let $(r_i^{U,\ord}(k))_{i=1}^{s_k}$ be the ordering of $(r_i^U(k))_{i=1}^{s_k}$ in decreasing order, 
and let $\tau_{U;k}$ 
be permutations such that $(r_{\tau_{U;k}(i)}^U(k))_{i=1}^{s_k} = (r_i^{U,\ord}(k))_{i=1}^{s_k}$. 
Finally, let $\lambda^U(k)$ be the partition
\begin{align} \label{eq:defpartitionlambda}
\lambda^U(k) := \big( r_i^{U,\ord}(k)-s_k+i-1 \big)_{i=1}^{s_k}.
\end{align} 
\end{notation}

\begin{remark} \label{rem:remarklambdaU} 
Note that the elements of $(r_i^{U,\ord}(k))_{i=1}^{s_k}$ are all different, since $U$ is chosen with no two equal entries in the same row. 
Therefore, the elements of $\lambda^U(k)$ are nonnegative. 
Slightly abusing notation, if $\lambda^U(k)$ contains zeros at its tail, 
we identify it with the partition obtained by removing these zeros. 
\end{remark}

\begin{remark}
Let us mention that if $s_k=1$ for some $k$, that is, the entry \quote{$k$} appears exactly once in the filling, 
then $\lambda^F(k)$ is simply the row number where $k$ lies minus $1$. 
Also, if the entry \quote{$k$} appears in the rows $1,2,\ldots,s_k$ exactly once, then $\lambda^F(k) = \emptyset$.
\end{remark}

As a matter of convenience for the readers, we record two examples below. 
We focus on Young diagrams with two columns only, since only diagrams of this shape are considered in the subsequent sections (and in Theorem~\ref{thm:theoremA}).

\begin{example} \label{ex:examplepartition22} 
Let $\lambda=(2,2)$, so that $\Bar{\lambda} = (2,2)$ and $\mathfrak{Q}_\lambda = \SymGrp_2 \times \SymGrp_2$. 
Consider 
\begin{align*}
F \; = \; {\ytableausetup{nosmalltableaux}\begin{ytableau} 
1 & 2 \\
2 & 3
\end{ytableau}}  
\; \in \; \Fill 
\end{align*}
with $\multii = (s_1,s_2,s_3) = (1,2,1)$. 
The orbit $\mathfrak{Q}_\lambda \star F$ reads 
\begin{align*}
\mathfrak{Q}_\lambda \star F  \; = \; 
\Bigg\{ \; 
{\ytableausetup{nosmalltableaux}\begin{ytableau} 
1 & 2 \\
2 & 3
\end{ytableau}} 
\; , \; 
{\ytableausetup{nosmalltableaux}\begin{ytableau} 
2 & 2 \\
1 & 3
\end{ytableau}} 
\; , \; 
{\ytableausetup{nosmalltableaux}\begin{ytableau} 
1 & 3 \\
2 & 2
\end{ytableau}} 
\; , \; 
{\ytableausetup{nosmalltableaux}\begin{ytableau} 
2 & 3 \\
1 & 2
\end{ytableau}} 
\; \Bigg\} .
\end{align*}
Note that the stabilizer $\Stab{\mathfrak{Q}_\lambda}{F}$ consists only of the identity element, 
since there is no repeated entry within the same column. 
Therefore, we have $|\mathfrak{Q}_\lambda \star F| = |\mathfrak{Q}_\lambda|$. 
Moreover, for fillings containing repeated entries in the same row, we have 
\begin{align*}
W_F \; = \; 
\Bigg\{ \; 
{\ytableausetup{nosmalltableaux}\begin{ytableau} 
2 & 2 \\
1 & 3
\end{ytableau}}  
\; , \; 
{\ytableausetup{nosmalltableaux}\begin{ytableau} 
1 & 3 \\
2 & 2
\end{ytableau}}  
\; \Bigg\} .
\end{align*}
Finally, we have $r^F(1) = (1)$, $r^F(2) = (2,1)$, and $r^F(3) = (2)$. 
Hence we infer from~\eqref{eq:defpartitionlambda} that 
$\lambda^F(1) = \emptyset$, $\lambda^F(2) = \emptyset$, and $\lambda^F(3) = (1)$. 
\end{example}

\begin{example} \label{ex:examplepartition222} 
Let $\lambda = (2,2,2)$, so that $\Bar{\lambda} = (3,3)$ and $\mathfrak{Q}_\lambda = \SymGrp_3 \times \SymGrp_3$. 
Consider 
\begin{align*}
F \; = \; {\ytableausetup{nosmalltableaux}\begin{ytableau} 
2 & 3 \\
1 & 2 \\ 
2 & 3
\end{ytableau}}  
\; \in \; \Fill 
\end{align*}
with $\multii = (s_1,s_2,s_3)=(1,3,2)$. 
In this case, we have $\Stab{\mathfrak{Q}_\lambda}{F} \cong \SymGrp_2 \times \SymGrp_2$, 
since the entries \quote{$2$} and \quote{$3$} appear twice on the left and right column, respectively. 
Therefore, the orbit $\mathfrak{Q}_\lambda \star F$ contains $6^2/4 = 9$ elements 
(and $W_F$ consists of the last $6$ elements in $\mathfrak{Q}_\lambda \star F$):
\begin{align*}
%\mathfrak{Q}_\lambda.F  \; = \; 
%\Bigg\{ \; 
{\ytableausetup{nosmalltableaux}\begin{ytableau} 
2 & 3 \\
1 & 2 \\
2 & 3
\end{ytableau}} 
\; , \; 
{\ytableausetup{nosmalltableaux}\begin{ytableau} 
2 & 3 \\ 
2 & 3 \\ 
1 & 2
\end{ytableau}} 
\; , \; 
{\ytableausetup{nosmalltableaux}\begin{ytableau} 
1 & 2 \\ 
2 & 3 \\ 
2 & 3 
\end{ytableau}} 
\; , \; 
{\ytableausetup{nosmalltableaux}\begin{ytableau} 
2 & 2 \\ 
1 & 3 \\ 
2 & 3
\end{ytableau}} 
\; , \; 
{\ytableausetup{nosmalltableaux}\begin{ytableau} 
2 & 3 \\ 
1 & 3 \\ 
2 & 2
\end{ytableau}} 
\; , \; 
{\ytableausetup{nosmalltableaux}\begin{ytableau} 
2 & 2 \\ 
2 & 3 \\ 
1 & 3
\end{ytableau}} 
\; , \; 
{\ytableausetup{nosmalltableaux}\begin{ytableau} 
2 & 3 \\ 
2 & 2 \\ 
1 & 3
\end{ytableau}} 
\; , \; 
{\ytableausetup{nosmalltableaux}\begin{ytableau} 
1 & 3 \\ 
2 & 2 \\ 
2 & 3
\end{ytableau}} 
\; , \; 
{\ytableausetup{nosmalltableaux}\begin{ytableau} 
1 & 3 \\ 
2 & 3 \\ 
2 & 2
\end{ytableau}} 
%\; \Bigg\} 
.
\end{align*}
Finally, we have $r^F(1) = (2)$, $r^F(2) = (1,3,2)$, and $r^F(3) = (1,3)$. 
Thus $r^{F,\ord}(1) = r^F(1) = (2)$, $r^{F,\ord}(2) = (3,2,1)$, and $r^{F,\ord}(3) = (3,1)$. 
We then have the permutations $\tau_{F;2} = (132)$ and $\tau_{F;3} = (13)$, 
and we infer from~\eqref{eq:defpartitionlambda} that $\lambda^F(1) = (1)$, $\lambda^F(2) = \emptyset$, and $\lambda^F(3) = (1)$.
\end{example}

We now give a combinatorial formula for the fused Specht polynomials, equivalent to 
Equation~\eqref{eq:deffusedspecht} in Definition~\ref{def:fusedspecht}.
This formula is key to obtain the linear independence in the proof of Theorem~\ref{thm:theoremA},
and it is also of independent interest.  

\begin{proposition} \label{prop:combinatorialformula}
Fix $F \in \Fill$. To each $U \in (\mathfrak{Q}_\lambda \star F) \setminus W_F$, we associate the following monomial\footnote{Here, we use the convention that $\binom{a}{b}=0$ if $a<b$.}: 
\begin{align} \label{defmonomialmU}
m_U = m_U(x_1,\ldots ,x_\np) := \prod_{k=1}^\np \frac{(-1)^{\binom{s_k}{2}} \sign(\tau_{U;k})}{s_k!} \, \Schur_{\lambda^U(k)}(1^{s_k}) \, x_k^{|\lambda^U(k)|},
\end{align}
where $\Schur_{\lambda^U(k)}(1^{s_k})$ is the Schur polynomial associated with the partition $\lambda^U(k)$ 
and evaluated at $1$ for each of its $s_k$ variables 
(see Appendix~\ref{app:appendixschur} for the definition of Schur polynomials). 

Then, the fused Specht polynomial defined in~\eqref{def:fusedspecht} admits the following combinatorial formula:
\begin{align} \label{eq:combinatorialformula}
\mathcal F_F = |\Stab{\mathfrak{Q}_\lambda}{F}| \sum_{U \in (\mathfrak{Q}_\lambda \star F) \setminus W_F} \sign(\sigma_{F;U}) \; m_U .
\end{align}
\end{proposition}

\begin{remark}
The evaluation at $(1^{s_k})$ of the Schur polynomial $\Schur_{\lambda^U(k)}$,
\begin{align*}
\Schur_{\lambda^U(k)}(1^{s_k}) = \prod_{1\leq i < j \leq s_k} \frac{\lambda^U_i(k)-\lambda^U_j(k)+j-i}{j-i} ,
\end{align*}
equals the number of column-strict (semistandard) Young tableaux of shape $\lambda^U(k)$ 
and entries in $\{1,\ldots,s_k\}$ (and any content). 
In particular, $\Schur_\emptyset(1^{s_k}) = 1 = \Schur\sub{1^{s_k}}(1^{s_k})$ for $s_k\geq 1$.
\end{remark}

\begin{example}
Consider again 
\begin{align*}
F \; = \; {\ytableausetup{nosmalltableaux}\begin{ytableau} 
1 & 2 \\ 
2 & 3 
\end{ytableau}}  
\; \in \; \Fill 
\end{align*}
with $\lambda=(2,2)$ and $\multii = (s_1,s_2,s_3)=(1,2,1)$. 
As explained in Example~\ref{ex:examplepartition22}, the set $\mathfrak{Q}_\lambda \star F$ has four elements, 
two of them lying in $W_F$, and $|\Stab{\mathfrak{Q}_\lambda}{F}|=1$. 
Hence, the fused Specht polynomial $\mathcal F_F = \mathcal F_F(x_1,x_2,x_3)$ in~\eqref{eq:combinatorialformula} 
is a linear combination of two monomials:
\begin{align*}
\mathcal F_{\;\ytableausetup{smalltableaux}\ytableaushort{12,23}} 
\; = \; m_{\;\ytableausetup{smalltableaux}\ytableaushort{12,23}} \; + \; m_{\;\ytableausetup{smalltableaux}\ytableaushort{23,12}} 
\; = \; m_F \; + \; m_U ,
\qquad \textnormal{where} \qquad 
U \; = \; {\ytableausetup{nosmalltableaux}\begin{ytableau} 
2 & 3 \\ 
1 & 2
\end{ytableau}} .
\end{align*}
The permutation $\sigma_{F;U}$ is a product of two transpositions, so $\sign(\sigma_{F;U})=1$. 
The monomial $m_F$ is then calculated from Equation~\eqref{defmonomialmU} as follows:
\begin{align*}
\; & m_{\;\ytableausetup{smalltableaux}\ytableaushort{12,23}} \\
= \; & \frac{(-1)^{\binom{s_1}{2} + \binom{s_2}{2} + \binom{s_3}{2}} \sign(\tau_{F;1} \tau_{F;2} \tau_{F;3})}{s_1!s_2!s_3!} 
\, \Schur_{\lambda^F(1)}(1^{s_1}) \, \Schur_{\lambda^F(2)}(1^{s_2}) \, \Schur_{\lambda^F(3)}(1^{s_3}) 
\, x_1^{|\lambda^F(1)|} x_2^{|\lambda^F(2)|} x_3^{|\lambda^F(3)|} \\
= \; & \frac{(-1)^{\binom{1}{2} + \binom{2}{2} + \binom{1}{2}}}{1!2!1!} 
\, \Schur_{\emptyset}(1) \, \Schur_{\emptyset}(1,1) \, \Schur\sub{1}(1) 
\, x_1^{|\emptyset|} x_2^{|\emptyset|} x_3^{|(1)|} 
%\\
%= \; & - \frac{x_3}2.
\; = \; - \frac{x_3}2.
\end{align*}
The computation of 
\begin{align*}
m_{\;\ytableausetup{smalltableaux}\ytableaushort{23,12}} \; = \; \frac{x_1}2 
\end{align*}
is quite similar, with the difference that $r^U(2) = (1,2)$, so $r^{U,\ord}(2) = (2,1)$, 
which yield the transposition $\tau_{U;2} = (12)$ with $\sign(\tau_{U;2}) = -1$. 
We finally conclude that 
\begin{align*}
\mathcal F_{\;\ytableausetup{smalltableaux}\ytableaushort{12,23}} = \frac{x_1-x_3}{2}.
\end{align*}
\end{example}

\begin{example}
Consider then
\begin{align*}
F \; = \; {\ytableausetup{nosmalltableaux}\begin{ytableau} 
2 & 3 \\ 
1 & 2 \\ 
2 & 3
\end{ytableau}}  
\; \in \; \Fill 
\end{align*}
with $\lambda=(2,2,2)$ and $\multii = (s_1,s_2,s_3)=(1,3,2)$. 
As explained in Example~\ref{ex:examplepartition222}, 
the sets $\mathfrak{Q}_\lambda \star F$ and $W_F$ contain $9$ and $6$ elements, respectively. 
Hence, the fused Specht polynomial $\mathcal F_F$ is a linear combination of three monomials. 
Each monomial is weighted by a factor $|\Stab{\mathfrak{Q}_\lambda}{F}|=4$ 
and by the sign of the shortest permutation sending $U$ to $F$. 
More precisely, straightforward computations show that
\begin{align*}
\mathcal F_{\;\ytableausetup{smalltableaux}\ytableaushort{23,12,23}} 
\; = \; 4 \, m_{\;\ytableausetup{smalltableaux}\ytableaushort{23,12,23}} 
\; + \;
4 \, m_{\;\ytableausetup{smalltableaux}\ytableaushort{23,23,12}} 
\; + \;
4 \, m_{\;\ytableausetup{smalltableaux}\ytableaushort{12,23,23}} . 
\end{align*}
The monomial $m_F$ is then calculated from Equation~\eqref{defmonomialmU} as follows:
\begin{align*}
\; & m_{\;\ytableausetup{smalltableaux}\ytableaushort{23,12,23}} \\
= \; & \frac{(-1)^{\binom{s_1}{2} + \binom{s_2}{2} + \binom{s_3}{2}} \sign(\tau_{F;1} \tau_{F;2} \tau_{F;3})}{s_1!s_2!s_3!} 
\, \Schur_{\lambda^F(1)}(1^{s_1}) \, \Schur_{\lambda^F(2)}(1^{s_2}) \, \Schur_{\lambda^F(3)}(1^{s_3}) 
\, x_1^{|\lambda^F(1)|} x_2^{|\lambda^F(2)|} x_3^{|\lambda^F(3)|} \\
= \; & - \frac{(-1)^{\binom{1}{2} + \binom{3}{2} + \binom{2}{2}}}{1!3!2!} 
\, \Schur\sub{1}(1) \, \Schur_{\emptyset}(1,1,1) \, \Schur\sub{1}(1,1) 
\, x_1^{|(1)|} x_2^{|\emptyset|} x_3^{|(1)|} \\
= \; & - \frac{1}{12} \cdot 1 \cdot 1 \cdot 2 \cdot x_1 x_3 
%\\
%= \; & - \frac{x_1 x_3}{6}.
\; = \; - \frac{x_1 x_3}{6}.
\end{align*}
Note that the new subtlety in this example is that the Schur polynomial $\Schur\sub{1}(1,1)$ equals $2$. 
The other monomials are computed in a similar way. 
Altogether, we find that
\begin{align*}
\mathcal F_{\;\ytableausetup{smalltableaux}\ytableaushort{23,12,23}} 
= - \frac{2x_1 x_3}{3} + \frac{x_1^2}3 + \frac{x_3^2}3 = \frac{(x_1-x_3)^2}{3} .
\end{align*}
\end{example}

With the notation explained, we now proceed with the proof of Proposition~\ref{prop:combinatorialformula}.

\begin{proof}[Proof of Proposition~\ref{prop:combinatorialformula}]
The proof consists of an explicit computation of the formula~\eqref{eq:deffusedspecht} in Definition~\ref{def:fusedspecht} 
utilizing the expression~\eqref{eq:spechtmonomials} for the Specht polynomial as a sum over monomials. 
First of all, we write the Specht polynomial $\mathcal P_{\tilde{F}} = \mathcal P_{\tilde{F}}(x_1,\ldots ,x_\Summed)$ as follows:
\begin{align} \label{rewritespechtpolynomial}
\mathcal P_{\tilde{F}} = \sum_{N \in \mathfrak{Q}_\lambda \star \tilde{F}} \sign(\sigma_{\tilde{F};N}) \prod_{k=1}^\np \prod_{i=\summ_k}^{\summ_{k+1}-1} x_i^{l_i^{N}(k)-1},
\end{align}
where $(l^N_i(k))_{i=1}^{s_k} = (r^N(\summ_k),r^N(\summ_k+1),\ldots ,r^N(\summ_{k+1}-1))$ is the sequence of row numbers of the entries $\summ_k,\ldots,\summ_{k+1}-1$. 
Recall that $\idpt. \mathcal P_{\tilde{F}}$ in~\eqref{def:fusedspecht} is the antisymmetrization of $\mathcal P_{\tilde{F}}$ with respect to the groups of variables 
$x_{\summ_k},\ldots ,x_{\summ_{k+1}-1}$ for all $k \in \llbracket 1, \np\rrbracket$. 
The formula~\eqref{def:fusedspecht} can readily be brought to the following form:
\begin{align}\label{eq:fusedspechtsum}
\mathcal F_F = \left[\sum_{N \in \mathfrak{Q}_\lambda \star \tilde{F}} \sign(\sigma_{\tilde{F};N}) \prod_{k=1}^\np \frac{1}{s_k! (-1)^{\binom{s_k}{2}}} \frac{\sum_{\sigma \in \SymGrp_{s_k}} \sign(\sigma) \prod_{i=\summ_k}^{\summ_{k+1}-1} x_{\sigma(i)}^{l^N_i(k)-1}}{\prod_{\summ_k\leq i<j<\summ_{k+1}} (x_i-x_j)} \right]_{\textnormal{eval}}.
 \end{align}
(Note that we introduced a factor $(-1)^{\binom{s_k}{2}}$ to replace $x_j-x_i$ by $x_i-x_j$ in the denominator.) 
Now, denote $\hat{W}_k := \{N \in \mathfrak{Q}_\lambda \star \tilde{F} \; | \; l^N_i(k) = l^N_j(k) \; \textnormal{for some} \; (i,j) \in \llbracket 1,s_k \rrbracket^2, \, i \neq j \}$, and set $\hat{W} := \bigcup_{k=1}^\np \hat{W}_k$. 
Any numbering $N \in \hat{W}$ leads to a vanishing term in the sum~\eqref{eq:fusedspechtsum} 
because the product $\prod_{m=\summ_k}^{\summ_{k+1}-1} x_m^{l^N_m(k)-1}$ is a symmetric function of at least two variables, 
which therefore vanishes upon antisymmetrization. 
Thus, we obtain
\begin{align*}
\mathcal F_F = \left[\sum_{N \in (\mathfrak{Q}_\lambda \star \tilde{F})\setminus \hat{W}} \sign(\sigma_{\tilde{F};N}) \prod_{k=1}^\np \frac{1}{s_k! (-1)^{\binom{s_k}{2}}} \frac{\sum_{\sigma \in \SymGrp_{s_k}} \sign(\sigma) \prod_{i=\summ_k}^{\summ_{k+1}-1} x_{\sigma(i)}^{l^N_i(k)-1}}{\prod_{\summ_k\leq i<j<\summ_{k+1}} (x_i-x_j)} \right]_{\textnormal{eval}}.
\end{align*}

Let $l^{N,\ord}_i(k)$ be the ordering of $l^N_i(k)$ in decreasing order, 
and let $\tau_{N;k}$ be a permutation such that $(l^N_{\tau_{N;k}(i)}(k))_{i=1}^{s_k} = (l^{N,\ord}_i(k))_{i=1}^{s_k}$. 
We reorganize the sum in the numerator as 
\begin{align*}
\mathcal F_F = \left[\sum_{N \in (\mathfrak{Q}_\lambda \star \tilde{F})\setminus \hat{W}} \sign(\sigma_{\tilde{F};N}) \prod_{k=1}^\np \frac{\sign(\tau_{N;k})}{s_k! (-1)^{\binom{s_k}{2}}} \frac{\sum_{\sigma \in \SymGrp_{s_k}} \sign(\sigma) \prod_{i=\summ_k}^{\summ_{k+1}-1} x_{\sigma(i)}^{l^{N,\ord}_i(k)-1}}{\prod_{\summ_k\leq i<j<\summ_{k+1}} (x_i-x_j)} \right]_{\textnormal{eval}}.
\end{align*}
For a given $N \in (\mathfrak{Q}_\lambda \star \tilde{F})\setminus \hat{W}$, 
we introduce the partition\footnote{Since $N \notin \hat{W}$, it follows that $\big(\hat{\lambda}_i^N(k) \big)_{i=1}^{s_k}$ is a partition, similarly as in Remark~\ref{rem:remarklambdaU}.}
\begin{align*}
\big( \hat{\lambda}_i^N(k) \big)_{i=1}^{s_k} = \big( l^{N,\ord}_i(k) - s_k + i - 1 \big)_{i=1}^{s_k}.
\end{align*}
We now recognize the Schur polynomial~\eqref{eq:defslambda} discussed in Appendix~\ref{app:appendixschur}:
\begin{align*}
\frac{\sum_{\sigma \in \SymGrp_{s_k}} \sign(\sigma) \prod_{i=\summ_k}^{\summ_{k+1}-1} x_{\sigma(i)}^{\hat{\lambda}_i^N(k)+s_k-i}}{\prod_{\summ_k\leq i<j<\summ_{k+1}} (x_i-x_j)} = \Schur_{\hat{\lambda}^N(k)}(x_{\summ_k},\ldots ,x_{\summ_{k+1}-1}).
\end{align*}
Therefore, we infer that
\begin{align*}
\mathcal F_F= \left[\sum_{N \in (\mathfrak{Q}_\lambda \star \tilde{F})\setminus \hat{W}} \sign(\sigma_{\tilde{F};N}) \prod_{k=1}^\np \frac{\sign(\tau_{N;k})}{s_k! (-1)^{\binom{s_k}{2}}} 
\, \Schur_{\hat{\lambda}^N(k)}(x_{\summ_k},\ldots ,x_{\summ_{k+1}-1}) \right]_{\textnormal{eval}}.
\end{align*}

We now investigate the sum over the numberings $N$ in more detail. We have
\begin{align*}
(\mathfrak{Q}_\lambda \star \tilde{F}) \setminus \hat{W} = \bigcup_{U \in (\mathfrak{Q}_\lambda \star F)\setminus W_F} (\Stab{\mathfrak{Q}_\lambda}{U}) \star \tilde{U},
\end{align*}
where $\tilde{U}$ is the numbering associated with the filling $U$, 
and $(\Stab{\mathfrak{Q}_\lambda}{U}) \star \tilde{U}$ is the orbit of $\tilde{U}$ under the action of $\Stab{\mathfrak{Q}_\lambda}{U}$. 
Since the right-hand side is clearly a disjoint union of sets, we see that $\mathcal F_F$ equals
\begin{align} \label{eq:intermediateeq}
\hspace*{-4mm}
%\mathcal F_F = 
\left[\sum_{U \in (\mathfrak{Q}_\lambda \star F)\setminus W_F} \sum_{N \in (\Stab{\mathfrak{Q}_\lambda}{U}) \star \tilde{U}} \sign(\sigma_{\tilde{F};N}) \prod_{k=1}^\np \frac{\sign(\tau_{N;k})}{s_k! (-1)^{\binom{s_k}{2}}} 
\, \Schur_{\hat{\lambda}^N(k)}(x_{\summ_k},\ldots ,x_{\summ_{k+1}-1})  \right]_{\textnormal{eval}}.
\end{align}

The last step of the proof consists of showing that all of the terms 
in the sum over $N \in (\Stab{\mathfrak{Q}_\lambda}{U}) \star \tilde{U}$ are equal. 
To this end, let us consider some filling $U \in (\mathfrak{Q}_\lambda \star F)\setminus W_F$ and two numberings $N_1, N_2 \in (\Stab{\mathfrak{Q}_\lambda}{U}) \star \tilde{U}$ with $N_1 \neq N_2$. 
The key observation is that, although the sequences $l^{N_1}(k)$ and $l^{N_2}(k)$ are different for at least one index $k \in \llbracket 1, \np\rrbracket$, 
we have $l^{N_1,\ord}(k)=l^{N_2,\ord}(k)$ and therefore $\hat\lambda^{N_1}(k) = \hat\lambda^{N_2}(k)$ for all $k \in \llbracket 1, \np\rrbracket$. 
(This observation holds because $\Stab{\mathfrak{Q}_\lambda}{U}$ can only permute numbers in a subset $\{q_k,\cdots,q_{k+1}-1\}$.) 
Hence, we have $\hat{\lambda}^N(k)=\lambda^U(k)$ for all $N \in (\Stab{\mathfrak{Q}_\lambda}{U}) \star \tilde{U}$. 
It remains to prove that
\begin{align} \label{eq:sign1}
\sign(\sigma_{\tilde{F};N_1}) \prod_{k=1}^\np \sign(\tau_{N_1;k}) = \sign(\sigma_{\tilde{F};N_2}) \prod_{k=1}^\np \sign(\tau_{N_2;k}).
\end{align}
Because there exists a permutation $\omega \in \Stab{\mathfrak{Q}_\lambda}{U}$ such that $\omega \star N_1 = N_2$, we have
\begin{align} \label{eq:sign2}
\sign(\sigma_{\tilde{F};N_2}) = \sign(\sigma_{\tilde{F};\omega \star N_1}) = \sign(\omega) \, \sign(\sigma_{\tilde{F};N_1}).
\end{align}
Moreover, on the one hand, $\omega$ takes the form $\omega = \prod_{k=1}^\np \omega_k$, 
where each $\omega_k$ acts on the boxes containing the entries $\summ_k,\ldots ,\summ_{k+1}-1$, 
while on the other hand, we have
\begin{align*}
l^{N_2}(k) = l^{\omega_k \star N_1}(k), \qquad \textnormal{for all } k \in \llbracket 1, \np\rrbracket ,
\end{align*}
which implies that
\begin{align} \label{eq:sign3}
\sign(\tau_{N_2;k}) = \sign(\tau_{\omega_k \star N_1; k}) = \sign(\omega_k) \, \sign(\tau_{N_1;k}).
\end{align}
Therefore~\eqref{eq:sign1} follows from~(\ref{eq:sign2},~\ref{eq:sign3}). 
We thereby conclude that the second sum in~\eqref{eq:intermediateeq} contains $|\Stab{\mathfrak{Q}_\lambda}{U}|$ times the same term, 
and we in fact have $|\Stab{\mathfrak{Q}_\lambda}{U}| = |\Stab{\mathfrak{Q}_\lambda}{F}|$. 
Taking $\tilde{U}$ for the representative of the orbit $(\Stab{\mathfrak{Q}_\lambda}{U}) \star \tilde{U}$, we finally obtain
\begin{align*}
\mathcal F_F = |\Stab{\mathfrak{Q}_\lambda}{F}| \left[\sum_{U \in (\mathfrak{Q}_\lambda \star F)\setminus W_F} \sign(\sigma_{\tilde{F};\tilde{U}}) \prod_{k=1}^\np \frac{\sign(\tau_{\tilde{U}; k})}{s_k! (-1)^{\binom{s_k}{2}}} 
\, \Schur_{\lambda^{ U}(k)}(x_{\summ_k},\ldots ,x_{\summ_{k+1}-1}) \right]_{\textnormal{eval}}.
\end{align*}

It finally remains to perform the evaluations of variables $x_{\summ_k}=x_{\summ_k+1}=\cdots =x_{\summ_{k+1}-1}$ for all $k \in \llbracket 1, \np\rrbracket$. 
First of all, note that $\sign(\tau_{\tilde{U}; k}) = \sign(\tau_{U;k})$. 
Permutations between $F$ and $U$ may differ by products of transpositions exchanging boxes in the same column and having the same entry. 
However, the permutation sending $\Tilde{F}$ to $\Tilde{U}$ does not contain any such transposition in its decomposition. 
Thus, the permutation sending $\tilde{F}$ to $\tilde{U}$ is the shortest permutation sending $F$ to $U$. 
Hence, we have $\sign(\sigma_{\tilde{F};\tilde{U}}) = \sign(\sigma_{F;U})$. 
Lastly, the evaluation of the Schur polynomial is obtained from the identity~\eqref{eq:evaluationschur} from Appendix~\ref{app:appendixschur}. 
\end{proof}

\begin{remark} \label{remarkfusedspechtforvalence1}
If $s_k=1$ for all $k$, then the filling $F$ becomes a numbering and the fused Specht polynomial $\mathcal F_F$ 
in~\eqref{eq:combinatorialformula} becomes a Specht polynomial $\mathcal P_F$. To see this, let us choose $N \in \NB$ in the formula~\eqref{eq:combinatorialformula}. In this case $W_N$ is the empty set and $\Stab{\mathfrak{Q}_\lambda}{N}$ is the trivial group. Moreover, since $s_k=1$ for all $k$, 
the sequences $(r_i^U(k))_{i=1}^{s_k}$ contain one element only, which is the row number $r^U(k)$ where the entry \quote{$k$} lies. Thus, $\tau_{U;k}$ is the identity permutation, and the partition in~\eqref{eq:defpartitionlambda} becomes $(r^U(k)-1)$. 
This implies that $s_{\lambda^U(k)}(1^{s_k})=1$. 
Altogether, the formula~\eqref{eq:combinatorialformula} reduces to~\eqref{eq:spechtmonomials}, as expected:
\begin{align*}
\mathcal F_N = \sum_{U \in \Columns(N).N} \sign(\sigma_{N;U}) \prod_{k=1}^\np x_k^{r^U(k)-1} =  \mathcal P_N.
\end{align*}
\end{remark}

\subsubsection{Linear independence of the fused Specht polynomials}

We will next show that the set $\{\mathcal F_T \;|\; T\in \RSYT\}$ is a collection of non-zero vectors (Lemma~\ref{lem:lemmanonzero1}). 
In the case where $\lambda$ is a Young diagram with two columns, we show in addition that 
$\{\mathcal F_T \;|\; T\in \RSYT\}$ is a set of linearly independent vectors (Proposition~\ref{prop:proplinearindepfusedspecht}). 
This implies Theorem~\ref{thm:theoremA}.

\begin{lemma} \label{lem:lemmanonzero1}
Let $T \in \RSYT$ and consider the set 
$\{(|\lambda^{U}(k)|)_{k\in\llbracket 1,d \rrbracket} \;|\; U\in (\mathfrak{Q}_\lambda \star T) \setminus W_T \}$, 
where $(\mathfrak{Q}_\lambda \star T) \setminus W_T$ indexes the sum in~\eqref{eq:combinatorialformula}; recall also Notation~\ref{not:combinatorial}. 
Then, $(|\lambda^{T}(k)|)_{k\in\llbracket 1, d \rrbracket}$ is the unique minimum for the lexicographic order in this set. 
Hence, the coefficient of the monomial 
\begin{align} \label{eq:monomial}
\prod_{k=1}^{d} x_k^{|\lambda^{T}(k)|} 
\end{align}
in $\mathcal{F}_T$ equals 
(and implies in particular that $\mathcal{F}_T$ is non-zero)
\begin{align} \label{eq:monomial_coef}
|\Stab{\mathfrak{Q}_\lambda}{T}|\prod_{k=1}^{d}  \frac{\sign(\tau_{T;k})}{s_k!(-1)^{\binom{s_k}{2}}} \, \Schur_{\lambda^T(k)}(1^{s_k}) ,
\end{align}
where $\Schur_{\lambda^T(k)}(1^{s_k})$ is the Schur polynomial associated with the partition $\lambda^T(k)$ 
and evaluated at $1$ for each of its $s_k$ variables 
(see Appendix~\ref{app:appendixschur} for the definition of Schur polynomials). 
\end{lemma}

\begin{proof}
Fix $U\in (\mathfrak{Q}_\lambda \star T) \setminus W_T$ such that $U\neq T$. 
Let $i\in \llbracket 1,d\rrbracket$ be the smallest index such that $U$ and $T$ differ at the positions of $i$. 
Consider two \quote{skew} Young tableaux $T'$ and $U'$ obtained from $T$ and $U$ 
by removing boxes containing a number in $\llbracket1,i-1\rrbracket$. 
Since the entries of $T'$ are weakly increasing along the columns, we see that 
\begin{align*}
\sum_{j=1}^{s_i}r_j^U(i) > \sum_{j=1}^{s_i}r_j^T(i),
\end{align*}
which, using the definition~\eqref{eq:defpartitionlambda} leads to
\begin{align*}
(|\lambda^{U}(k)|)_{k\in\llbracket 1, d \rrbracket}>(|\lambda^{T}(k)|)_{k\in\llbracket 1, d \rrbracket}.
\end{align*}
We then infer that $(|\lambda^{T}(k)|)_{k\in\llbracket 1, d \rrbracket}$ is indeed 
a minimum for the lexicographic order in the set 
$\{(|\lambda^{U}(k)|)_{k\in\llbracket 1,d \rrbracket} \;|\; U\in (\mathfrak{Q}_\lambda \star T) \setminus W_T \}$.
Consequently, the only monomial in
\begin{align*}
\mathcal F_T = |\Stab{\mathfrak{Q}_\lambda}{T}| \sum_{U \in (\mathfrak{Q}_\lambda \star T) \setminus W_T} \sign(\sigma_{T;U}) \; m_U
\end{align*}
proportional to~\eqref{eq:monomial} is obtained at $U=T$. 
This gives the coefficient~\eqref{eq:monomial_coef}.
\end{proof}

\begin{lemma} \label{lem:lemmainj1} 
Let $\lambda$ be a Young diagram with two columns. 
The map $T\in \RSYT \mapsto (|\lambda^{T}(k)|)_{k\in\llbracket 1, d \rrbracket}$ is injective. 
\end{lemma}

\begin{proof}
Let $T,T'\in\RSYT$ such that 
$(|\lambda^{T}(k)|)_{k\in\llbracket 1, d \rrbracket}=(|\lambda^{T'}(k)|)_{k\in\llbracket 1, d \rrbracket}$. 
Suppose $T\neq T'$. 
Let $i$ be the smallest index such that $T$ and $T'$ differ at the positions of $i$. 
Consider two \quote{skew} Young diagrams obtained by keeping only boxes containing $i$ in $T$ (resp.~$T'$):
both consist of either one single column diagram, or two disconnected column diagrams.
Because these skew diagrams are different for $T$ and $T'$, both $T$ and $T'$ have two columns 
and have the same set of removed boxes, we have
\begin{align*}
\sum_{j=1}^{s_i}r_j^{T'}(i) \neq \sum_{j=1}^{s_i}r_j^T(i) .
\end{align*}
However, this implies $|\lambda^{T'}(i)|\neq |\lambda^{T}(i)|$, a contradiction.
Hence, we deduce that $T =T'$.
\end{proof}

\begin{remark}
Lemma~\ref{lem:lemmainj1} does not hold when $\lambda$ is a Young diagram with more than two columns. 
As a counterexample, let us consider the two row-strict Young tableaux   
\begin{align*}
T_1 \; = \; {\ytableausetup{nosmalltableaux}\begin{ytableau} 
1 & 2 & 3 \\
1 & 4 \\
1 & 4  \\
3 
\end{ytableau}} 
\qquad \qquad \textnormal{and} \qquad \qquad
T_2 \; = \; {\ytableausetup{nosmalltableaux}\begin{ytableau} 
1 & 2 & 4 \\
1 & 3 \\
1 & 3  \\
4 
\end{ytableau}} 
\end{align*}
In this case, we have
\begin{align*}
\sum_{j=1}^{2}r_j^{T_1}(3) = \sum_{j=1}^{2}r_j^{T_2}(3) = 5 = \sum_{j=1}^{2}r_j^{T_1}(4) = \sum_{j=1}^{2}r_j^{T_2}(4) ,
\end{align*}
which implies in particular that
\begin{align*}
|\lambda^{T_1}(i)| = |\lambda^{T_2}(i)|, \qquad \textnormal{for all } i=1,2,3,4 .
\end{align*}
Hence, the map $T\in \RSYT \mapsto (|\lambda^{T}(k)|)_{k\in\llbracket 1, d \rrbracket}$ is not injective in this case.
\end{remark}

\begin{proposition} \label{prop:proplinearindepfusedspecht}
Let $\lambda$ be a Young diagram with two columns. 
The set $\{\mathcal F_T \; | \; T \in \RSYT\}$ is linearly independent. 
\end{proposition}

\begin{proof}
Suppose that there is a linear relation with coefficients 
$(\alpha_T)$ not identically zero:
\begin{align} \label{eq:linear_relation}
\sum_{T\in \RSYT} \alpha_T \mathcal{F}_T=0 , \qquad \alpha_T \in \bR 
, \quad \textnormal{ and } \quad 
\alpha_T \neq 0 \textnormal{ for some } T \in \RSYT .
\end{align}
Take $T \in \{U \in \RSYT \; | \; \alpha_U \neq 0 \}$ 
such that $(|\lambda^{T}(k)|)_{k\in\llbracket 1, d \rrbracket}$ 
is the minimum for the lexicographic order (unique by Lemma~\ref{lem:lemmainj1}). 
By Lemma~\ref{lem:lemmanonzero1}, the coefficient of the monomial~\eqref{eq:monomial} in
the linear relation~\eqref{eq:linear_relation} is nonzero, which is a contradiction.
\end{proof}

This concludes the proof of Theorem~\ref{thm:theoremA}. 
To prove Conjecture~\ref{conj:theoremA}, one should find an argument replacing Lemma~\ref{lem:lemmainj1}.

\bigskip{}
\section{The space of $c=1$ degenerate conformal blocks} 
\label{sec:section3}
Next, we apply the fused Specht polynomials from Section~\ref{sec:sectionhecke} 
to construct a basis for a space of conformal blocks 
in a CFT with central charge $c=1$ 
and conformal weights in the Kac table~\eqref{eq:conf_weights}.
Their correlation functions form a basis for a solution space $\SolSp_\multii$ of a special class of BPZ PDEs, 
known as \quote{Beno{\^i}t~\&~Saint-Aubin equations}~\cite{BSA:Degenerate_CFTs_and_explicit_expressions_for_some_null_vectors} (that we will call \quote{conformal block basis functions}),  see~Theorem~\ref{thm:theoremBSA}.
We also gather some algebraic structure related to the conformal block basis: 
in particular, we show that $\SolSp_\multii$ is isomorphic to 
a standard module of the valenced Temperley-Lieb algebra (Proposition~\ref{prop:repvTL}).

The key importance of these conformal block basis functions is that they 
are expected (and in some cases known) 
to give rise to a family of conformally invariant boundary conditions for the 
Gaussian free field (GFF)~\cite{Sheffield-Miller:Imaginary_geometry1,Peltola-Wu:Global_and_local_multiple_SLEs_and_connection_probabilities_for_level_lines_of_GFF,Liu-Wu:Scaling_limits_of_crossing_probabilities_in_metric_graph_GFF}.
Concrete formulas for them will thus be needed in applications for problems in random geometry 
(to which we plan to return in future work).
With this in mind, we briefly discuss the relationship of our construction with the prior literature and show that special cases of our conformal block basis functions 
indeed equal the ones used in GFF applications --- see Section~\ref{subsec:valence_1_blocks}.

Throughout the rest of this section, 
we assume that $\Summed=2N$ is a given positive even integer, 
and all (fused) Specht polynomials will be associated with two-column rectangular Young tableaux of $\Summed$ boxes.

\subsection{Conformal blocks for unit valences $\multii=(1^\Summed)$ with $\Summed=2N$}
\label{subsec:valence_1_blocks}

Recall that $\SYTof{N,N}$ is the 
set of standard Young tableaux of shape $\lambda = (N,N)$.
For each $T \in \SYTof{N,N}$, we associate its \emph{transpose} $T^t \in \SYTof{2^{N}}$ which is 
obtained by exchanging the rows and columns of $T$.

\begin{definition} \label{def:CBsh12def}
For each $T \in \SYTof{N,N}$, we define 
the \emph{conformal block basis function} as
\begin{align} \label{eq:CBsh12}
\CobloF_T(x_1,\ldots ,x_{2N}) := 
\Delta(x_1,\ldots ,x_{2N})^{-1/2} \, \mathcal P_{T^t}(x_1,\ldots ,x_{2N}),
\end{align}
where $\mathcal P_{T^t}$ is the Specht polynomial~\eqref{eq:spechtfactorized} 
and $\Delta$ is the Vandermonde determinant~\eqref{eq:Vandermonde determinant}.
\end{definition}

Note that the conformal block functions are positive functions $\CobloF_T \colon \chamber_{2N} \to \bRpos$ on 
\begin{align*}
\chamber_{2N} := \{(x_1,\ldots ,x_{2N}) \in \bR^{2N} \; | \; x_1<\cdots <x_{2N}\}.
\end{align*}
Because they are also M\"obius covariant and satisfy a system of second order BPZ PDEs 
(see~(\ref{eq:BPZeq},~\ref{eq:covariancepropertyvalence1})), 
they give rise to \quote{partition functions} for interacting Schramm-Loewner evolution, $\SLE_4$ curves. This fact is important for their probabilistic interpretation~\cite{Peltola-Wu:Global_and_local_multiple_SLEs_and_connection_probabilities_for_level_lines_of_GFF}.

Generally speaking, in this section we consider positive smooth functions $F \colon \chamber_{2N} \to \bRpos$ 
satisfying the below three properties. The first one is the following system of BPZ PDEs: 
\begin{align} \label{eq:BPZeq}
\bigg(\frac{\partial^2}{\partial x_j^2} 
+ \sum_{k \neq j} \bigg( \frac{1}{x_k-x_j} \frac{\partial}{\partial x_k} 
- \frac{1/4}{(x_k-x_j)^2} \bigg) \bigg) \; F(x_1,\ldots ,x_{2N}) = 0, \qquad j \in\llbracket 1, 2N\rrbracket .
\end{align}
Second, for all M\"obius transformations $\varphi \colon \bH \to \bH := \{z \in \bC \;|\; \im(z) > 0\}$ of the upper half-plane
such that $\varphi(x_1)<\cdots <\varphi(x_{2N})$, 
we require the covariance 
\begin{align} \label{eq:covariancepropertyvalence1}
F(\varphi(x_1),\ldots ,\varphi(x_{2N})) = \prod_{i=1}^{2N} \varphi'(x_i)^{-1/4} \times F(x_1,\ldots ,x_{2N}).
\end{align}
Finally, we insist that there exist constants $C>0$ and $p>0$ such that for all $N \geq 1$ and 
$(x_1, \ldots , x_{2N}) \in \chamber_{2N}$, the following power-law bound holds:
\begin{align} \label{eq:estimateZ}
F(x_1,\ldots ,x_{2N}) \leq C \hspace*{-3mm} \prod_{1 \leq i<j \leq {2N}} |x_j-x_i|^{\mu_{ij}(p)} \quad \textnormal{with} \quad \mu_{ij}(p) := 
\begin{cases} 
-p, \quad |x_i-x_j|<1, \\ 
+p, \quad |x_i-x_j| \geq 1. 
\end{cases}
\end{align}

The first space of interest to us describes correlation functions with Kac type conformal weights $h_{1,2} = 1/4$ as in~\eqref{eq:conf_weights} 
for a conformal field theory of central charge $c=1$:
\begin{align} \label{eq:defspaceS}
\SolSp\sub{1^{2N}} = 
\SolSp\sub{1,\ldots,1} 
:= \{F \colon \chamber_{2N} \to \bR \;|\; F \; \textnormal{satisfies}~\eqref{eq:BPZeq},~\eqref{eq:covariancepropertyvalence1},~\textnormal{and}~\eqref{eq:estimateZ}\} .
\end{align} 
It follows from the results~\cite{Flores-Kleban:Solution_space_for_system_of_null-state_PDE1, 
Flores-Kleban:Solution_space_for_system_of_null-state_PDE2} 
of Flores~\&~Kleban 
that $\dim \SolSp\sub{1^{2N}}$ equals the $N$-th Catalan number.
We will see that the conformal block basis functions $\{\CobloF_T \; | \; T \in \SYTof{N,N} \}$ of Definition~\ref{def:CBsh12def} indeed 
span $\SolSp\sub{1^{2N}}$ and are linearly independent.
Indeed, to establish this we only need to show that they coincide with the conformal blocks in~\cite[Eq.~(6.1)]{Peltola-Wu:Global_and_local_multiple_SLEs_and_connection_probabilities_for_level_lines_of_GFF},
which was proven to be a basis for $\SolSp\sub{1^{2N}}$ 
by Peltola~\&~Wu~\cite{Peltola-Wu:Global_and_local_multiple_SLEs_and_connection_probabilities_for_level_lines_of_GFF}.

\begin{lemma} \label{prop:lemma3p3}
The collection $\{\CobloF_T \; | \; T \in \SYTof{N,N} \}$ is a basis for $\SolSp\sub{1^{2N}}$.
\end{lemma}

\begin{proof}
Observe that the set $\SYTof{N,N}$ of standard Young tableaux of shape $\lambda = (N,N)$ is in bijection with 
the set $\LP_{N}$ of planar $N$-link patterns, that is, planar pair partitions 
$\alpha = \{\{a_1,b_1\},\{a_2,b_2\},\ldots ,\{a_{N},b_{N}\}\}$ of the set $\{1,2,\ldots ,{2N}\}$.
(The latter can be used to label connectivities of planar curves as in~\cite{Peltola-Wu:Global_and_local_multiple_SLEs_and_connection_probabilities_for_level_lines_of_GFF}.)
Indeed, without loss of generality, we may assume that
$a_1 < a_2 < \cdots < a_{N}$ and $a_j < b_j$ for all $j$.
Then, it is not hard to check that a bijection is obtained by 
sending the element of the first row and $i$-th column of a tableau $T \in \SYTof{N,N}$ to 
the $i$-th element of $\{a_1, a_2, \ldots, a_{N}\}$ associated with $\alpha \in \LP_{N}$,
and the elements of the second row of $T$ to the elements of $\{b_1,b_2,\ldots ,b_{N}\}$ ordered in such a way that $\alpha$ becomes a planar pairing --- 
by the fact that $T$ is strictly increasing across each row and down each column 
one ensures\footnote{The bijection is more precise after passing to balanced parenthesis expressions as in~\cite[Sect.~2.1]{KKP:Boundary_correlations_in_planar_LERW_and_UST}.} 
that there is a unique way for this ordering. 
Conversely, given $\alpha$, 
by placing $\{a_1, a_2, \ldots, a_{N}\}$ and $\{b_1,b_2,\ldots ,b_{N}\}$ into the two rows of $T$ 
with the latter rearranged in increasing order readily yields a standard Young tableau $T \in \SYTof{N,N}$. 
Using this bijection, we obtain
\begin{align*}
\CobloF_T(x_1,\ldots ,x_{2N}) 
= \prod_{1 \leq i < j \leq {2N}} (x_j-x_i)^{\frac12 \theta_{\alpha(T)}(i,j)} ,
\end{align*}
where $\alpha(T) \in \LP_{N}$ is the link pattern corresponding to $T \in \SYTof{N,N}$ and
\begin{align*}
\theta_\alpha(i,j) := 
\begin{cases} 
+1, \quad i,j \in \{a_1,a_2,\ldots ,a_{N} \}  \; \textnormal{or} \; i,j \in \{b_1,b_2,\ldots ,b_{N}\}, \\
%\hphantom{+} 0, \quad i=j, \\
-1, \quad \textnormal{otherwise}.
\end{cases}
\end{align*}
This is exactly~\cite[Eq.~(6.1)]{Peltola-Wu:Global_and_local_multiple_SLEs_and_connection_probabilities_for_level_lines_of_GFF},
which is known to form a basis for $\SolSp\sub{1^{2N}}$.
\end{proof}

\begin{remark} \label{rem:Dyck path}
Both $\SYTof{N,N}$ and $\LP_{N}$ are in bijection with the set $\DP_{N}$ of 
\emph{Dyck paths}: walks $\varpi$ on $\bZnn$ of ${2N}$ steps with steps of length one, starting and ending at zero. 
The conformal block functions $\CobloF_T$ can then be related to piecewise constant Dirichlet boundary conditions for the GFF (see~\cite[Sect.~6.4]{Peltola-Wu:Global_and_local_multiple_SLEs_and_connection_probabilities_for_level_lines_of_GFF} for details) as follows. 
For fixed $x_1 < \cdots < x_{2N}$, 
consider the GFF $\Gamma_\varpi$ on the upper-half plane $\bH := \{z \in \bC \;|\; \im(z) > 0\}$ with boundary data\footnote{Here, we use the convention that $x_{0} = -\infty$ and $x_{{2N}+1} = +\infty$.}  
\begin{align*}
\tfrac{\pi}{2} \, (2\varpi(k) - 1) , 
\quad & \textnormal{if } x \in ( x_{k}, x_{k+1} ) , \qquad \textnormal{ for all } 
k \in \llbracket 0, 2N\rrbracket ,
\end{align*}
Then, the level lines of $\Gamma_\varpi$ started at the points $(x_1, \ldots , x_{2N})$ 
are $\SLE_4$ curves with partition function $\CobloF_{T(\varpi)}(x_1, \ldots , x_{2N})$, 
where $T(\varpi) \in \SYTof{N,N}$ is the Young tableau corresponding to the Dyck path $\varpi$ 
(this is a special case of~\cite[Thm.~1.1]{Sheffield-Miller:Imaginary_geometry1}; 
see~\cite[Prop.~6.8]{Peltola-Wu:Global_and_local_multiple_SLEs_and_connection_probabilities_for_level_lines_of_GFF} for details). 
This model was further investigated by Liu~\&~Wu in~\cite{Liu-Wu:Scaling_limits_of_crossing_probabilities_in_metric_graph_GFF}.
\end{remark}

\begin{remark} \label{rem:Dyck path gen}
Note that since each Dyck path $\varpi$ has steps of length one (cf.~Remark~\ref{rem:Dyck path}), the height gaps in the GFF $\Gamma_\varpi$ have absolute value $\pi$.
This is also the most common height gap when considering level lines of the GFF~\cite{Schramm-Sheffield:A_contour_line_of_the_continuum_GFF}. 
Liu~\&~Wu defined in~\cite[Eq.~(5.15,~5.16,~5.17)]{Liu-Wu:Scaling_limits_of_crossing_probabilities_in_metric_graph_GFF} 
three functions generalizing the conformal block functions and related them to height gaps of absolute value $2\pi$.
It is not hard to check that these functions are the three elements of our conformal block basis $\SolSp\sub{2,2,2,2}$, 
which we define in the next section using the fused Specht  polynomials. 
We shall detail this connection in Remark~\ref{rem:Dyck path gen2}.
\end{remark}

\subsection{Temperley-Lieb action and braiding}
\label{subsec: TL action}

Next, we make explicit the action on the space $\SolSp\sub{1^{2N}}$
of the Temperley-Lieb algebra $\TL_{2N} = \TL_{2N}(\nu) = \TL_{2N}(2)$, 
with loop fugacity $\nu := -q - q^{-1} \in \bC$ equaling $2$ for $q=-1$. 
It arises from \emph{braiding} of the conformal block basis functions $\CobloF_T$, when viewed as functions on 
$\{(z_1,\ldots ,z_{2N}) \in \bC^{2N} \; | \; z_i \neq z_j \textnormal{ for } i \neq j \}$.

The \emph{braid group} $\mathfrak{B}_\Summed$ on $\Summed$ strands is generated by $b_i \in \mathfrak{B}_\Summed$ for $i\in\llbracket 1, \Summed-1\rrbracket$ with relations
\begin{equation*}
\begin{aligned}
b_ib_{i+1}b_i = \; & b_{i+1}b_ib_{i+1} , && \textnormal{for } i\in\llbracket 1, \Summed-2\rrbracket , \\
b_ib_j = \; & b_jb_i , && \textnormal{for } |j-i|>1.
\end{aligned}
\end{equation*}
$\mathfrak{B}_\Summed$ is isomorphic to 
the fundamental group (the first homotopy group)
of the complex quotient manifold $\mathcal{C}_\Summed := \{(z_1,\ldots ,z_\Summed) \in \bC^\Summed \; | \; z_i \neq z_j \textnormal{ for } i \neq j \} \backslash \SymGrp_\Summed$, where the symmetric group acts by permutation of coordinates 
(see, e.g.,~\cite[Rem.~2.3 in Sect.~XIX.2]{Kassel:Quantum_groups}).

The braid group $\mathfrak{B}_{2N}$ acts naturally on the conformal block functions $\CobloF_T$ by 
\begin{align} \label{eq:braid action}
b_k . \CobloF_T(\ldots ,z_k,z_{k+1},\ldots) =  \CobloF_T(\ldots ,z_{k+1},z_k,\ldots), \qquad k \in \llbracket 1, {2N}-1\rrbracket ,
\end{align}
where $z_k$ and $z_{k+1}$ are exchanged along a counterclockwise loop of the fundamental group. 

\begin{remark} \label{rem:remarkcomb}
Let us emphasize that the conformal block basis for $\SolSp\sub{1^{2N}}$ \emph{does not correspond to} the \quote{comb basis} frequently used in the literature (e.g.,~\cite{DMS:CFT,KKP:Conformal_blocks_q_combinatorics_and_quantum_group_symmetry}). A simple counterexample is the basis of conformal block functions for $N=2$: 
\begin{align*}
\CobloF_{\;\ytableausetup{smalltableaux}\ytableaushort{13,24}}
\; = \;\; & \Delta(\bs{x}_{1,2,3,4})^{-1/2} \, 
\mathcal P_{\;\ytableausetup{smalltableaux}\ytableaushort{12,34}}
\; = \; \sqrt{\frac{(x_3-x_1)(x_4-x_2)}{(x_2-x_1)(x_3-x_2)(x_4-x_1)(x_4-x_3)}} , \\
\CobloF_{\;\ytableausetup{smalltableaux}\ytableaushort{12,34}}
\; = \;\; & \Delta(\bs{x}_{1,2,3,4})^{-1/2} \, 
\mathcal P_{\;\ytableausetup{smalltableaux}\ytableaushort{13,24}}
\; = \; \sqrt{\frac{(x_2-x_1)(x_4-x_3)}{(x_3-x_1)(x_3-x_2)(x_4-x_1)(x_4-x_2)}}.
\end{align*}
Namely, each element of the comb basis is an eigenvector of $b_1 \in \mathfrak{B}_4$, while $\CobloF_{\;\ytableausetup{smalltableaux}\ytableaushort{13,24}}$ is not.
\end{remark}

Recall from Lemmas~\ref{lem:polytabloids}~\&~\ref{lem:isomorphism} that
the space $P\super{2^{N}}$ of Specht polynomials is a simple module of the symmetric group $\SymGrp_{2N}$, 
where permutations act on the variables 
$\bs{x}_{1,\ldots,2N} := (x_1,\ldots ,x_{2N})$. The action~\eqref{eq:braid action} of the braid group generators $b_k$ on $\SolSp\sub{1^{2N}}$ 
can be related to the action of the symmetric group generators $\tau_k = (k,k+1)$ (transpositions) on $P\super{2^{N}}$ as 
\begin{align} \label{eq:braid action in terms of sym}
b_k . \CobloF_T 
= - \ii \, 
\Delta(\bs{x}_{1,\ldots,2N})^{-1/2} \, \tau_k . \mathcal P_{T^t} , \qquad k \in\llbracket 1, {2N}-1\rrbracket.
\end{align}
In particular, this induces an action of the symmetric group $\SymGrp_{2N}$ (or, equivalently, of the Hecke algebra $\Hecke\sub{1^{2N}}(-1)$) on $\SolSp\sub{1^{2N}}$. 
The action of the generators $\tau_k$ is then given by 
\begin{align} \label{eq:tau action on coblo}
\tau_k . \CobloF_T
= - \ii \, b_k . \CobloF_T 
= \Delta(\bs{x}_{1,\ldots,2N})^{-1/2} \,
(-\tau_k) . \mathcal P_{T^t} , \qquad k \in\llbracket 1, {2N}-1\rrbracket .
\end{align}
(So the action of $\tau_k$ on $\CobloF_T$ is not just transposition of the $k$-th and $(k+1)$-st coordinates.)
Utilizing the involutive algebra automorphism $\omega$ of $\bC[\SymGrp_{2N}]$ defined in~\eqref{eq:defomega}, we have
\begin{align} \label{eq:actionofa}
a . \CobloF_T 
= \; & \Delta(\bs{x}_{1,\ldots,2N})^{-1/2} \,
\omega(a) . \mathcal P_{T^t} , \qquad a\in \bC[\SymGrp_{2N}] .
\end{align}

\begin{remark}
From Remark~\ref{rem:remarkomega}, we see that
the $\SymGrp_{2N}$-module $\SolSp\sub{1^{2N}}$ is isomorphic to the simple module $V\super{N,N}$.
Moreover, the conformal block basis $\{\CobloF_T \; | \; T \in \SYTof{N,N} \}$ of $\SolSp\sub{1^{2N}}$ is sent in this isomorphism to the basis of dual polytabloids of $V\super{N,N}$
--- see Remark~\ref{rem:dual polytabloids}. 
\end{remark}

We now proceed with the action on $\SolSp\sub{1^{2N}}$ 
of the Temperley-Lieb algebra. 

\begin{definition}
The \emph{Temperley-Lieb algebra} $\TL_{\Summed}(\nu)$ with 
fugacity $\nu := -q - q^{-1} \in \bC$ parameterized by 
$q \in \bC \setminus \{0\}$
is generated by $e_i \in \TL_{\Summed}(\nu)$ for $i\in\llbracket 1, \Summed-1\rrbracket$ with relations
\begin{equation} \label{eq:TLpres}
\begin{aligned}
e_i^2 = \; & \nu \, e_i , && \textnormal{for } i\in\llbracket 1, \Summed-1\rrbracket , \\
e_ie_{i+1}e_i = \; &  e_i , && \textnormal{for }  i\in\llbracket 1, \Summed-2\rrbracket , \\
e_ie_{i-1}e_i = \; &  e_i , && \textnormal{for } i\in\llbracket 2, \Summed-1\rrbracket, \\
e_ie_j = \; & e_je_i, && \textnormal{for } |j-i|>1.
\end{aligned}
\end{equation}
\end{definition}

The Temperley-Lieb algebra $\TL_{\Summed}(\nu)$ is isomorphic to a diagram  algebra~\cite{Kauffman:An_invariant_of_regular_isotopy} 
which, as a vector space, is generated by non-crossing planar tangles 
embedded in a rectangle connecting $2\Summed$ points lying on the boundary. 
In this presentation, there are exactly $\Summed$ points on the left boundary and $\Summed$ points on the right boundary: e.g., two elements of $\TL_{4}(\nu)$ are
\begin{align*}
\vcenter{\hbox{\includegraphics[scale=0.275]{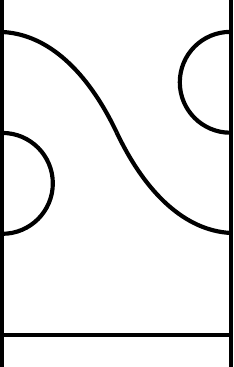}}} \qquad \qquad \textnormal{and} \qquad \qquad
\vcenter{\hbox{\includegraphics[scale=0.275]{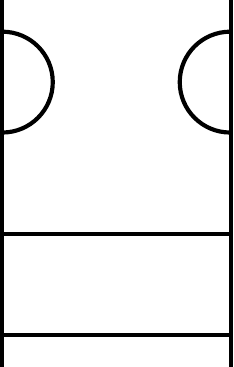}  .}}
\end{align*} 
Multiplication of diagrams is defined to be their concatenation with the additional rule that, whenever a loop is formed, it is removed and replaced by a scalar factor of $\nu = -q-q^{-1}$:
\begin{align*}
& \vcenter{\hbox{\includegraphics[scale=0.275]{figures/e-TLalgebra2-new-new.pdf}}} \quad 
\vcenter{\hbox{\includegraphics[scale=0.275]{figures/e-TLalgebra1-new.pdf}}} \quad := \quad 
\vcenter{\hbox{\includegraphics[scale=0.275]{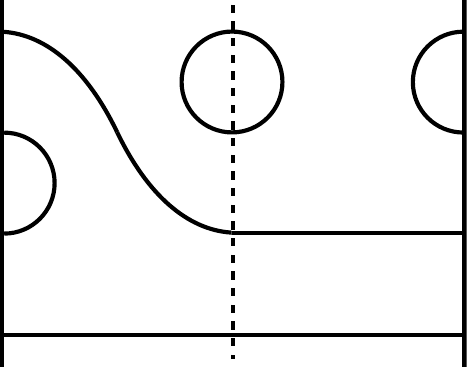}}} \quad = \quad \nu \,\, \times \,\, 
\vcenter{\hbox{\includegraphics[scale=0.275]{figures/e-TLalgebra2-new-new.pdf}  .}}
\end{align*}
The product is extended bilinearly to the whole algebra. The isomorphism between the algebra defined by the presentation~\eqref{eq:TLpres} and the diagram algebra is explicitly given by
\begin{align*}
e_i \quad = \quad \vcenter{\hbox{\includegraphics[scale=0.275]{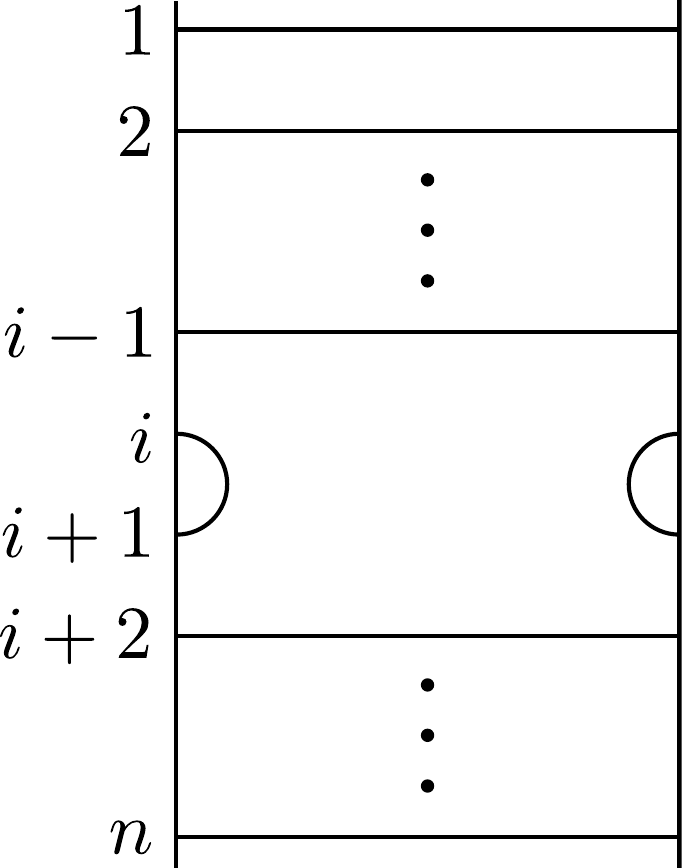} ,}} 
\qquad i\in\llbracket 1, \Summed-1\rrbracket ,
\end{align*}
and the unit of the algebra is given by the through-line diagram
\begin{align*}
1 \quad = \quad
\vcenter{\hbox{\includegraphics[scale=0.275]{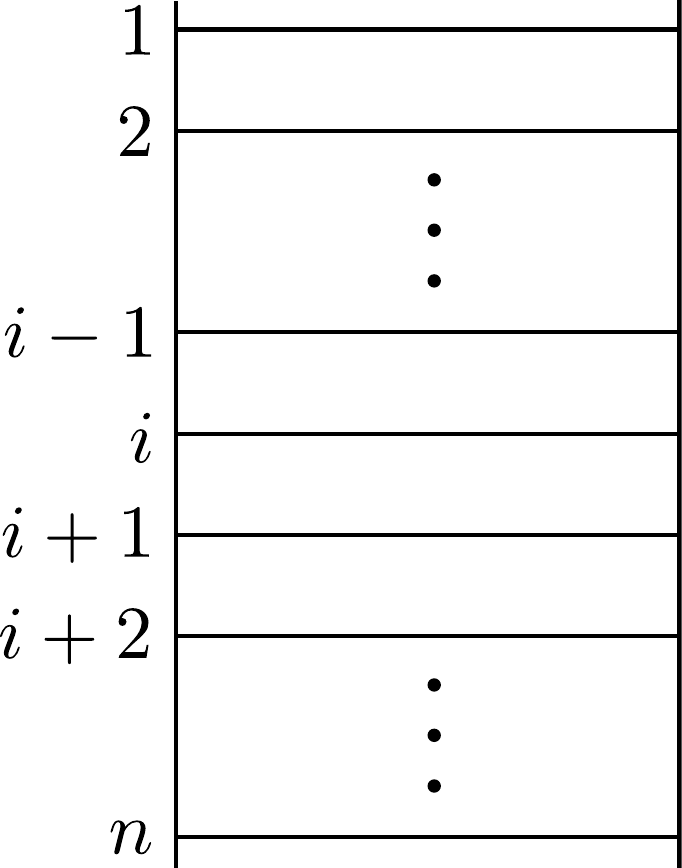}  .}}
\hphantom{\qquad i\in\llbracket 1, \Summed-1\rrbracket ,}
\end{align*}

When $q$ is not a root of unity, or when $q=\pm 1$, 
the algebra $\TL_{\Summed}(\nu)$ with $\nu = -q-q^{-1}$ 
is semisimple, with its simple modules given by the so-called \emph{standard modules} (cell modules)
\begin{align*}
\{\LS_{\Summed}\super{s} \; | \; s\in \{\Summed \textnormal{ mod }2,\Summed \textnormal{ mod }2+2,\ldots,\Summed \} .
\end{align*}
Elements in the standard module $\LS_{\Summed}\super{s}$
can be understood diagrammatically as non-crossing planar tangles embedded in a rectangle and connecting $\Summed+s$ points on the boundary, with $n$ points on the left boundary and $s$ points on the right boundary, and such that the $s$ points cannot be connected among each other. 
(See, e.g.,~\cite{Ridout-Saint-Aubin:Standard_modules_induction_and_structure_of_TL, Flores-Peltola:Standard_modules_radicals_and_the_valenced_TL_algebra} 
for a detailed account.)
The multiplication rule is then given by concatenation with the rules that a loop is replaced by a factor $\nu$ as before, 
and whenever the resulting diagram connects points on the right boundary, it is set to zero. 
Examples of the action of $\TL_{4}(\nu)$ on 
$\LS_{4}\super{2}$ are
\begin{align*}
\vcenter{\hbox{\includegraphics[scale=0.275]{figures/e-TLalgebra1-new.pdf}}} \quad 
\vcenter{\hbox{\includegraphics[scale=0.275]{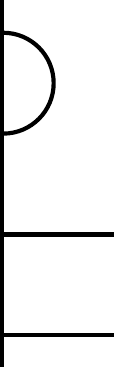}}} \quad
= \; & \quad \vcenter{\hbox{\includegraphics[scale=0.275]{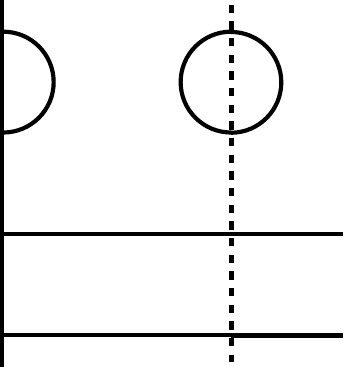}}}  \quad
 = \quad \nu \,\, \times \,\, 
\vcenter{\hbox{\includegraphics[scale=0.275]{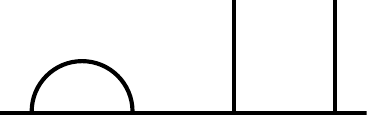} .}}
\\
\vcenter{\hbox{\includegraphics[scale=0.275]{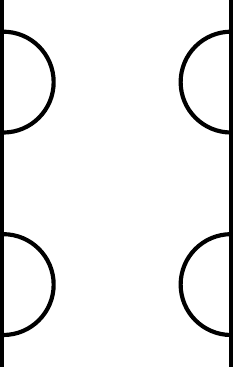}}} \quad 
\vcenter{\hbox{\includegraphics[scale=0.275]{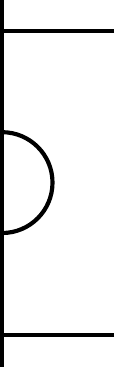}}} \quad
= \; & \quad \vcenter{\hbox{\includegraphics[scale=0.275]{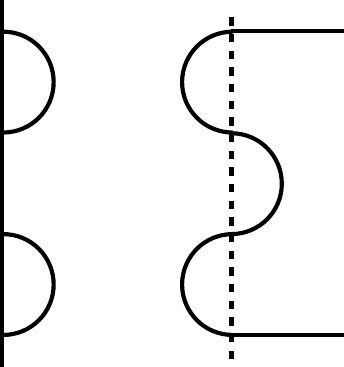}}} \quad
= \quad 0 \,\, \times \,\, 
\vcenter{\hbox{\includegraphics[scale=0.275]{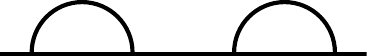}}} 
\quad = \quad 0 .
\end{align*}

\begin{remark} \label{rem:quotient_Hecke}
Setting $\tau_k=1-e_k$ for all $k$, 
the defining relations~\eqref{eq:TLpres} of $\TL_\Summed = \TL_\Summed(2)$ 
with $q=-1$ can be written in the form
\begin{equation*}
\begin{aligned}
\tau_i^2 = \; & 1 , && \textnormal{for } i\in\llbracket 1, \Summed-1\rrbracket , \\
\tau_i\tau_{i+1}\tau_i = \; & \tau_{i+1}\tau_i\tau_{i+1} , && \textnormal{for } i\in\llbracket 1, \Summed-2\rrbracket, \\
\tau_i\tau_j = \; & \tau_j\tau_i , && \textnormal{for } |j-i|>1,
\end{aligned}
\end{equation*}
together with the important relation 
\begin{align} \label{eq:TLrelation} 
1-\tau_i-\tau_{i+1}+\tau_i\tau_{i+1}+\tau_{i+1}\tau_i-\tau_i\tau_{i+1}\tau_i = \; &  0 , && \textnormal{for } 
i\in\llbracket 1, \Summed-2\rrbracket.
\end{align}
This makes it clear that $\TL_\Summed$ %with $q = -1$ 
is a (nontrivial) quotient of the group algebra $\bC[\SymGrp_\Summed]$ of the symmetric group (or equivalently, 
of the Hecke algebra $\Hecke\sub{1^{\Summed}}(-1)$). 
\end{remark}

\begin{proposition} \label{prop:propdescendtoTL}
The representation of $\bC[\SymGrp_{2N}]$ on $\SolSp\sub{1^{2N}}$ descends to a representation of $\TL_{2N}$. 
\end{proposition}

\begin{proof}
By Remark~\ref{rem:quotient_Hecke}, it suffices to verify the relation~\eqref{eq:TLrelation}. 
Denote by $\langle \tau_k,\tau_{k+1}\rangle \cong \SymGrp_3$ the subgroup of $\SymGrp_{2N}$ generated by the transpositions $\tau_k$ and $\tau_{k+1}$. 
By~\eqref{eq:actionofa}, we have
\begin{align} \label{eq:aux_rel}
\begin{split}
\sum_{\sigma \in \langle \tau_k,\tau_{k+1}\rangle} \sign(\sigma) \sigma . \CobloF_T 
= \; & \Delta(\bs{x}_{1,\ldots,2N})^{-1/2} \,
\sum_{\sigma \in \langle \tau_k,\tau_{k+1}\rangle}
\sign(\sigma) \omega(\sigma) . \mathcal P_{T^t} \\
= \; & \Delta(\bs{x}_{1,\ldots,2N})^{-1/2} \,
\sum_{\sigma \in \langle \tau_k,\tau_{k+1}\rangle}
\sigma . \mathcal P_{T^t} ,
    \qquad k\in \llbracket 1, 2N-2 \rrbracket .
\end{split}
\end{align}
Since $T^t$ is a Young tableau with two columns, at least two entries among $k,k+1,k+2$ lie on the same column. 
We thus infer that $\mathcal P_{T^t}$ is antisymmetric in at least two variables among $(x_k,x_{k+1},x_{k+2})$. 
Hence, the symmetrization of $\mathcal P_{T^t}$ with respect to $(x_k,x_{k+1},x_{k+2})$ gives zero, 
which together with~\eqref{eq:aux_rel} implies that the relation~\eqref{eq:TLrelation} is satisfied. 
\end{proof}

\begin{corollary} \label{cor:TL_rep}
The $\TL_{2N}$-module $\SolSp\sub{1^{2N}}$ is isomorphic to the standard module $\LS_{2N}\super{0}$.
\end{corollary}

\begin{proof}
$\SolSp\sub{1^{2N}}$ is isomorphic to $V\super{N,N}$ as a $\bC[\SymGrp_{2N}]$-module. 
By~\cite[Lem.~4.2]{PPR:Promotion_and_cyclic_sieving_via_webs}, the latter is isomorphic to $\LS_{2N}\super{0}$. 
Proposition~\ref{prop:propdescendtoTL} shows that these representations descend to the quotient $\TL_{2N}$ (cf.~Remark~\ref{rem:quotient_Hecke}), which proves the claim. 
\end{proof}

\begin{remark} 
By virtue of Proposition~\ref{prop:propdescendtoTL}, 
throughout this section we will often employ loose notations --- e.g., we identify $\tau_k\in \bC [\SymGrp_{2N}]$ and $\tau_k=1-e_k \in \TL_{2N}$ when acting on $\SolSp\sub{1^{2N}}$.
Note however that, when acting on functions in $\SolSp\sub{1^{2N}}$, 
the action~\eqref{eq:tau action on coblo} of the element $\tau_k \in \TL_{2N}$ 
is \emph{not a permutation of the variables}:
for each $k \in\llbracket 1, {2N}-1\rrbracket$, we have
\begin{align*}
\tau_k . \CobloF_T(x_{1},\ldots ,x_{2N})
= \; & \Delta(x_{1},\ldots ,x_{2N})^{-1/2} \,
(-\tau_k) . \mathcal P_{T^t}(x_{1},\ldots,x_{k},x_{k+1},\ldots ,x_{2N}) \\
= \; & - \Delta(x_{1},\ldots ,x_{2N})^{-1/2} \,
\mathcal P_{T^t}(x_{1},\ldots,x_{k+1},x_{k},\ldots ,x_{2N}) .
\end{align*} 
\end{remark} 

\begin{remark} 
With this identification, 
the relation $b_k = -\ii \tau_k = -\ii (1-e_k)$ 
in the action~\eqref{eq:braid action in terms of sym} of $\TL_\Summed$ 
corresponds to the familiar \quote{skein relation} for the Kauffman bracket polynomial~\cite{Kauffman:State_models_and_the_Jones_polynomial}, 
with deformation parameter $q=-1$ (and fugacity $\nu = 2$).
\end{remark}

\subsection{Conformal blocks for general valences $\multii$}
\label{subsec: fused blocks}

We now construct the spaces of $c=1$ conformal blocks, denoted $\SolSp_\multii$, for any valences $\multii=(s_1,\ldots,s_\np) \in \bZpos^\np$, and show that they carry representations of \quote{fused} versions of the Temperley-Lieb algebra,
called \emph{valenced Temperley-Lieb algebras}~\cite{Flores-Peltola:Standard_modules_radicals_and_the_valenced_TL_algebra, Flores-Peltola:Higher_spin_QSW}.
We begin with the definition of the valenced Temperley-Lieb algebra, which also gives systematic tools to carry out the fusion of the conformal blocks.

\begin{definition} 
The \emph{Jones-Wenzl idempotents}~\cite{Wenzl:On_sequences_of_projections} 
in the Temperley-Lieb algebra $\TL_\Summed(\nu)$ are
nonzero elements $\textnormal{JW}_{i,j} \neq 0$ for $i,j \in \llbracket 1,\Summed\rrbracket$ with $i<j$, defined recursively via
\begin{align*}
\textnormal{JW}_{i,j}\textnormal{JW}_{i,j} = \; & \textnormal{JW}_{i,j}, \\
e_k\;\textnormal{JW}_{i,j} = \; & \textnormal{JW}_{i,j}\;e_k=0 , \qquad\textnormal{for all } k\in \llbracket i,j-1\rrbracket. 
\end{align*}
In the case $q=-1$ (and $\nu=2$), the Jones-Wenzl idempotents are given by the \emph{symmetrizers}:
\begin{align*}
\textnormal{JW}_{i,j} = \frac{1}{(j-i+1)!} 
\sum_{\sigma \in \langle \tau_i,\tau_{i+1},\ldots ,\tau_{j-1}\rangle} \sigma 
\end{align*}
(or rather, their images under the quotient map in Remark~\ref{rem:quotient_Hecke}).
\end{definition}

Consider the $\multii$-\emph{symmetrizer} idempotent $\sym$ of $\bC [\mathfrak{S}_\Summed]$ defined in~\eqref{eq:defsymmetrizer} and denote by $\symTL$ the corresponding image in the $\TL_\Summed$ quotient (cf.~Remark~\ref{rem:quotient_Hecke}). 
Then, we have
\begin{align*}
\symTL = \prod_{k=1}^\np \textnormal{JW}_{\summ_k,\summ_{k+1}-1} ,
\qquad
\textnormal{where}
\qquad
\summ_k := 1+ \sum_{j=1}^{k-1} s_j, \qquad k \in \llbracket 1, \np\rrbracket .
\end{align*} 

To define the fused conformal blocks,  
we use the notation~\eqref{eq: eval notation} for $[f]_\textnormal{eval}$. 
We will show that for each $f \in  \SolSp\sub{1^{2N}}$, 
the evaluations 
$x_{\summ_k}=x_{\summ_k+1}=\cdots=x_{\summ_{k+1}-1}$ for all $k \in \llbracket 1, \np\rrbracket$ of
\begin{align*}
\frac{\symTL. f(x_1,\ldots ,x_{2N})}{\sqrt{\prod_{k=1}^\np \prod_{\summ_k\leq i<j<\summ_{k+1}} (x_j-x_i)}}
\end{align*} 
yields a finite value. This gives the following result.

\begin{proposition}\label{prop:propdefCvarsigma}
The following space of functions is well-defined:
\begin{align} \label{eq:defCvarsigma}
\SolSp_{\multii} := \bigg\{F \colon \chamber_\np \to \bR \; \bigg | \; F 
= \left[ \frac{\symTL.f}{\sqrt{\prod_{k=1}^\np \prod_{\summ_k\leq i<j<\summ_{k+1}} (x_j-x_i)}} \right]_{\textnormal{eval}} , \; f \in \SolSp\sub{1^{2N}}\bigg\} .
\end{align}
\end{proposition}

\begin{proof}
By Lemma~\ref{prop:lemma3p3}, for every $f \in \SolSp\sub{1^{2N}}$ there is a polynomial $\mathcal P \in P^{(2^{N})}$ such that
\begin{align*}
f = \Delta(\bs{x}_{1,\ldots,2N})^{-1/2} \, \mathcal P .
\end{align*}
Thus, the claim follows by noting that
\begin{align} \label{eq:actionObecomesactionp}
\symTL . f = \Delta(\bs{x}_{1,\ldots,2N})^{-1/2} \,\omega(\sym) . \mathcal P 
= \Delta(\bs{x}_{1,\ldots,2N})^{-1/2} \, \idpt . \mathcal P,
\end{align}
and $\idpt . \mathcal P$ is divisible by $\prod_{k=1}^\np \prod_{\summ_k\leq i<j<\summ_{k+1}} (x_j-x_i)$.
\end{proof}

Using Proposition~\ref{prop:propdefCvarsigma}, we
define a family $\{\CobloF_T \; | \; T \in \CSYTof{N,N} \}$, 
where $\CSYT$ is the set of column-strict Young tableaux of shape $(N,N)$ and content $\multii$. We set
\begin{align} \label{eq: sol_space general}
\SolSp_\multii := \Spn_{\bC} \{\CobloF_T \; | \; T \in \CSYTof{N,N} \} .
\end{align}
In Lemma~\ref{lem:lemmaUintermsoffusedSpecht}, we write $\CobloF_T$ in the form 
$\CobloF_T = \mathcal K \mathcal F_{T^t}$ where 
$\mathcal K$ is a normalization factor independent of $T$, 
and $\mathcal F_{T^t}$ is the fused Specht polynomial (Definition~\ref{def:fusedspecht}). 
As the collection 
$\{\mathcal F_{T^t} \; | \; T \in \CSYTof{N,N} \}$ is linearly independent by Proposition~\ref{prop:proplinearindepfusedspecht}, 
this implies that the collection
$\{\CobloF_T \; | \; T \in \CSYTof{N,N} \}$ is linearly independent and thus forms a basis for $\SolSp_\multii$.

\begin{definition} \label{def:firstdeffusedCBs}
For each $T \in \CSYTof{N,N}$, we define 
the \emph{conformal block basis function} as
\begin{align} \label{eq:infirstdeffusedCBs}
\CobloF_T(x_1,\ldots ,x_\np) := 
\left[ \frac{\symTL . \CobloF_{\hat{T}}}{\sqrt{\prod_{k=1}^\np \prod_{\summ_k\leq i<j<\summ_{k+1}} (x_j-x_i)}} \right]_{\textnormal{eval}} ,
\end{align}
where $\hat{T} := (\tilde{F})^t$ with $F = T^t \in \RSYTof{2^N}$ being the transpose 
of $T \in \CSYTof{N,N}$, and $\tilde{F} \in \SYT$ constructed as in Definition~\ref{def:tildeT} and Lemma~\ref{lem:tildeT}.
\end{definition}

This definition is motivated by fusion in CFT: the left-hand side in~\eqref{eq:infirstdeffusedCBs} should be a correlation function of CFT fields obtained from fusion of fields with Kac type conformal weights $h_{1,2} = 1/4$ as in~\eqref{eq:conf_weights} --- and the correlation functions of the latter are given by the functions in the solution space $\SolSp\sub{1^{2N}}$ in~\eqref{eq:defspaceS}~\cite{Flores-Kleban:Solution_space_for_system_of_null-state_PDE1}. 
Now, Lemma~\ref{prop:lemma3p3} implies that the conformal block functions appearing on the right-hand side in~\eqref{eq:infirstdeffusedCBs} form a basis for this space, and can hence be thought of as conformal blocks. 
Finally, the evaluation operation on the right-hand side in~\eqref{eq:infirstdeffusedCBs} is nothing but a fusion with the appropriate fusion channels, 
to obtain CFT fields with more general Kac type conformal weights in~\eqref{eq:conf_weights} labeled by the valences $\multii$. 
We shall make this heuristics precise in the course of the rest of this article.

Let us begin by observing that the (fused) conformal block basis functions can be written even more explicitly using the fused Specht polynomials $\mathcal F_F$ from Definition~\ref{def:fusedspecht}.
Recall that the latter also have an explicit combinatorial formula~\eqref{eq:combinatorialformula} obtained in Proposition~\ref{prop:combinatorialformula}.

\begin{lemma} \label{lem:lemmaUintermsoffusedSpecht}
Fix $\multii=(s_1,\ldots,s_\np) \in \bZpos^\np$ with $s_1+\cdots+s_\np = 2N$. 
Let $T \in \CSYTof{N,N}$, and let $T^t \in \RSYTof{2^{N}}$ be its transpose. 
The conformal block function $\CobloF_T$ then reads
\begin{align} \label{eq:defUvarsigmaT}
\CobloF_T (x_1,\ldots ,x_\np) 
=  \prod_{1\leq i < j \leq \np} (x_j-x_i)^{-\frac{s_i s_j}2}\times \mathcal F_{T^t}(x_1,\ldots ,x_\np)  .
\end{align}
\end{lemma}

\begin{proof}
From Definition~\ref{def:firstdeffusedCBs} (denoting $\hat{T} := (\tilde{F})^t$ with $F = T^t$) and~\eqref{eq:actionObecomesactionp}, we obtain 
\begin{align*} 
\CobloF_T (x_1,\ldots ,x_\np) 
= \; & 
\left[ \frac{\symTL . \CobloF_{\hat{T}}}{\sqrt{\prod_{k=1}^\np \prod_{\summ_k\leq i<j<\summ_{k+1}} (x_j-x_i)}}\right]_{\textnormal{eval}} \\
= \; & \left[ \frac{\idpt . \mathcal P_{\widetilde{T^t}}}{\prod_{k=1}^\np \prod_{\summ_k\leq i<j<\summ_{k+1}} (x_j-x_i)} \frac{\sqrt{\prod_{k=1}^\np \prod_{\summ_k\leq i<j<\summ_{k+1}} (x_j-x_i)}}{\sqrt{\prod_{1\leq i < j \leq 2N} (x_j-x_i)}}\right]_{\textnormal{eval}} 
\end{align*}
These two ratios have a well-defined evaluation. Indeed, it follows from Definition~\ref{def:fusedspecht} that the first fraction gives $\mathcal F_{T^t}$. 
Moreover, by Proposition~\ref{prop:fusedspechtonecolumn} we have 
\begin{align*}
\left[ \frac{\sqrt{\prod_{k=1}^\np \prod_{\summ_k\leq i<j<\summ_{k+1}} (x_j-x_i)}}{\sqrt{\prod_{1\leq i < j \leq 2N} (x_j-x_i)}}\right]_{\textnormal{eval}}
= \prod_{1\leq i < j \leq \np} (x_j-x_i)^{-\frac{s_i s_j}2} .
\end{align*}
This gives the asserted identity~\eqref{eq:defUvarsigmaT}.
\end{proof}

\begin{proposition} \label{prop:basisformathcalC}
The collection $\{\CobloF_T \; | \; T \in \CSYTof{N,N} \}$ is a basis for $\SolSp_\multii$ in~\eqref{eq: sol_space general}.
\end{proposition}

\begin{proof}
Lemma~\ref{lem:lemmaUintermsoffusedSpecht} gives an explicit expression of each conformal block function $\CobloF_T$ in terms of fused Specht polynomials. 
The claim thus follows from  Proposition~\ref{prop:proplinearindepfusedspecht}.
\end{proof}

\begin{remark}
If $\multii=(1,1,\ldots,1)=1^\Summed$ with $\np = \Summed = 2N$, then $T \in \SYTof{N,N}$. 
Moreover, by Remark~\ref{remarkfusedspechtforvalence1} the fused Specht polynomial $\mathcal F_{T^t}$ then becomes the classical Specht polynomial $\mathcal P_{T^t}$. Hence $\CobloF_T$ reduces to the conformal block function~\eqref{eq:CBsh12} of~\cite{Peltola-Wu:Global_and_local_multiple_SLEs_and_connection_probabilities_for_level_lines_of_GFF}.
\end{remark}

\begin{remark} \label{rem:Dyck path gen2}
Let us check that our functions match with~\cite[Eq.~(5.15,~5.16,~5.17)]{Liu-Wu:Scaling_limits_of_crossing_probabilities_in_metric_graph_GFF} 
as discussed in Remark~\ref{rem:Dyck path gen}.
These functions were shown to have an important interpretation for GFF level sets with height gaps $\pm 2\pi$.
Analogously, we expect that our more general functions 
play the same role for general GFF level sets of the type considered in~\cite{ALS:First_passage_sets_of_2D_GFF} with more general height gaps,
and we plan to return to this in future work.

Consider the conformal block functions in $\SolSp\sub{2,2,2,2}$. 
We have $\dim\SolSp\sub{2,2,2,2} = 3$,  because there are three column-strict Young tableaux with this set of contents:
\begin{align*}
T_1 \; = \; 
{\ytableausetup{nosmalltableaux}\begin{ytableau} 
1 & 1 & 2 & 2 \\
3 & 3 & 4 & 4
\end{ytableau}}
, \qquad
T_2 \; = \; 
{\ytableausetup{nosmalltableaux}\begin{ytableau} 
1 & 1 & 3 & 3 \\
2 & 2 & 4 & 4
\end{ytableau}}
, \qquad
T_3 \; = \; 
{\ytableausetup{nosmalltableaux}\begin{ytableau} 
1 & 1 & 2 & 3 \\
2 & 3 & 4 & 4 
\end{ytableau}}.
\end{align*}
Utilizing Lemma~\ref{lem:lemmaUintermsoffusedSpecht}, 
we compute the conformal block functions
$\{ \CobloF_{T_1}, \CobloF_{T_2}, \CobloF_{T_3} \}$ spanning  $\SolSp\sub{2,2,2,2}$. 
First of all, the fused Specht polynomials $\mathcal F_{T_1^t}$ and $\mathcal F_{T_2^t}$ are immediately 
computed, because any given entry appears only in one column (see Proposition~\ref{prop:fusedspechtonecolumn}). 
This gives
\begin{align*}
\mathcal F_{T_1^t} = (x_2-x_1)^4 (x_4-x_3)^4 
\qquad 
\textnormal{and} 
\qquad 
\mathcal F_{T_2^t} = (x_3-x_1)^4 (x_4-x_2)^4 .
\end{align*}
Therefore, thanks to Lemma~\ref{lem:lemmaUintermsoffusedSpecht} we immediately obtain
\begin{align*}
\CobloF_{T_1} = \; & \bigg(\frac{(x_2-x_1)(x_4-x_3)}{(x_3-x_1)(x_3-x_2)(x_4-x_1)(x_4-x_2)}\bigg)^2
\\
\CobloF_{T_2} = \; & \bigg(\frac{(x_3-x_1)(x_4-x_2)}{(x_2-x_1)(x_3-x_2)(x_4-x_1)(x_4-x_3)}\bigg)^2. 
\end{align*}
Hence, we readily see that $\CobloF_{T_1}$ and $\CobloF_{T_2}$ correspond to~\cite[Eq.~(5.16,~5.15)]{Liu-Wu:Scaling_limits_of_crossing_probabilities_in_metric_graph_GFF}. 

The third conformal block function $\CobloF_{T_3}$ is slightly more intricate, because the entries \quote{$2$} and \quote{$3$} appear in two different columns in $T_3^t$. 
Thus, consider $\hat{T}_3 \in \SYTof{4,4}$, with transpose
\begin{align*}
\widetilde{T_3^t} \; = \; 
{\ytableausetup{nosmalltableaux}\begin{ytableau} 
    1 & 4 \\
    2 & 6 \\
    3 & 7 \\
    5 & 8
\end{ytableau}}
\end{align*}
as in Definition~\ref{def:tildeT}.
From Definition~\ref{def:fusedspecht} and Lemma~\ref{lem:lemmaUintermsoffusedSpecht}, we obtain 
\begin{align}\label{eq:T3}
\frac{\CobloF_{T_3}(x_1,x_2,x_3,x_4)}{\prod_{1 \leq i < j \leq 4}(x_j-x_i)^{-2}} = 
\left[ \frac{\idptof{2,2,2,2} . \mathcal P_{\widetilde{T_3^t}}(x_1,\ldots,x_8)}{(x_2-x_1)(x_4-x_3)(x_6-x_5)(x_8-x_7)} \right]_{\substack{x_1, \, x_2  \, \mapsto x_1, \\ x_3, \, x_4  \, \mapsto x_2 \\ x_5,  \, x_6  \, \mapsto x_3, \\ x_7,  \, x_8  \, \mapsto x_4}} ,
\end{align}
where $\idptof{2,2,2,2}$ acts by antisymmetrizing with respect to the sets of variables $\{x_1,x_2\}$, $\{x_3,x_4\}$, $\{x_5,x_6\}$, and $\{x_7,x_8\}$. 
Note that this formula slightly simplifies because the Specht polynomial $\smash{\mathcal P_{\widetilde{T_3^t}}}$ is by definition antisymmetric with respect to $\{x_1,x_2\}$ and $\{x_7,x_8\}$.  
Thus, we have 
\begin{align*}
\textnormal{\eqref{eq:T3}} = 
\left[\frac{\mathcal P_{\widetilde{T_3^t}} \, - \, \mathcal P_{\widetilde{T_3^t}}(x_4 \leftrightarrow x_3) \, - \, \mathcal P_{\widetilde{T_3^t}}(x_6 \leftrightarrow x_5) \, + \, \mathcal P_{\widetilde{T_3^t}}(x_4 \leftrightarrow x_3 , x_6 \leftrightarrow x_5)}{4(x_2-x_1)(x_4-x_3)(x_6-x_5)(x_8-x_7)} \right]_{\substack{x_1, \, x_2 \, \mapsto x_1, \\ x_3, \, x_4  \,\mapsto x_2 \\ x_5, \, x_6  \, \mapsto x_3, \\ x_7, \, x_8 \, \mapsto x_4}}
\end{align*}
(denoting by $x_i \leftrightarrow x_j$ the interchange of the variables $x_i$ and $x_j$). 
An explicit computation then finally leads to
\begin{align*}
\CobloF_{T_3} = \frac{1}{2(x_4-x_1)^2 (x_3-x_2)^2} + \frac{1}{4(x_2-x_1)(x_3-x_1)(x_4-x_2)(x_4-x_3)},
\end{align*}
which is equivalent to~\cite[Eq.~(5.17)]{Liu-Wu:Scaling_limits_of_crossing_probabilities_in_metric_graph_GFF} up to a factor of $4$. 
\end{remark}

\subsection{Valenced Temperley-Lieb action}
\label{subsec: Valenced TL action}

We will next consider the \emph{valenced Temperley-Lieb algebra} $\TL_\multii = \TL_\multii(2)$ 
with fugacity $\nu = -q - q^{-1} = 2$, i.e., with deformation parameter $q=-1$. 
It is isomorphic to a diagram algebra of valenced tangles~\cite{Flores-Peltola:Standard_modules_radicals_and_the_valenced_TL_algebra} 
(which we will not, however, use in the present work).  
It is conveniently defined as the associative algebra 
\begin{align*}
\TL_\multii := \TL_\multii(2) := \symTL \TL_{2N}(2) \symTL
\end{align*} 
with unit $\symTL$.
Moreover, by Lemma~\ref{lem:lemmaPAP}, as the Temperley-Lieb algebra $\TL_{2N} = \TL_{2N}(2)$ is semisimple, 
so is $\TL_\multii$, and 
its simple modules are given by $\symTL(\LS_{2N}\super{s})$ whenever nontrivial\footnote{See~\cite{Flores-Peltola:Standard_modules_radicals_and_the_valenced_TL_algebra,Flores-Peltola_Generators_projectors_and_the_JW_algebra} for a thorough study of this algebra.}. 
Let us lastly note that $\TL_\multii$ is also isomorphic to a quotient of the fused Hecke algebras $\Hecke_\multii := \Hecke_\multii(-1) = \idpt \bC[\SymGrp_{2N}]\idpt \cong \sym \bC[\SymGrp_{2N}] \sym =: \Hecke_\multii(1)$, discussed in Section~\ref{subsec:fused Hecke algebra}
(cf.~\cite{Crampe-Poulain-d-Andecy:Fused_braids_and_centralisers_of_tensor_representations_of_Uq_gln}). 

\smallskip

Proposition~\ref{prop:basisformathcalC} implies the following result, which is an analog of Theorem~\ref{thm:theoremA}:

\begin{proposition} \label{prop:repvTL}
Fix $\multii=(s_1,\ldots,s_\np) \in \bZpos^\np$ with $s_1+\cdots+s_\np = 2N$. 
The map 
\begin{align} \label{eq:iso}
\symTL .f  \qquad \mapsto \quad  
\left[ \frac{\symTL . f}{\sqrt{\prod_{k=1}^\np \prod_{\summ_k\leq i<j<\summ_{k+1}} (x_j-x_i)}}\right]_{\textnormal{eval}} , \qquad f \in \SolSp\sub{1^{2N}} ,
\end{align}
is a linear isomorphism from $\symTL(\SolSp\sub{1^{2N}})$ to $\SolSp_\multii$, and it induces an isomorphism of $\TL_\multii$-modules as 
\begin{align} \label{eq:valTLaction}
(\symTL a \symTL) . F = \left[ \frac{(\symTL a \symTL) . f}{\sqrt{\prod_{k=1}^\np \prod_{\summ_k\leq i<j<\summ_{k+1}} (x_j-x_i)}}\right]_{\textnormal{eval}},\quad a\in \textnormal{TL}_{2N}(2) ,
\end{align} 
where $f \in \SolSp\sub{1^{2N}}$ is chosen such that 
\begin{align*}
F = \left[\frac{\symTL . f}{\sqrt{\prod_{k=1}^\np \prod_{\summ_k\leq i<j<\summ_{k+1}} (x_j-x_i)}}\right]_{\textnormal{eval}} 
\; \in \; \SolSp_{\multii}.
\end{align*}  
Moreover $\SolSp_{\multii}$ is isomorphic to the simple module $\symTL(\LS_{2N}\super{0})$.
\end{proposition}

\begin{proof}
Recall that in~\eqref{eq:defCvarsigma}, 
$\SolSp_{\multii}$ is defined via the map~\eqref{eq:iso}, and
by Proposition~\ref{prop:basisformathcalC}, we have
$\dim\SolSp_\multii = |\CSYTof{N,N}|$. 
Also, Corollary~\ref{cor:TL_rep} (and its proof) 
shows that $\symTL(\SolSp\sub{1^{2N}})$ is isomorphic to the $\TL_\multii$-module $\symTL(\LS_{2N}\super{0})$ and
$\dim\symTL(\LS_{2N}\super{0}) = \dim\sym(V\super{N,N})$.
The claim now follows, since 
from the proof of Proposition~\ref{prop:linind1}, 
we obtain $\dim\sym(V\super{N,N}) = |\CSYTof{N,N}|$.
\end{proof}

\subsection{Covariance properties}
\label{subsec: cov}

The purpose of this section is to verify that the conformal block functions satisfy the M\"obius covariance of the primary fields with Kac weights~\eqref{eq:conf_weights}. 

\begin{proposition} \label{prop:wardidentity}
Let $\varphi \colon \bH \to \bH$ be a M\"obius transformation 
such that $\varphi(x_1)<\cdots<\varphi(x_\np)$. 
The conformal block functions $\CobloF_T$ satisfy the covariance property
\begin{align*}
\CobloF_T(\varphi(x_1),\ldots ,\varphi(x_\np)) 
=  \prod_{i=1}^\np |\varphi'(x_i)|^{-s_i^2 / 4} \times \CobloF_T(x_1,\ldots ,x_\np).
\end{align*}
\end{proposition}

\begin{proof}
Applying $x_i \mapsto \varphi(x_i)$ for all $i=1,\ldots ,\np$ in~\eqref{eq:defUvarsigmaT} of Lemma~\ref{lem:lemmaUintermsoffusedSpecht}, we obtain
\begin{align*} 
\CobloF_T (\varphi(x_1),\ldots ,\varphi(x_\np))
=  \prod_{1\leq i < j \leq \np} (\varphi(x_j)-\varphi(x_i))^{-s_i s_j / 2}\times \mathcal F_{T^t}(\varphi(x_1),\ldots ,\varphi(x_\np))  .
\end{align*}
The claim now follows from Lemmas~\ref{lem:lemmamobius1} and~\ref{lem:lemmamobius2}, proven below.
\end{proof}

\begin{lemma} \label{lem:lemmamobius1}
Fix $\multii=(s_1,\ldots,s_\np) \in \bZpos^\np$ with $s_1+\cdots+s_\np = 2N$. 
Let $\varphi \colon \bH \to \bH$ be a M\"obius transformation 
such that $\varphi(x_1)<\cdots<\varphi(x_\np)$. 
We have
\begin{align*}
\prod_{1\leq i < j \leq \np} (\varphi(x_j)-\varphi(x_i))^{-s_i s_j / 2}
= \prod_{i=1}^\np \varphi'(x_i)^{s_i (s_i-2N) /4} \times 
\prod_{1\leq i < j \leq \np} (x_j-x_i)^{-s_i s_j / 2} .
\end{align*}
\end{lemma}

\begin{proof}
This can be directly verified by utilizing the identity
\begin{align} \label{eq:identitymobius}
\varphi(x)-\varphi(y) = (x-y) \sqrt{\varphi'(x) \varphi'(y)} ,
\end{align}
satisfied by all M\"obius transformations $\varphi$, 
combined with the identity $2N = \sum_{j=1}^\np s_j$. 
\end{proof}

\begin{lemma} \label{lem:lemmamobius2}
Fix $\multii=(s_1,\ldots,s_\np) \in \bZpos^\np$ with $s_1+\cdots+s_\np = 2N$. 
We have
\begin{align*}
\mathcal F_{T^t}(\varphi(x_1),\ldots ,\varphi(x_\np)) 
= \prod_{i=1}^\np \varphi'(x_i)^{s_i(N-1)/2 - s_i(s_i-1)/2} \times \mathcal F_{T^t}(x_1,\ldots ,x_\np) .
\end{align*}
\end{lemma}

\begin{proof}
Using~\eqref{eq:CBsh12} from Definition~\ref{def:CBsh12def} 
and Equations~(\ref{eq:covariancepropertyvalence1},~\ref{eq:identitymobius}) for $\CobloF_T$, 
we obtain
\begin{align} \label{eq:intermediatelemma2}
\begin{split} 
\; & \frac{\idpt . \mathcal P_{\widetilde{T^t}}(\varphi(x_1),\ldots ,\varphi(x_{2N}))}{\prod_{k=1}^\np \prod_{\summ_k\leq i<j<\summ_{k+1}} (\varphi(x_j)-\varphi(x_i))} 
\\
= \; & \frac{\idpt . \Big( \prod_{1\leq i < j \leq {2N}} \varphi'(x_i)^{1/4} \varphi'(x_j)^{1/4}
\times 
\prod_{i=1}^{2N} \varphi'(x_i)^{-1/4} \times \mathcal P_{\widetilde{T^t}}(x_1,\ldots ,x_{2N})\Big)}{\prod_{k=1}^\np \prod_{\summ_k\leq i<j<\summ_{k+1}} (x_j-x_i) \varphi'(x_i) ^{1/2}\varphi'(x_j)^{1/2}}. 
\end{split}
\end{align}
Now, straightforward computations similar to those in the proof of Lemma~\ref{lem:lemmamobius1} lead to 
\begin{align} 
\label{eq:identity1Lemma2} 
\prod_{1\leq i < j \leq 2N} \varphi'(x_i)^{1/4}  \varphi'(x_j)^{1/4} 
= \; & \prod_{i=1}^{2N} \varphi'(x_i)^{(2N-1)/4} , \\
\label{eq:identity2Lemma2}
\prod_{\summ_k\leq i<j<\summ_{k+1}} \varphi'(x_i)^{1/2}  \varphi'(x_j)^{1/2} 
= \; & \prod_{i=\summ_k}^{\summ_{k+1}-1} \varphi'(x_i)^{(s_k-1)/2} 
\end{align}
(recall here that $\summ_{k+1}-\summ_k=s_k$).
Substituting~(\ref{eq:identity1Lemma2},~\ref{eq:identity2Lemma2})  into~\eqref{eq:intermediatelemma2} yields
\begin{align*}
\left[\textnormal{\eqref{eq:intermediatelemma2}}\right]_{\textnormal{eval}} 
= \left[\frac{\idpt . \big( \prod_{j=1}^{2N} \varphi'(x_j)^{(N-1)/2} \times 
\mathcal P_{\widetilde{T^t}}(x_1,\ldots ,x_{2N})\big)}{\prod_{k=1}^\np \prod_{i=\summ_k}^{\summ_{k+1}-1} \varphi'(x_i)^{(s_k-1)/2}\times\prod_{k=1}^\np \prod_{\summ_k\leq i<j<\summ_{k+1}} (x_j-x_i)} \right]_\textnormal{eval}.
\end{align*}
Here, since the product $\prod_{j=1}^{2N} \varphi'(x_j)^{(N-1)/2}$ is symmetric in the $x_j$ variables, we may also take it out of the antisymmetrizer $\idpt$. 
Therefore, we infer from Definition~\ref{def:fusedspecht} that
\begin{align*}
\mathcal F_{T^t}(\varphi(x_1),\ldots ,\varphi(x_\np)) 
= \; & \left[\textnormal{\eqref{eq:intermediatelemma2}}\right]_{\textnormal{eval}} 
= \left[\frac{\prod_{j=1}^{2N} \varphi'(x_j)^{(N-1)/2}}{\prod_{k=1}^\np \prod_{i=\summ_k}^{\summ_{k+1}-1} \varphi'(x_i)^{(s_k-1)/2}} \right]_\textnormal{eval} 
\times \mathcal F_{T^t}(x_1,\ldots ,x_\np) 
\\
= \; & \prod_{k=1}^\np \varphi'(x_k)^{s_k(N-1)/2 - s_k(s_k-1)/2} 
\times \mathcal F_{T^t}(x_1,\ldots ,x_\np) . 
\end{align*}
This completes the proof. 
\end{proof}

\subsection{BPZ partial differential equations}

In this section, we consider a system of BPZ partial differential equations for the conformal block functions. 
To write them explicitly, let 
\begin{align} \label{eq:deffirstorderdiffop}
\mathcal L_m\super{j} := -\sum_{\substack{1\leq i \leq d\\i\neq j}} \bigg( (x_i-x_j)^{1+m} \frac{\partial}{\partial x_i} +  \frac{1+m}{4} \, s_i^2  (x_i-x_j)^m\bigg) , \qquad m \in \bZ ,
\end{align}
be first order differential operators,
with $h_{1,s_i+1} = s_i^2 / 4$ in the Kac table~\eqref{eq:conf_weights}, 
and define
\begin{align} \label{eq:defBSAdifferentialoperator}
\mathcal{D}_{s_j+1}\super{j} := \sum_{k=1}^{s_j+1} \sum_{\substack{m_1,\ldots ,m_k \geq 1 \\ m_1+\cdots +m_k = {s_j+1}}} \frac{ (-1)^{k-s_j-1} (s_j!)^2}{\prod_{l=1}^{k-1} (\sum_{i=1}^l m_i)(\sum_{i=l+1}^k m_i)} \times \mathcal L_{-m_1}\super{j} \; \cdots \; \mathcal L_{-m_k}\super{j} .
\end{align}
(These are also known as Beno{\^i}t~\&~Saint-Aubin equations~\cite{BSA:Degenerate_CFTs_and_explicit_expressions_for_some_null_vectors}, in a CFT with central charge
$c=1$.)
A special case of this is the second order PDE system~\eqref{eq:BPZeq}, satisfied by the functions $\CobloF_T(x_1,\ldots ,x_{2N})$ in Definition~\ref{def:CBsh12def}, 
which we will use to derive the general case.

\begin{theorem} \label{thm:theoremBSA}
For each $T \in \CSYTof{N,N}$, the conformal block function 
of Definition~\ref{def:firstdeffusedCBs} satisfies 
\begin{align*} 
\mathcal{D}_{s_j+1}\super{j} \; \CobloF_T(x_1,\ldots ,x_\np) = 0, 
\qquad \textnormal{for all } j \in \llbracket 1, \np\rrbracket .
\end{align*}
\end{theorem}

Due to the complexity of the general BPZ differential operators in Equation~\eqref{eq:defBSAdifferentialoperator}, 
our proof of Theorem~\ref{thm:theoremBSA} 
does not rely on a direct computation utilizing the explicit representation of Lemma~\ref{lem:lemmaUintermsoffusedSpecht}. 
Instead, we follow a recursive approach.
A key result for the proof will be to show that, 
if we start from a solution of two BPZ equations of orders $s_j+1$ and $s_{j+1}+1=2$ at $x_j$ and $x_{j+1}$ having a specific asymptotic behavior as $x_{j+1} \to x_j$, 
then we can construct a solution of a new BPZ equation of order $s_j+2$ at $x_j$ which no longer depends on $x_{j+1}$. 
More precisely, the following result is the key to the proof of Theorem~\ref{thm:theoremBSA}.

\begin{restatable}{theorem}{intermediatetheorem}
\label{thm:intermediatetheorem}
Fix $\np \geq 2$. 
Fix $\multii=(s_1,\ldots,s_\np) \in \bZpos^\np$ such that 
$s_k = \ell-1$ and $s_{k+1}=1$ for some $k \in \llbracket 1,\ldots,\np-1\rrbracket$. 
Also, let $F \colon \chamber_\np \to \bR$ be 
a smooth function satisfying the BPZ PDEs
\begin{align} \label{eq:saintaubineq}
\mathcal{D}\super{j}_{s_j+1} F(x_1,\ldots,x_\np) = 0, 
\qquad \textnormal{for all } j \in \llbracket 1, \np \rrbracket .
\end{align}
Finally, using the indices $h_{s+1} := h_{1,s+1}$ in the Kac table~\eqref{eq:conf_weights}, 
assume that when $|x_{k+1}-x_k| > 0$ is small enough, the following (convergent) expansion holds: 
\begin{align} \label{eq:asyf}
F(x_1,\ldots,x_\np) = (x_{k+1}-x_k)^{h_{\ell+1}-h_{\ell}-h_{2}} 
\sum_{i\geq 0} f_i(\ldots, x_k,x_{k+2},\ldots) (x_{k+1}-x_k)^i,
\end{align}
where $f_i(x_1,\ldots, x_k,x_{k+2},\ldots, x_\np)$ are smooth functions on $\chamber_{\np-1}$. 
Then, $f_0$ satisfies the BPZ~PDEs 
\begin{align} 
\label{eq:jBSA} 
\mathcal{D}\super{j}_{s_j+1} f_0(x_1,\ldots, x_k,x_{k+2},\ldots, x_\np) = \; & 0, 
\qquad j \in \llbracket 1, \np \rrbracket , \; j \neq k, k+1 , \\
\label{eq:kp1BSA} 
\mathcal{D}\super{k}_{\ell+1} f_0(x_1,\ldots, x_k,x_{k+2},\ldots, x_\np) = \; & 0.
\end{align}
\end{restatable}

A result similar to Theorem~\ref{thm:intermediatetheorem} was proven 
through a direct computation by Karrila, Kyt\"ol\"a, and Peltola in~\cite[Lem.~5.6]{KKP:Boundary_correlations_in_planar_LERW_and_UST} 
in a specific scenario where the two \quote{merging} points $x_k$ and $x_{k+1}$ have $s_k=s_{k+1}=1$ and the other \quote{spectator} points have $s_j=1$ or $s_j=2$. 
However, extending their proof to the case of arbitrary $s_k\geq 1$ for one of the two merging points is, again, a priori out of reach due to the complexity of the BPZ differential operator. 
Instead, we follow an approach developed by Dub\'edat in \cite{Dubedat:SLE_and_Virasoro_representations_localization, 
Dubedat:SLE_and_Virasoro_representations_fusion}, 
which relies on the framework of Virasoro uniformization developed in particular by Kontsevich and Friedrich~\cite{Kontsevich:Virasoro_and_Teichmuller_spaces, 
Kontsevich:CFT_SLE_and_phase_boundaries, 
Friedrich-Kalkkinen:On_CFT_and_SLE,
Friedrich:On_connections_of_CFT_and_SLE}. 
Specifically,~\cite[Thm.~15]{Dubedat:SLE_and_Virasoro_representations_fusion} 
is a result similar to our Theorem~\ref{thm:intermediatetheorem}, 
except that it only applies to \emph{irrational} central charges $c \notin \bQ$, 
whereas the present case of interest concerns \emph{unit central charge}, $c=1$. 
Nevertheless, several key lemmas to the proof of~\cite[Thm.~15]{Dubedat:SLE_and_Virasoro_representations_fusion}
do still apply as well to $c=1$ --- 
and we will use them for the proof of Theorem~\ref{thm:intermediatetheorem} (in Section~\ref{subsec:proof}). 
As a matter of convenience for the readers, the majority of this proof will be relegated to the next Section~\ref{sec:section4}.

Recall that Definition~\ref{def:firstdeffusedCBs} expresses 
the conformal block functions 
in terms of an evaluation of a linear combination of conformal block functions for 
$\multii=(1^{2N})$. 
In the next result, we rewrite~\eqref{eq:infirstdeffusedCBs} 
in such a form that Theorem~\ref{thm:intermediatetheorem} can be applied recursively.

\begin{lemma} \label{lem:limits}
Fix $\multii=(s_1,\ldots,s_\np) \in \bZpos^\np$ with $s_1+\cdots+s_\np = 2N$. 
Let $f \in \SolSp\sub{1^{2N}}$. Then, we have
\begin{align} \label{eq:fusionpairbypair} 
\; & \left[\frac{\symTL . f}{\sqrt{\prod_{k=1}^\np \prod_{\summ_k\leq i<j<\summ_{k+1}} (x_j-x_i)}}\right]_{\textnormal{eval}} \\
\nonumber
= \; & 
\lim_{x_{\summ_\np} \to x_\np} \lim_{x_{\summ_\np+s_\np-1} \to x_{\summ_\np}} \frac{1}{(x_{\summ_\np+s_\np-1}-x_{\summ_\np})^{\frac{s_\np-1}2}} 
\cdots \hspace*{-4mm} 
\lim_{x_{\summ_\np+2} \to x_{\summ_\np}} \frac{1}{(x_{\summ_\np+2}-x_{\summ_\np})} \lim_{x_{\summ_\np+1} \to x_{\summ_\np}} \frac{1}{(x_{\summ_\np+1}-x_{\summ_\np})^{\frac12}}  \\
\nonumber
\; & \times \cdots \times
\lim_{x_{\summ_1} \to x_1} \lim_{x_{\summ_1+s_1-1} \to x_{\summ_1}} \frac{1}{(x_{\summ_1+s_1-1}-x_{\summ_1})^{\frac{s_1-1}2}}  
\cdots \hspace*{-4mm} 
%\lim_{x_{\summ_1+2} \to x_{\summ_1}} \frac{1}{(x_{\summ_1+2}-x_{\summ_1})} 
\lim_{x_{\summ_1+1} \to x_{\summ_1}} \frac{\symTL. f(x_1,\ldots ,x_{2N}) }{(x_{\summ_1+1}-x_{\summ_1})^{\frac12}}.
\end{align}
\end{lemma}

\begin{proof}
By Lemma~\ref{prop:lemma3p3}, we can write
$f = \Delta(\bs{x}_{1,\ldots,2N})^{-1/2} \, \mathcal P$ for some polynomial $\mathcal P \in P\super{2^{N}}$. 
We first rewrite the left-hand side of~\eqref{eq:fusionpairbypair} utilizing~\eqref{eq:actionObecomesactionp}:
\begin{align*}
\frac{\symTL . f}{\sqrt{\prod_{k=1}^\np \prod_{\summ_k\leq i<j<\summ_{k+1}} (x_j-x_i)}} 
= \frac{\idpt . \mathcal P \Big( \frac{\prod_{k=1}^\np \prod_{\summ_k\leq i<j<\summ_{k+1}} (x_j-x_i)^{1/2}}{\prod_{1\leq i < j \leq n} (x_j-x_i)^{1/2}} \Big)}{\prod_{k=1}^\np \prod_{\summ_k\leq i<j<\summ_{k+1}} (x_j-x_i)} .
\end{align*}
Note now that 
\begin{align}
\label{eq:defg1}
\idpt . \mathcal P = \; & \prod_{k=1}^\np \prod_{\summ_k\leq i<j<\summ_{k+1}} (x_j-x_i)\times Q_1, \\
\label{eq:defg2}
\prod_{1\leq i < j \leq n} (x_j-x_i) = \; & \prod_{k=1}^\np \prod_{\summ_k\leq i<j<\summ_{k+1}} (x_j-x_i)\times Q_2,
\end{align} 
where $Q_1$ and $Q_2$ are some polynomials, where in particular, $Q_2$ does not vanish at $x_i=x_j$ for $(i,j) \in \llbracket \summ_k,\ldots,\summ_{k+1}-1 \rrbracket^2$ and $k \in \llbracket 1, \np\rrbracket$. 
This leads to the formula
\begin{align} \label{eq:g1oversqrtg2}
\left[\frac{\symTL . f}{\sqrt{\prod_{k=1}^\np \prod_{\summ_k\leq i<j<\summ_{k+1}} (x_j-x_i)}}\right]_{\textnormal{eval}}
= \; &  \left[ \frac{Q_1}{\sqrt{Q_2}}\right]_{\textnormal{eval}} 
\end{align}
for the left-hand side of~\eqref{eq:fusionpairbypair}.
We now examine the right-hand side of~\eqref{eq:fusionpairbypair}:
\begin{align*}
\; & \lim_{x_{\summ_\np} \to x_\np} \lim_{x_{\summ_\np+s_\np-1} \to x_{\summ_\np}} 
\hspace*{-3mm} \cdots 
%\lim_{x_{\summ_\np+2} \to x_{\summ_\np}} \lim_{x_{\summ_\np+1} \to x_{\summ_\np}} 
%\hspace*{-3mm} \cdots 
\lim_{x_{\summ_1} \to x_1} \lim_{x_{\summ_1+s_1-1} \to x_{\summ_1}} 
\hspace*{-3mm} \cdots 
%\lim_{x_{\summ_1+2} \to x_{\summ_1}} 
\lim_{x_{\summ_1+1} \to x_{\summ_1}}
\frac{\symTL. f(x_1,\ldots ,x_{2N}) }{\prod_{k=1}^\np\prod_{m=1}^{s_k-1}(x_{\summ_k+m} - x_{\summ_k})^{m/2}}
\\
= \; & \lim_{x_{\summ_\np} \to x_\np} \lim_{x_{\summ_\np+s_\np-1} \to x_{\summ_\np}} 
\hspace*{-3mm} \cdots 
%\lim_{x_{\summ_\np+2} \to x_{\summ_\np}} \lim_{x_{\summ_\np+1} \to x_{\summ_\np}} 
%\hspace*{-3mm} \cdots 
\lim_{x_{\summ_1} \to x_1} \lim_{x_{\summ_1+s_1-1} \to x_{\summ_1}} 
\hspace*{-3mm} \cdots 
%\lim_{x_{\summ_1+2} \to x_{\summ_1}} 
\lim_{x_{\summ_1+1} \to x_{\summ_1}}
\frac{\idpt . \mathcal P \; \Big( \frac{\prod_{k=1}^\np\prod_{m=1}^{s_k-1}(x_{\summ_k+m} - x_{\summ_k})^{m/2}}{\prod_{\summ_k \leq i < j < \summ_{k+1}} (x_j-x_i)^{1/2}} \Big)}{\prod_{k=1}^\np\prod_{m=1}^{s_k-1}(x_{\summ_k+m} - x_{\summ_k})^{m}} . 
\end{align*}
We compute the chain of limits of each ratio separately. 
Fix $l \in \llbracket 1,\np\rrbracket$. 
By~\eqref{eq:defg1}, 
\begin{align*}
\; & \lim_{x_{\summ_l+s_l-1} \to x_{\summ_l}} \hspace*{-3mm} \cdots \lim_{x_{\summ_l+1} \to x_{\summ_l}}
\frac{\idpt. \mathcal P(x_1,\ldots ,x_{2N})}{\prod_{k=1}^\np\prod_{m=1}^{s_k-1}(x_{\summ_k+m} - x_{\summ_k})^{m}}
\\
= \; & \lim_{x_{\summ_l+s_l-1} \to x_{\summ_l}} 
\hspace*{-3mm} \cdots \lim_{x_{\summ_l+1} \to x_{\summ_l}} 
\frac{\prod_{k=1}^\np\prod_{i=0}^{j-1} \prod_{j=1}^{s_k-1} (x_{\summ_k+j}-x_{\summ_k+i})}{\prod_{k=1}^\np\prod_{m=1}^{s_k-1}(x_{\summ_k+m} - x_{\summ_k})^{m}} \; Q_1(x_1,\ldots ,x_{2N}).
\end{align*}
This chain of limits acts only on the terms $k=l$ in the product over $k$, 
so all terms for $k \neq l$ can be taken out. 
For the terms $k=l$ on which the limits act, we obtain
\begin{align*}
\; & 
\frac{\prod_{k=1, k\neq l}^\np\prod_{m=1}^{s_k-1}(x_{\summ_k+m} - x_{\summ_k})^{m}} {\prod_{k=1, k \neq l}^\np\prod_{i=0}^{j-1} \prod_{j=1}^{s_k-1} (x_{\summ_k+j}-x_{\summ_k+i})}
\; \lim_{x_{\summ_l+s_l-1} \to x_{\summ_l}} 
\hspace*{-3mm} \cdots \lim_{x_{\summ_l+1} \to x_{\summ_l}} 
\frac{\idpt . \mathcal P(x_1,\ldots ,x_{2N})}{\prod_{k=1}^\np\prod_{m=1}^{s_k-1}(x_{\summ_k+m} - x_{\summ_k})^{m}} \\
= \; & \lim_{x_{\summ_l+s_l-1} \to x_{\summ_l}} \hspace*{-3mm} \cdots 
\lim_{x_{\summ_l+1} \to x_{\summ_l}} \hspace*{-3mm} 
\frac{\prod_{i=0}^{s_l-2} (x_{\summ_l+s_l-1}-x_{\summ_l+i})}{(x_{\summ_l+s_l-1} - x_{\summ_l})^{s_l-1}} 
\cdots \frac{(x_{\summ_l+2}-x_{\summ_l+1})}{(x_{\summ_l+2} - x_{\summ_l})} 
%\frac{(x_{\summ_l+2}-x_{\summ_l})(x_{\summ_l+2}-x_{\summ_l+1})}{(x_{\summ_l+2} - x_{\summ_l})^2} 
%\frac{(x_{\summ_l+1}-x_{\summ_l})}{(x_{\summ_l+1} - x_{\summ_l})} 
\; Q_1(x_1,\ldots ,x_{2N}) .
\end{align*}
Since the limit of each ratio is finite at each step of the chain of limits, we obtain
\begin{align*}
\; & \lim_{x_{\summ_l+s_l-1} \to x_{\summ_l}} \hspace*{-3mm} \cdots \lim_{x_{\summ_l+1} \to x_{\summ_l}} 
\frac{\idpt . \mathcal P(x_1,\ldots ,x_{2N})}{\prod_{k=1}^\np\prod_{m=1}^{s_k-1}(x_{\summ_k+m} - x_{\summ_k})^{m}} \\
= \; & 
\frac{\prod_{k=1, k \neq l}^\np\prod_{i=0}^{j-1} \prod_{j=1}^{s_k-1} (x_{\summ_k+j}-x_{\summ_k+i})}{\prod_{k=1, k \neq l}^\np\prod_{m=1}^{s_k-1}(x_{\summ_k+m} - x_{\summ_k})^{m}} \lim_{x_{\summ_l+s_l-1} \to x_{\summ_l}} \hspace*{-3mm} \cdots \lim_{x_{\summ_l+1} \to x_{\summ_l}} Q_1(x_1,\ldots ,x_{2N}).
\end{align*}
Repeating the chain of limits as above for all $l \in \llbracket 1,\np\rrbracket$, 
we conclude that
\begin{align*}
\lim_{x_{\summ_\np} \to x_\np} \lim_{x_{\summ_\np+s_\np-1} \to x_{\summ_\np}} 
\hspace*{-3mm} \cdots 
%\lim_{x_{\summ_\np+2} \to x_{\summ_\np}} \lim_{x_{\summ_\np+1} \to x_{\summ_\np}} 
%\hspace*{-3mm} \cdots 
\lim_{x_{\summ_1} \to x_1} 
%\lim_{x_{\summ_1+s_1-1} \to x_{\summ_1}} 
\hspace*{-3mm} \cdots 
%\lim_{x_{\summ_1+2} \to x_{\summ_1}} 
\lim_{x_{\summ_1+1} \to x_{\summ_1}}
\frac{\idpt . \mathcal P(x_1,\ldots ,x_{2N})}{\prod_{k=1}^\np\prod_{m=1}^{s_k-1}(x_{\summ_k+m} - x_{\summ_k})^{m}} 
= \left[ Q_1\right]_{\textnormal{eval}} .
\end{align*}
Utilizing~\eqref{eq:defg2}, similar arguments can be invoked to show that 
\begin{align*}
\lim_{x_{\summ_\np} \to x_\np} \lim_{x_{\summ_\np+s_\np-1} \to x_{\summ_\np}} 
\hspace*{-3mm} \cdots 
%\lim_{x_{\summ_\np+2} \to x_{\summ_\np}} \lim_{x_{\summ_\np+1} \to x_{\summ_\np}} 
%\hspace*{-3mm} \cdots 
\lim_{x_{\summ_1} \to x_1} 
%\lim_{x_{\summ_1+s_1-1} \to x_{\summ_1}} 
\hspace*{-3mm} \cdots 
%\lim_{x_{\summ_1+2} \to x_{\summ_1}} 
\lim_{x_{\summ_1+1} \to x_{\summ_1}}
\hspace*{-2mm}
\frac{\prod_{k=1}^\np\prod_{m=1}^{s_k-1}(x_{\summ_k+m} - x_{\summ_k})^{m/2}}{\prod_{\summ_k \leq i < j < \summ_{k+1}} (x_j-x_i)^{1/2}}
= \left[\frac{1}{\sqrt{Q_2}}\right]_{\textnormal{eval}}.
\end{align*}
Multiplying the two equations above, 
we obtain the right-hand side of the sought~\eqref{eq:g1oversqrtg2}. 
\end{proof}

\begin{lemma} \label{lem:lemmaASY}
Let $F \colon \chamber_\np \to \bR$ such that $F = P_1/\sqrt{P_2}$ for some polynomials $P_1$ and $P_2$. 
Then, for any index $k \in \llbracket 1, \np-1 \rrbracket$, 
there exists $\theta_k \in \bZ$ such that
when $|x_{k+1}-x_k| > 0$ is small enough, the following (convergent) expansion holds: 
\begin{align*}
F(x_1,\ldots,x_\np) = (x_{k+1}-x_k)^{\theta_k/2} \sum_{i \geq 0} f_i(\ldots, x_k,x_{k+2},\ldots) \, (x_{k+1}-x_k)^i .
\end{align*}
Moreover, all coefficients $f_i$ have the form $P_i/\sqrt{Q_i}$, 
where $P_i$ and $Q_i$ are polynomials in the variables 
$x_1,\ldots,x_k,x_{k+2},\ldots,x_\np$ (which do not depend on $x_{k+1}$).
\end{lemma}

\begin{proof}
The first claim readily follows, since for any $k \in \llbracket 1, \np-1 \rrbracket$, 
there exists $\theta_k \in \bZ$ such that 
$F(x_1,\ldots,x_\np) = (x_{k+1}-x_k)^{\theta_k/2} g(x_1,\ldots,x_\np)$, 
where $g$ is (real) analytic at $x_{k+1}=x_k$. 
Concerning the second claim, any $f_i$ is of the form
\begin{align*}
f_i = \frac{1}{i!} \left[ \partial_{x_{k+1}} \big((x_{k+1}-x_k)^{-\theta_k/2} F(x_1,\ldots,x_\np)\big) \right]_{x_{k+1}=x_k} 
= \frac{P_i}{\sqrt{Q_i}} ,
\end{align*}
where $P_i$ and $Q_i$ are polynomials independent of $x_{k+1}$. 
This completes the proof. 
\end{proof}

We are now ready to prove Theorem~\ref{thm:theoremBSA}.
\begin{proof}[Proof of Theorem~\ref{thm:theoremBSA}] 
By definition~\eqref{eq:defspaceS}, 
each function $f \in \SolSp\sub{1^{2N}}$ satisfies $\mathcal{D}\super{j}_2 f=0$, 
for all $j \in \llbracket 1, 2N \rrbracket$.
Note that Lemma~\ref{lem:limits} consists of $\np$ chains of limits.
We proceed by induction on the number of limits, 
where the base case will be governed by $\symTL . f \in \SolSp\sub{1^{2N}}$.
As the induction hypothesis, 
we suppose that for given $i\in \llbracket 0,\np \rrbracket$, the function 
\begin{align*}
& \; g(x_1,x_2,\ldots,x_{i-1},x_{\summ_i},x_{\summ_i+1},\ldots,x_{2N}) \\
= \; & \lim_{x_{\summ_{i-1}} \to x_{i-1}} \lim_{x_{\summ_{i-1}+s_{i-1}-1} \to x_{\summ_{i-1}}} \frac{1}{(x_{\summ_{i-1}+s_{i-1}-1}-x_{\summ_{i-1}})^{(s_{i-1}-1)/2}}  
%\cdots \hspace*{-4mm} 
\\
\; & \times \cdots \times
\lim_{x_{\summ_{i-1}+2} \to x_{\summ_{i-1}}} \frac{1}{(x_{\summ_{i-1}+2}-x_{\summ_{i-1}})} 
\lim_{x_{\summ_{i-1}+1} \to x_{\summ_{i-1}}} \frac{1}{(x_{\summ_{i-1}+1}-x_{\summ_{i-1}})^{1/2}} 
\\
\; & \times \cdots \times
\lim_{x_{\summ_1} \to x_1} \lim_{x_{\summ_1+s_1-1} \to x_{\summ_1}} \frac{1}{(x_{\summ_1+s_1-1}-x_{\summ_1})^{(s_1-1)/2}} 
\cdots \hspace*{-4mm} 
%\lim_{x_{\summ_1+2} \to x_{\summ_1}} \frac{1}{(x_{\summ_1+2}-x_{\summ_1})} 
\lim_{x_{\summ_1+1} \to x_{\summ_1}} \frac{\symTL . f(x_1,\ldots ,x_{2N})}{(x_{\summ_1+1}-x_{\summ_1})^{1/2}} 
\end{align*}
is of the form $P/\sqrt{Q}$ with some polynomials $P$ and $Q$, 
and that $g$ satisfies the BPZ PDEs
\begin{align*}
\begin{cases}
\mathcal{D}_{s_j+1}\super{j} \, g = 0, & \textnormal{for all} \; j\in \llbracket 1,i-1 \rrbracket , 
\\[.3em]
\mathcal{D}_{2}\super{j} \, g = 0, & \textnormal{for all} \; j\in \llbracket \summ_i,2N \rrbracket.
\end{cases}
\end{align*}
(So the base case is $i=0$, in which case we just have $\symTL . f$ --- 
indeed of the form $P/\sqrt{Q}$ and satisfies $\mathcal{D}\super{j}_2 f=0$ for all $j\in \llbracket 1,2N \rrbracket$.) 
Now, define $g_k$ with $k\in \llbracket 0, s_i-1\rrbracket$ and such that 
\begin{align} \label{eq:defgk} 
\; & g_k(x_1,x_2,\ldots,x_{i-1},x_{q_i},x_{\summ_i+k+1},x_{\summ_i+k+2},\ldots,x_{2N}) \\
\nonumber 
= \; & \lim_{x_{\summ_{i}+k} \to x_{\summ_{i}}} \frac{1}{(x_{\summ_{i}+k}-x_{\summ_{i}})^{k/2}}  \cdots \hspace*{-2mm} \lim_{x_{\summ_{i}+2} \to x_{\summ_{i}}} \frac{1}{(x_{\summ_{i}+2}-x_{\summ_{i}})} \lim_{x_{\summ_{i}+1} \to x_{\summ_{i}}} \frac{g(\ldots,x_{i-1},x_{\summ_i},x_{\summ_i+1},\ldots)}{(x_{\summ_{i}+1}-x_{\summ_{i}})^{1/2}} .
\end{align}
We now perform the induction step, 
i.e., show that the function $g_{s_i-1}$ is also of the form $P/\sqrt{Q}$ for some (different) polynomials $P$ and $Q$, and that it satisfies the BPZ PDEs
\begin{align}\label{eq:gsim1}
\begin{cases}
\mathcal{D}_{s_j+1}\super{j} \, g = 0, & \textnormal{for all} \; j\in \llbracket 1,i \rrbracket,
\\[.3em]
\mathcal{D}_{2}\super{j} \, g = 0, & \textnormal{for all} \; j\in \llbracket i+1, 2N \rrbracket.
\end{cases}
\end{align}
By induction, this then implies that the following function 
satisfies $\mathcal{D}_{s_j+1}\super{j} h = 0$ for all $j$: %$j\in \llbracket 1,\np \rrbracket$:
\begin{align*}
h := \left[\frac{\symTL. f}{\sqrt{\prod_{k=1}^\np\prod_{\summ_k\leq i<j<\summ_{k+1}} (x_j-x_i)}}\right]_{\textnormal{eval}} .
\end{align*}
Since the conformal block functions $\CobloF_T(x_1,\ldots,x_\np)$ 
are defined by~\eqref{eq:infirstdeffusedCBs}, 
it thus suffices to take $f = \CobloF_{\hat{T}}(x_1,\ldots,x_{2N})$ and 
$T^t \in \SYTof{2^{N}}$ to conclude the proof of Theorem~\ref{thm:theoremBSA}.

In order to finish the induction step, 
we again proceed by induction, now on $k\in \llbracket 0 , s_i -1\rrbracket$. 
Suppose $g_k = P_k/\sqrt{Q_k}$ where $P_k$ and $Q_k$ are polynomials, and $g_k$ satisfies the BPZ PDEs
\begin{align*}
\begin{cases}
\mathcal{D}_{s_j+1}\super{j} \, g_k = 0, & \textnormal{for all} \; j\in \llbracket 1,i-1 \rrbracket , 
\\[.3em]
\mathcal{D}_{k+2}\super{\summ_i} \, g_k = 0 , 
\\[.3em]
\mathcal{D}_{2}\super{j} \, g_k = 0, & \textnormal{for all} \; j\in \llbracket \summ_i+k+1,2N \rrbracket .
\end{cases}
\end{align*}
This is obviously true for the base case $k=0$, since $g_0=g$. 
Next, suppose this is true for a given $k\in \llbracket 1, s_i-1\rrbracket$. By~\eqref{eq:defgk}, the function $g_{k+1}$ is given by
\begin{align*}
g_{k+1}=\lim\limits_{x_{\summ_{i}+k+1} \to x_{\summ_{i}}} \frac{g_k}{(x_{\summ_{i}+k+1}-x_{\summ_{i}})^{(k+1)/2}} .
\end{align*}
Therefore, by Lemma~\ref{lem:lemmaASY} we have
\begin{align*}
g_k=(x_{\summ_i+k+1}-x_{\summ_i})^{(k+1)/2} \sum_{m \geq 0} u_m \, (x_{\summ_i+k+1}-x_{\summ_i})^k,
\end{align*}
where the coefficients $u_m$ are smooth for all $m$ and, in particular, 
$g_{k+1}=u_0$ is of the form $P_{k+1}/\sqrt{Q_{k+1}}$ for some polynomials $P_{k+1}$ and $Q_{k+1}$. Moreover, because
\begin{align*}
h_{1,k+3}-h_{1,k+2}-h_{1,2} = (k+1)/2,
\end{align*}
we can apply Theorem~\ref{thm:intermediatetheorem} 
to deduce that $g_{k+1}$ satisfies the BPZ PDEs 
\begin{align*}
\begin{cases}
\mathcal{D}_{s_j+1}\super{j} \, g_{k+1} = 0, & \textnormal{for all} \; j \in \llbracket 1,i-1 \rrbracket,
\\[.3em]
\mathcal{D}_{k+3}\super{\summ_i} \, g_{k+1} = 0,
\\[.3em]
\mathcal{D}_{2}\super{j}g_{k+1} = 0,  & \textnormal{for all} \; j \in \llbracket \summ_i+k+2,2N \rrbracket . 
\end{cases}
\end{align*}
We then conclude by induction that $g_{s_i-1}$ satisfies~\eqref{eq:gsim1}; 
thereby proving Theorem~\ref{thm:theoremBSA}.
\end{proof}

\bigskip{}
\section{Fusion argument for BPZ PDEs --- Proof of Theorem~\ref{thm:intermediatetheorem}} 
\label{sec:section4}
To prove that the BPZ PDEs are satisfied at all valences, 
we follow a fusion argument bootstrapping from the already known lower order PDEs 
to the higher order ones. 
This approach, which seems to us to be the most amenable one to carry out systematically, 
utilizes a combination of tools from algebra and analytic geometry, 
and rigorously appeared in~\cite{Dubedat:SLE_and_Virasoro_representations_fusion} 
for the case of irrational central charges. 
Since the case of present interest is that of unit central charge $c=1$, 
we have to modify the argument to account for slightly more complicated representation structure of the Virasoro algebra.
We present the gist of the proof in this section in a manner that does not assume prior knowledge of~\cite{Dubedat:SLE_and_Virasoro_representations_fusion}.

First of all, it is well-known and not too hard to check that the differential operators 
$\{\mathcal L_m\super{j} \;|\; m \in \bZ \}$ in~\eqref{eq:deffirstorderdiffop} 
satisfy the commutation relations of the Witt algebra. 
In fact, there is a natural action of $\np$ copies of the Witt algebra acting 
on the space of functions $\SolSp_\multii$, one copy for each point. 
These actions do not commute with each other. 
An essential step for the proof of Theorem~\ref{thm:intermediatetheorem} will be 
to extend such non-commuting actions of the Witt algebra 
to \emph{commuting} actions of the \emph{Virasoro algebra}. 
Such an extension was investigated in detail by Dub\'edat in~\cite{Dubedat:SLE_and_Virasoro_representations_localization, 
Dubedat:SLE_and_Virasoro_representations_fusion} 
within the geometric framework of Virasoro uniformization~\cite{Kontsevich:Virasoro_and_Teichmuller_spaces, Beilinson-Schechtman:Determinant_bundles_and_Virasoro_algebra,Friedrich:On_connections_of_CFT_and_SLE, 
Friedrich-Kalkkinen:On_CFT_and_SLE}). 
In this approach, the Virasoro algebra acts on the space of sections of a suitable line bundle over an extended Teichm\"uller space. 
Whereas the Teichm\"uller space of a surface parametrizes its equivalence classes of complex (Riemann surface) structures, the extension we consider emerges 
from the fact that the Riemann surface is endowed with more data. More precisely, the extra data consist of a choice of a local coordinate for each marked point, and vanishing at the corresponding marked point. 
This viewpoint is closely related to Segal's sewing formalism~\cite{Segal:Definition_of_CFT}, which instead considers Riemann surfaces with parametrized boundary circles and their filling with analytic disks. 
Let us mention that while Dub\'edat's framework~\cite{Dubedat:SLE_and_Virasoro_representations_localization,
Dubedat:SLE_and_Virasoro_representations_fusion} 
holds for general bordered Riemann surfaces with marked points, we will specifically study the case of genus zero 
Riemann surfaces with one boundary component and with $\np$ marked points lying on the boundary. 
The (extended) Teichm\"uller space of such surfaces is simpler because its first homology group 
(and, therefore, their mapping class group) is trivial.

The proof of Theorem~\ref{thm:intermediatetheorem} consists of three steps 
and utilizes various results of~\cite{Dubedat:SLE_and_Virasoro_representations_localization,
Dubedat:SLE_and_Virasoro_representations_fusion}. 
The first step is to construct the extension of the solution space of a set of $\np$ 
BPZ differential equations giving rise to $\np$ non-commuting actions of the Witt algebra 
to the space of sections of a line bundle over the extended Teichm\"uller space, giving rise to 
$\np$ commuting actions of the Virasoro algebra. 
We describe such an extension at the beginning of Section~\ref{subsec:section6p2}, 
and it essentially recalls the results of~\cite[Sect.~4]{Dubedat:SLE_and_Virasoro_representations_localization}. 
Once such a space of sections is identified, the second step of the proof consists of choosing a local coordinate which encircles the marked points $x_k$ and $x_{k+1}$ and 
studying what the two Virasoro representations at $x_k$ and $x_{k+1}$ become in the limit $|x_{k+1} - x_k| \to 0$. 
A~crucial point here is to translate the problem, written in analytic-geometric form, into an equivalent algebraic problem, which then becomes amenable. 
As a matter of convenience for the readers, and because this is the key difference 
to~\cite{Dubedat:SLE_and_Virasoro_representations_localization,
Dubedat:SLE_and_Virasoro_representations_fusion},  
we first address the algebraic part of the problem separately in Section~\ref{subsec:section5p1p2} 
(see Lemma~\ref{lem:analoglemma1Dubedat}). 
Finally, once the algebraic problem is solved, 
it remains to utilize the extension the other way around 
to get back to solutions of higher order BPZ differential equations, as desired (see Section~\ref{subsec:proof}).

\subsection{Verma modules over the Virasoro algebra and fusion} 

The \emph{Virasoro algebra} $\Vir$ is the infinite-dimensional Lie algebra 
generated by the Virasoro modes $\{L_n \;|\; n \in \bZ\}$ and the central element $C$, 
\begin{align*}
\Vir = \bC C \oplus \bigoplus_{n \in \bZ} \bC L_n ,
\end{align*}
with the following commutation relations:
\begin{align} \label{eq:comm rel}
\begin{split}
[L_m,L_n] = \; & (m-n) L_{m+n} + \delta_{m,-n} \frac{m^2(m-1)}{12} C, \qquad m,n \in \bZ, \\
[C, \Vir] = \; & 0, \qquad n \in \bZ
\end{split}
\end{align}
(where $\delta_{i,j}$ stands for the Kronecker delta function, equaling zero unless $i=j$). 
It has the triangular decomposition $\Vir = \Vir^- \oplus \mathfrak h \oplus \Vir^+$, 
where $\mathfrak h = \bC C \oplus L_0$ and $\smash{\Vir^\pm = \underset{\pm n>0}{\oplus} \bC L_n}$. 
The universal enveloping algebra of the subalgebra $\Vir^-$ is
\begin{align*}
\mathcal U(\Vir^-) 
= \bigoplus_{\substack{0 < i_1 \leq \cdots  \leq i_k \\ k \geq 0}} \bC L_{-i_k} \cdots  L_{-i_1} ,
\end{align*}
and it has the \quote{standard basis} 
$\{L_{-i_k} \cdots  L_{-i_1} \;|\; 0< i_1 \leq \cdots  \leq i_k, \; k \geq 0\}$ by the Poincar\'e-Birkhoff-Witt theorem. 
Let us also note that $\mathcal U(\Vir)$ is a $\bZ$-graded algebra with \emph{degree} $\textnormal{deg}(L_n) := -n$ and $\textnormal{deg}(C) := 0$.
(See the textbook~\cite{Iohara-Koga:Representation_theory_of_Virasoro} for more background on $\Vir$.)

Let $V$ be a $\Vir$-module. 
For $(c,h) \in \bC^2$, a $(c,h)$-\emph{highest-weight vector} $v_{h}^{c} \in V$ is 
an element satisfying $C v_{h}^{c} = c v_{h}^{c}$, $L_0 v_{h}^{c} = h v_{h}^{c}$, 
and $L_n v_{h}^{c} = 0$ for all $n>0$. 
In this context, $c \in \bC$ is called the \emph{central charge} and 
$h \in \bC$ is called the \emph{weight} of $v_{h}^{c}$.  
The \emph{Verma module} $M_{h}^{c}$ is the $\bZ_{\geq 0}$-graded $\Vir$-module spanned by $\mathcal U(\Vir^-)v_{h}^{c}$, 
\begin{align*}
M_{h}^{c} = \bigoplus_{\ell \geq 0} (M_{h}^{c})_\ell ,
\qquad \textnormal{where} \qquad
(M_{h}^{c})_\ell := \bigoplus_{\substack{0 < i_1 \leq \cdots  \leq i_k \\ i_1+\cdots +i_k = \ell \\ k \geq 0}} \bC L_{-i_k}\cdots L_{-i_1} v_{h}^{c} .
\end{align*}
Note that the dimension of $(M_{h}^{c})_\ell$ is the number of partitions of $\ell$.
Moreover, it follows from the commutation relations~\eqref{eq:comm rel} 
that each element $v \in (M_{h}^{c})_\ell$ satisfies $L_0 v = (h+\ell) v$. 
Hence, we say that each $v \in (M_{h}^{c})_\ell$ is a vector in $M_{h}^{c}$ at \emph{level} (or degree) $\ell$.

A highest-weight vector $w_\ell \in M_{h}^{c}$ of level $\ell>0$ is called a \emph{singular vector}.  
If a non-zero singular vector 
can be found, 
then $M_{h}^{c}$ is said to be \emph{degenerate at level $\ell>0$}, 
and in this case, $w_\ell$ generates a proper submodule of $M_{h}^{c}$ isomorphic to $M_{h+\ell}^{c}$. 
Submodules of Verma modules were classified by B.~Fe{\u\i}gin and D.~Fuchs
\cite{Feigin-Fuchs:Invariant_skew-symmetric_differential_operators_on_the_line_and_Verma_modules_over_Virasoro,
Feigin-Fuchs:Verma_modules_over_Virasoro_book,
Feigin-Fuchs:Representations_of_Virasoro}.
(See, e.g., the book~\cite{Iohara-Koga:Representation_theory_of_Virasoro} for more background.) 
In particular, every submodule of $M_{h}^{c}$ is generated by singular vectors. 
There is an exceptional set of parameters $(c,h)$ for which $M_{h}^{c}$ is not irreducible --- the Kac table~\cite{Kac:Contravariant_form_for_infinite-dimensional_Lie_algebras_and_superalgebras,
Kac:Highest_weight_representations_of_infinite_dimensional_Lie_algebras} 
--- see~\eqref{eq:conf_weights} for an example with $c=1$ (relevant to the present work). 
Since irreducible modules generally appear in conformal field theory applications, it is important to classify singular vectors of $M_{h}^{c}$, 
which was also done in~\cite{Feigin-Fuchs:Invariant_skew-symmetric_differential_operators_on_the_line_and_Verma_modules_over_Virasoro,
Feigin-Fuchs:Verma_modules_over_Virasoro_book,
Feigin-Fuchs:Representations_of_Virasoro}.

Fix $c=1$. From now on, we only consider Verma modules of type $M_{h} := M_{h}^{1}$,
which possess a singular vector at level $\ell>0$ if and only if $h$
belongs to the Kac table~\eqref{eq:conf_weights}:
\begin{align} \label{eq:conf_weights again}
h_{\ell} = h_{1,\ell} := \tfrac{1}{4} (\ell-1)^2 \in \big\{ 0,\tfrac{1}{4},1,\tfrac{9}{4},4,\tfrac{25}{4},9,\tfrac{49}{4},16,\tfrac{81}{4} , \ldots \big\} .
\end{align}
Let $v_{\ell} = \smash{v_{h_{\ell}}^{1}}$ denote the highest-weight vector of $M_{{\ell}} := M_{h_{\ell}}$.
Then, the singular vector at level $\ell$ has the form $w_\ell = \Delta_\ell v_{\ell}$, 
where $\Delta_\ell \in \mathcal U(\Vir^-)$ is some polynomial in the negative Virasoro generators. 
As the coefficient of $L_{-1}^\ell$ in $\Delta_\ell$ cannot vanish~\cite[Sect.~5.2.1]{Iohara-Koga:Representation_theory_of_Virasoro}, we may normalize it to one.
An explicit formula for the polynomial $\Delta_\ell$ was found in~\cite{BSA:Degenerate_CFTs_and_explicit_expressions_for_some_null_vectors}\footnote{Note that in~\eqref{eq:defBSAoperator}, the Virasoro generators $L_{-i_j}$ are not ordered.}:
\begin{align} \label{eq:defBSAoperator}
\Delta_\ell = \sum_{k=1}^\ell \sum_{\substack{i_1,\ldots ,i_k \geq 1 \\ i_1+\cdots +i_k = \ell}} \frac{(-1)^{\ell-k} (\ell-1)!^2}{\prod_{l=1}^{k-1} (\sum_{j=1}^l i_j)(\sum_{j=l+1}^k i_j)} 
\; L_{-i_1} \cdots  L_{-i_k}
\end{align}
(see also~\cite{BDIZ:Singular_vectors_of_Virasoro_algebra}).
For instance, $\ell=1,2,3$ the formula~\eqref{eq:defBSAoperator} yields 
\begin{align*}
\Delta_1 = \; & L_{-1}, \\
\Delta_2 = \; & L_{-1}^2 - L_{-2}, \\
\Delta_3 = \; & L_{-1}^3 - 2(L_{-1} L_{-2} + L_{-2}L_{-1}) + 4 L_{-3}.
\end{align*}
Observe that $L_0 (\Delta_\ell v_{\ell}) = (h_{\ell}+\ell) (\Delta_\ell v_{\ell}) = h_{\ell+2} (\Delta_\ell v_{\ell})$. 
In fact, the singular vector $\Delta_\ell v_{\ell}$ generates a submodule of $M_{{\ell}}$ isomorphic to $M_{{\ell+2}}$, which is the maximal proper submodule. 
Generally, when $c=1$ there exists a one-dimensional infinite chain of submodules, 
where each arrow denotes the embedding of $M_{{j+2}}$ into $M_{{j}}$ giving its maximal proper submodule:
\begin{align} \label{eq:submodules}
M_{{\ell}} \hookleftarrow M_{{\ell+2}} \hookleftarrow M_{{\ell+4}} \hookleftarrow \cdots .
\end{align} 
This structure of the Verma module $M_{{\ell}}$ is referred to as \quote{chain} type (see~\cite[Figure~1]{Kytola-Ridout:On_staggered_indecomposable_Virasoro_modules}, 
and~\cite{Feigin-Fuchs:Verma_modules_over_Virasoro_book, Iohara-Koga:Representation_theory_of_Virasoro} for details). 
Let us also remark that the submodule structure of Verma modules can be more intricate 
for other rational values of the central charge~\cite{Feigin-Fuchs:Verma_modules_over_Virasoro, 
Iohara-Koga:Representation_theory_of_Virasoro}.

\subsection{Fusion: the key algebraic lemma} 
\label{subsec:section5p1p2}

Let $t$ be a formal variable. 
For $\alpha, h\in \bR$, consider the space $V_{\alpha,h} := \bC [[t]][t^{-1}]t^\alpha$ 
of formal series with finitely many negative terms:
\begin{align*}
t^\alpha \sum_{k\in \bZ} a_k t^k ,\quad a_k\in \bC , \qquad  
\textnormal{with} \quad \inf \{k \colon a_k\neq 0\} > -\infty .
\end{align*}
$V_{\alpha,h}$ is a $\Vir$-module with zero central charge $c=0$, 
where each generator $L_n$ acts by 
\begin{align*}
L_n \mapsto L_n^0 := -t^{n+1}\partial_t-(n+1)ht^n .
\end{align*}
(The role of the parameter $\alpha$ will become clear in the fusion procedure later, see Lemma~\ref{lem:analoglemma1Dubedat}, and also~\cite[Sect.~8.A]{DMS:CFT}.)
The operators $\{L_n^0 \;|\; n \in \bZ\}$ satisfy the commutation relations 
$[L_m^0,L_n^0] = (m-n) L_{m+n}^0$ of the \emph{Witt algebra} 
(so $V_{\alpha,h}$ is also a Witt-module\footnote{Recall that the Virasoro algebra is a one-dimensional central extension of the Witt algebra.}).
This action is motivated by CFT in the context of vertex algebras
(cf.~\cite{Huang:2D_Conformal_geometry_and_VOAs, Kac:Vertex_algebras_for_beginners,
Frenkel-Ben-Zvi:Vertex_Algebras_and_Algebraic_Curves}): 
for a Riemann surface with marked points, 
to each marked point one associates a representation of the Witt algebra 
(morally, the Lie algebra of deformations of the complex structure, 
where the formal variable represents a local coordinate):
a deformation near a given marked point is governed by the Witt-action on the corresponding module. 
(See also~\cite[Sect.~2.4]{Dubedat:SLE_and_Virasoro_representations_localization}.)

Next, let $W$ be a $\Vir$-module with central charge $c=1$, 
whose $\Vir$-action is simply denoted by $L_n$. 
Consider the space $W\otimes V_{\alpha,h}$ of formal series with coefficients in $W$:
\begin{align*}
t^\alpha \sum_{k\in \bZ} v_k t^k,\quad v_k\in W , \qquad 
\textnormal{with} \quad \inf \{k \colon v_k\neq 0\} > -\infty .
\end{align*}
Then, 
$W\otimes V_{\alpha,h}$ is a $\Vir$-module with central charge $c=1$,
where each generator $L_n$ acts by 
\begin{align} \label{eq:actionhatLn}
L_n \mapsto \hat{L}_n(vt^{\alpha+k}) := (L_nv)t^{\alpha+k}-(\alpha+k+(n+1)h)vt^{\alpha+k+n} ,
\qquad n,k \in \bZ, \; v \in W.
\end{align}
Let $\hat{\Delta}_\ell$ be the BPZ operator in~\eqref{eq:defBSAoperator} 
with the substitutions $L_n \mapsto \hat{L}_n$ for all $n \in \bZ$.

The reason to introduce the $\Vir$-module $W\otimes V_{\alpha,h}$ 
is motivated by fusion in CFT. 
If $W$ is the Virasoro module associated to a given marked point $x$, 
the tensor product $W\otimes V_{\alpha,h}$ associates another $\Vir$-module $V_{\alpha,h}$ at a nearby point $y=x+t$,
and the action~\eqref{eq:actionhatLn} can be thought of as a deformation at $x$ also keeping track of $y$. 
Conversely, a deformation at $y=x+t$ keeping track of $x$ can be represented by operators of type~(\ref{eq:actiontildeL1},~\ref{eq:actiontildeL2}) in Lemma~\ref{lem:analoglemma1Dubedat}.

We are now ready to state the key algebraic result, 
crucial for the proof of Theorem~\ref{thm:intermediatetheorem}. 
It is analogous to~\cite[Lem.~1]{Dubedat:SLE_and_Virasoro_representations_fusion}
--- however the proof slightly differs because the Virasoro submodule structure 
is more intricate for the present case of $c=1$ than for irrational $c$.
This is the main reason why we cannot use the results~\cite{Dubedat:SLE_and_Virasoro_representations_localization, 
Dubedat:SLE_and_Virasoro_representations_fusion} 
of Dub\'edat directly.

The result is an algebraic formulation of the fusion of two points $x$ and $y=x+t$ on a Riemann surface (as $t \to 0$). 
We assume that the point $x$ carries a Virasoro highest-weight representation of 
weight $h_{\ell} := h_{1,\ell}$ and 
the point $y$ carries a Virasoro highest-weight representation of 
weight $h_{2} := h_{1,2}$.
We expect from the CFT operator product expansion (fusion) 
of the corresponding two fields that 
\quote{$\Phi_{1,\ell}(x) \times \Phi_{1,2}(y) = \Phi_{1,\ell-1}(x) + \Phi_{1,\ell+1}(x)$} 
as $y \to x$.
In the present work, we are interested in the subleading channel $\Phi_{1,\ell+1}$,
which results in a conformal weight $h_{\ell+1} = h_{1,\ell+1}$ at higher level, 
needed for Theorem~\ref{thm:intermediatetheorem}.

\begin{lemma} \label{lem:analoglemma1Dubedat}
Fix $\ell \geq 2$. 
Using the notation from~\eqref{eq:conf_weights again}, let us 
denote $\smash{\Tilde{h}} := h_{2}$, $\smash{\hat{h}} := h_{\ell}$, and $\alpha := h_{\ell+1} - \smash{\hat{h}} - \smash{\Tilde{h}}$. 
Suppose $w = t^\alpha \sum_{k\geq 0} v_k t^k$ is a highest-weight vector of weight $\smash{\hat{h}}$ such that
\begin{align*}
\hat{\Delta}_\ell w = 0 
\qquad \textnormal{and} \qquad
\Tilde{\Delta}_2 w=0 ,
\end{align*}
where $\Tilde{\Delta}_2 := \Tilde{L}_{-1}^2-\Tilde{L}_{-2}$ is defined in terms of
\begin{align}
\label{eq:actiontildeL1} 
\Tilde{L}_{-1}:= \; & \partial_t \\
\label{eq:actiontildeL2} 
\Tilde{L}_{-2}:= \; & -t^{-1}\partial_t + t^{-1}L_{-1} + \smash{\hat{h}} t^{-2} + \sum_{k\geq 0} t^k L_{-2-k} .
\end{align}
Then, the coefficient $v_0$ 
is a highest-weight vector in $W$ of weight $h_{\ell+1}$ which satisfies $\Delta_{\ell+1} v_0=0$.
\end{lemma}

\begin{proof}
The proof consists of two steps. 
The first step is to check that $v_0 \in W$ is indeed a highest-weight vector of weight $h_{\ell+1}$.  
Indeed, we have $\hat{L}_0 w = \smash{\hat{h}} w$ at degree $\alpha$, 
which yields $L_0v_0=(\smash{\hat{h}}+\smash{\Tilde{h}}+\alpha)v_0$. 
Moreover, we have $\hat{L}_n w=0$ at degree $\alpha$, for all $n>0$, 
which yields $L_nv_0=0$, for all $n>0$. 
This shows that $v_0$ is a highest-weight vector of weight $h_{\ell+1}$.

The second and last step of the proof is to find an element 
$P_k \in \mathcal U(\Vir^-) \setminus \{0\}$ of degree $k < 2\ell+4$ such that $P_k v_0=0$. 
To see why this is useful, consider the homomorphism $\phi \colon M_{{\ell+1}} \to W$ 
of $\Vir$-modules 
which maps the highest-weight vector $v_{\ell+1} \in M_{{\ell+1}}$ to $v_0 \in W$. 
The first isomorphism theorem of modules implies that $\text{Ker}(\phi)$ is a proper submodule of $M_{{\ell+1}}$. 
Using the chain~\eqref{eq:submodules} of Verma modules, we obtain
\begin{align*}
M_{{\ell+1}} \hookleftarrow M_{{\ell+3}} \hookleftarrow M_{{\ell+5}} \hookleftarrow \cdots ,
\end{align*} 
where the image of $M_{{\ell+3}}$ is generated by $\Delta_{\ell+1} v_{\ell+1}$, 
the image of $M_{{\ell+5}}$ is generated by $\Delta_{\ell+3} (\Delta_{\ell+1} v_{\ell+1})$, etc. 
Now, if there exists $P_k \in \mathcal U(\Vir^-) \setminus \{0\}$ of degree $k$ such that $P_k v_0=0$, then it follows that $P_k v_{\ell+1} \in \text{Ker}(\phi)$. 
In particular, we have $\text{Ker}(\phi) = M_{{\ell+3}}$ if $k<2\ell+4$, 
in which case we may conclude that $0 = \phi(\Delta_{\ell+1} v_{\ell+1}) = \Delta_{\ell+1} v_0$, as desired.

It now remains to construct such a $P_k$.
To this end, consider first the assumption $\Tilde{\Delta}_2 w = 0$. 
Expanding by degree, we obtain
\begin{align*}
0  = \; & \rho(\alpha)v_0 \\
0  = \; & \rho(\alpha +1)v_1-L_{-1}v_{0} \\
0  = \; & \rho(\alpha +k)v_k-\sum_{j=1}^{k}L_{-j}v_{k-j} ,
\end{align*}
where $\rho(a) = a^2 - \smash{\hat{h}}$ has roots $\alpha$ and $h_{\ell-1}-\smash{\hat{h}}-\smash{\Tilde{h}} < \alpha$. 
Thus, we have $\rho(\alpha+k)\neq 0$ for all $k>0$, 
and there exist elements $R_0,R_1,\ldots \in \mathcal{U}(\Vir^-)$ such that $v_k = R_kv_0$ for all $k$.

Next, consider the assumption $\hat{\Delta}_\ell w=0$:
\begin{align*}
\hat{\Delta}_\ell \Big( t^\alpha \sum_{k\geq 0} t^kR_kv_0 \Big)=0.
\end{align*}
Write
\begin{align*}
\hat{\Delta}_\ell \Big( t^\alpha \sum_{k\geq 0} t^kR_k \Big) = t^{\alpha-\ell} \sum_{k\geq 0} t^kP_k ,
\end{align*}
for some polynomials $P_k\in \mathcal{U}(\Vir^-)$ of degree $k$ such that $P_kv_0=0$, for all $k\geq 0$. 
We first focus on the coefficients of $L_{-1}^k$ of $P_k$ decomposed in the standard basis.  If $P,Q\in \mathcal{U}(\Vir^-)$ 
are homogeneous and such that $P=aL_{-1}^k+\cdots $ and $Q=bL_{-1}^{k'}+\cdots $ in the standard basis, then $PQ=abL_{-1}^{k+k'}+\cdots $ in the standard basis. 
(This holds because the commutation relations of $\text{Vir}^-$ do not produce any monomial in $L_{-1}$.) 
We then see inductively that 
\begin{align*}
R_k=\frac{1}{\varrho(1)\cdots \varrho(k)}L_{-1}^k+\cdots , \qquad k \geq 1 ,
\end{align*}
with $\varrho(k)=\rho(\alpha+k)$.
Next, we write
\begin{align*}
\hat{\Delta}_\ell = \sum_{\substack{i+j+k=\ell \\ i,j,k\geq 0}} b_{i,j,k}t^{-i}\partial_t^jL_{-1}^k+\cdots , 
\end{align*}
where the remainder does not contain any monomial in $L_{-1}$. Note that $b_{0,0,\ell}=1$.

We now finally show that there exists an element $P_k \in \mathcal U(\Vir^-) \setminus \{0\}$ of degree $d < 2\ell+4$ such that $P_k v_0=0$. 
To this end, we assume towards a contradiction that no $P_k$ has a nonzero monomial 
in $L_{-1}$ for $k\leq 2\ell+3$. Then, we have
\begin{align*}
t^{\alpha-\ell} \sum_{k\geq 0} t^kP_k 
= \; & \hat{\Delta}_\ell \Big( t^\alpha \sum_{k\geq 0} t^kR_k \Big) \\
= \; & 
\bigg( \sum_{\substack{i+j+k=\ell \\ i,j,k\geq 0}} b_{i,j,k}t^{-i}\partial_t^jL_{-1}^k\bigg)\bigg( t^\alpha \sum_{l\geq 0} \frac{t^lL_{-1}^l}{\varrho(1)\cdots  \varrho(l)}  \bigg)+\cdots 
\end{align*}
For each $d \in \llbracket 0, \ell\rrbracket$, let $Q_d$ be the polynomial of degree at most $d$ (determined by explicit differentiation) such that
\begin{align*}
\bigg( \sum_{\substack{i+j=d \\ i,j\geq 0}} b_{i,j,\ell-d}t^{-i}\partial_t^j\bigg)t^{\alpha+m+d} = Q_d(l)t^{\alpha+m} , \qquad m \geq -d .
\end{align*}
Note that $Q_0 = b_{0,0,\ell}=1$. 
Now, we have 
\begin{align*}
t^{\alpha-\ell} \sum_{k\geq 0} t^kP_k 
= \sum_{d=0}^\ell \sum_{j\geq-d}\frac{Q_d(j)}{\varrho(1)\cdots  \varrho(j+d)}t^{\alpha+j}L_{-1}^{\ell+j}+\cdots .
\end{align*}
By assumption, we know that the coefficients of monomials in $L_{-1}$ of degree $\ell+j$ for $j \in \llbracket -\ell, \ell+3\rrbracket$ are vanishing. Thus, we obtain
\begin{align*}
0 = \; & Q_\ell(-\ell) , \\
0 = \; & \frac{Q_\ell(-\ell+1)}{\varrho(1)}+Q_{\ell-1}(-\ell+1) , \\
0 = \; & \frac{Q_\ell(-\ell+2)}{\varrho(1)\varrho(2)}+\frac{Q_{\ell-1}(-\ell+2)}{\varrho(1)}+Q_{\ell-2}(-\ell+2) , \\
\vdots \; & \\
0 = \; & \frac{Q_\ell(0)}{\varrho(1)\cdots  \varrho(\ell)}+\frac{Q_{\ell-1}(0)}{\varrho(1)\cdots  \varrho(\ell-1)}+\cdots  +Q_0(0) , \\
\vdots\; & \\
0 = \; & \frac{Q_\ell(\ell+3)}{\varrho(1)\cdots  \varrho(2\ell+3)}+\frac{Q_{\ell-1}(\ell+3)}{\varrho(1)\cdots  \varrho(2\ell+2)}+\cdots  +\frac{Q_0(\ell+3)}{\varrho(1)\cdots  \varrho(\ell+3)} .
\end{align*}
Multiplying the $i$-th equation by $\varrho(1)\cdots  \varrho(i-1)$, we obtain 
\begin{align*}
0 = \; & Q_\ell(-\ell) , \\
0 = \; & Q_\ell(-\ell+1)+Q_{\ell-1}(-\ell+1)\varrho(1) , \\
0 = \; & Q_\ell(-\ell+2)+Q_{\ell-1}(-\ell+2)\varrho(2)+Q_{\ell-2}(-\ell+2)\varrho(1)\varrho(2) , \\
\vdots \; & \\
0 = \; & Q_\ell(0)+Q_{\ell-1}(0)\varrho(\ell)+\cdots  Q_0(0)\varrho(1)\cdots  \varrho(\ell) , \\
\vdots \; & \\
0 = \; & Q_\ell(\ell+3)+Q_{\ell-1}(\ell+3) \varrho(2\ell+3)+\cdots  +Q_0(\ell+3)\varrho(\ell+4)\cdots  \varrho(2\ell+3) .
\end{align*}
Since $\varrho(0)=0$, we find that for all $m \in \llbracket -\ell, \ell+3 \rrbracket$
\begin{align*}
0 = Q_\ell(m)+Q_{\ell-1}(m)\varrho(\ell+m)+\cdots  +Q_0(m)\varrho(m+1)\cdots  \varrho(\ell+m) 
=: \; & O(m) .
\end{align*}
On the one hand, since $Q_0=1$, the last term is non-vanishing and of degree $2\ell$, 
while all the other terms are of degree at most $2\ell-1$. 
Thus, $O$ is not the zero polynomial. 
On the other hand, since $O$ is a polynomial of degree at most $2\ell$ with $2\ell+4$
zeroes, we infer that $O \equiv 0$. 
This is a contradiction. 
Hence, we conclude that there exists an element $P_k \in \mathcal U(\Vir^-) \setminus \{0\}$ of degree $d < 2\ell+4$ such that $P_k v_0=0$. 
This concludes the proof of Lemma~\ref{lem:analoglemma1Dubedat}.
\end{proof}

\subsection{Virasoro action on the determinant line bundle} 
\label{subsec:section6p2}

Next, 
we shall describe the geometric framework for the conformal block functions, viewed as sections of a line bundle. 
It turns out that the space $\SolSp_\multii$ of conformal block functions of Proposition~\ref{prop:propdefCvarsigma} carries non-commuting actions 
of the Witt algebra at each variable $x_i \in \bR = \partial \bH$, for $i \in \llbracket 1, \np\rrbracket$. 
This is a manifestation of the (infinitesimal) conformal symmetry in CFT. 
Our aim is to construct a space which carries \emph{commuting} actions of the Virasoro algebra, which leads to a structure underlying the BPZ partial differential equations. 
For this purpose, we first have to pass from the Witt algebra action to
an action of its central extension (viz.~the Virasoro algebra), on a space of sections of a one-dimensional line bundle (\quote{determinant line bundle}) 
over a Teichm\"uller space involving marked boundary points 
(cf.~the variables $(x_1,\ldots,x_\np)$). 

In the constructions and statements below, we mostly follow~\cite[Sect.~2~\&~4]{Dubedat:SLE_and_Virasoro_representations_localization}.

\subsubsection{Extended Teichm\"uller space and determinant lines}

Let $S$ be a simply-connected, compact Riemann surface with a single boundary component $\partial S$ and  
marked points $x_i \in \partial S$, for $i \in \llbracket 1, \np\rrbracket$. 
We endow $S$ with the following additional data. 
Let $z$ be a local coordinate at $x \in \partial S$. 
A \emph{$k$-jet at $x$} is an element of $\bR[z]/(z^{k+1}\bR[z])$ with a first order zero: for each
\begin{align*}
\eta = \sum_{i \geq 1} \eta_i z^i \; \in \; \bR[z] , \qquad \eta_1 > 0 ,
\end{align*}
we denote the associated $k$-jet as $[\eta]_k = \sum_{i=1}^k \eta_i z^i$.
For each $\bs{k} = (k_1,\ldots, k_\np)$, we define $\mathcal{T}_{\bs{k}}$ to be the space of equivalence classes of surfaces $S$ as above with a $k_i$-jet at $k_i$ at $x_i$, for $i \in \llbracket 1, \np\rrbracket$, 
where each equivalence class consists of all marked surfaces related by conformal isomorphisms sending marked points to marked points and $k_i$-jets to $k_i$-jets.

For each surface $S$, let $\cconf(S)$ be the set of conformal metrics on $S$
which near the boundary are pushforwards of the flat metric from 
the cylinder, so that in particular the boundary $\partial S$ is a geodesic. 
For two such conformally equivalent metrics $g \in \cconf(S)$ and $e^{2\sigma}g \in \cconf(S)$, with Weyl factor $\sigma \in C^\infty(S, \bR)$, 
we define the \emph{conformal anomaly}
\begin{align*}
\lfunct(\sigma, g) := \frac{1}{12 \pi} \iint_S \bigg(
\frac{1}{2} |\nabla_g \sigma|_g^2 + R_g \sigma
\bigg) \vol_g ,
\end{align*}
where $\nabla_g$, $R_g$, and $\vol_g$ are respectively the divergence, 
Gaussian curvature, and volume form on $S$ in the metric $g$.
We then define the (real) \emph{determinant line} associated to $S$ 
as the one-dimensional $\bR$-vector space $\Det_\bR(S) := (\bR \times \cconf(S))/_\sim$ 
consisting of pairs $(r,g) = r[g]$, 
where $r \in \bR$ and $g \in \cconf(S)$, subject to 
the equivalence relation \quote{$\sim$} given by 
$[g] = e^{-\lfunct(\sigma, g)} [e^{2\sigma}g]$ in terms of the anomaly.
(See also~\cite{Friedrich:On_connections_of_CFT_and_SLE, Kontsevich-Suhov:On_Malliavin_measures_SLE_and_CFT, Dubedat:SLE_and_Virasoro_representations_localization, Maibach-Peltola:From_the_conformal_anomaly_to_the_Virasoro_algebra}.)

We view $\Det := \{\Det_\bR(\bH)\}$ as the determinant line \quote{bundle} over the (genus zero, trivial) Teichm\"uller space $\mathcal{T} = \{\bH\}$ of simply-connected, compact Riemann surfaces with a single boundary component without any marked points 
(that can be represented by the upper half-plane $S = \bH$, say). 
We then define the \emph{determinant line bundle} $\Det_{\bs{k}}$ over $\mathcal{T}_{\bs{k}}$ as the pull-back of $\Det$ under the projection forgetting marked points and jets. 
In the spirit of the infinitesimal conformal symmetry in CFT~\cite{DMS:CFT,Schottenloher:Mathematical_introduction_to_CFT} 
and Virasoro uniformization~\cite{Kontsevich:Virasoro_and_Teichmuller_spaces, Beilinson-Schechtman:Determinant_bundles_and_Virasoro_algebra}, 
as explained in detail in~\cite[Sect.~2.4.4]{Dubedat:SLE_and_Virasoro_representations_localization}, 
there exists an action of the Witt algebra as (local) differential 
and multiplication operators $L_n^0 \mapsto -z^{n+1} \partial_z$ 
such that the negative generators send smooth functions on $\mathcal{T}_{\bs{k}}$ to 
smooth functions on $\mathcal{T}_{\bs{k}'}$ for some $k_i'\geq k_i$ for all $i$. 
To make the action on this tower $(\mathcal{T}_{\bs{k}})_{\bs{k}}$ closed, 
one considers the \emph{projective limit} 
given by the smooth projections $\mathcal{T}_{\bs{k}'}\rightarrow\mathcal{T}_{\bs{k}}$, consisting of truncations of the jets for $k_i'\geq k_i$ 
(see~\cite[Sect.~2.4.4]{Dubedat:SLE_and_Virasoro_representations_localization} for a detailed account),
\begin{align*}
\mathcal{T}_\infty := \lim_{\leftarrow} \mathcal{T}_{\bs{k}} .
\end{align*}
Elements of $\mathcal{T}_\infty$ may be thought of as equivalence classes of surfaces $S$ with marked points as above, 
but with \emph{formal coordinates}\footnote{The key difference is that is when a local coordinate $z$ is given, a formal coordinate is an element of $z\bR[[z]]$.}  
at each marked point instead of $k$-jets. 
The result in~\cite[Thm.~4]{Dubedat:SLE_and_Virasoro_representations_localization} shows that the space $C^\infty (\mathcal{T}_\infty,\Det_\infty)$ of sections  
of the pull-back bundle $\Det_\infty$ over $\mathcal{T}_\infty$
obtained from the projective limit construction 
carries a representation of $\np$ \emph{commuting copies of the Virasoro algebra} 
with central charge $c=1$: 
one for each marked point $x_1, \ldots, x_\np \in \partial S$. 
As the details of this construction are irrelevant for the purposes of understanding 
the present work, 
we refer the readers to~\cite{Dubedat:SLE_and_Virasoro_representations_localization, Dubedat:SLE_and_Virasoro_representations_fusion} for more details,
and only highlight the key ingredients for proving Theorem~\ref{thm:intermediatetheorem}.

\subsubsection{Conformal block functions as sections of the determinant line bundle}

After choosing a reference section $\mu_\zeta$ of the bundle $\Det_\infty$ 
(which can, for example, be constructed from the zeta-regularized determinant of the Laplacian~\cite[Sect.~3]{Dubedat:SLE_and_Virasoro_representations_localization}), 
we shall denote sections in $C^\infty (\mathcal{T}_\infty,\Det_\infty)$ 
by $f \mu_\zeta$, where $f \in C^\infty (\mathcal{T}_\infty)$. 
The functions $f$ will play the role of the correlation functions in $\SolSp_\multii$.
Indeed, to any given smooth function $F \colon \chamber_\np \to \bC$, 
we associate a smooth function $f \in C^\infty (\mathcal{T}_\infty)$ 
as (the lift\footnote{Abusing notation, we identify $f \in C^\infty (\mathcal{T}_{(1,...,1)})$ with its pullback under the projection map, $f \in C^\infty (\mathcal{T}_\infty)$.} of the one) described in Equation~\eqref{eq:deftildef} below.

We will use a convenient choice of smooth coordinates on $\mathcal{T}_{\bs{k}}$, 
associated to the choice of reference surface $S=\bH$ of $\mathcal{T}$ 
with coordinates $z$ around $0$ and $-1/z$ around $\infty$. 
Thanks to the action of the M\"obius group, 
we may also choose two of the marked points to be 
$x_1=0$ and $x_{\np+1} = \infty$, 
and we may choose the first order of the jet at $\infty$ to equal one:
thus, for one and two marked points, respectively, we have
\begin{align*}
[(\bH; 0; [\eta\super{0}]_k)] \in \; & \mathcal{T}_{k} , \qquad 
[\eta\super{0}]_k := \sum_{i=1}^k \eta\super{0}_i z^i , \quad 
\eta\super{0}_i \in \bR ,
\\
[(\bH; 0, \infty; [\eta\super{0}]_{k}, [\eta\super{\infty}]_{l})] \in \; & \mathcal{T}_{k,l}, \qquad 
[\eta\super{\infty}]_l := -\frac{1}{z} + \sum_{i=2}^l (-1)^i \eta\super{\infty}_i z^{-i} , \quad \eta\super{\infty}_i \in \bR ,
\end{align*}
and for at least three marked points, we obtain the representatives 
\begin{align*}
[(\bH; 0, x_2, \ldots , x_\np, \infty; [\eta\super{0}]_{k_1}, \; & [\eta\super{x_2}]_{k_2}, \ldots , [\eta\super{x_\np}]_{k_\np}, [\eta\super{\infty}]_{k_{\np+1}})] \in  \mathcal{T}_{k_1,\ldots,k_{\np+1}}, \\
[\eta\super{x_j}]_k := \; & \sum_{i=1}^k \eta\super{x_j}_i (z-x_j)^i , \quad \eta\super{x_j}_i \in \bR ,
\; j \in \llbracket 2,\np \rrbracket .
\end{align*}
Thus, the following collection provides a set of smooth coordinates on $\mathcal{T}_{\bs{k}}$ and hence on $\mathcal{T}_\infty$:
\begin{align} \label{eq:Teich_coord}
(x_2,\ldots, x_\np; \eta\super{0}_1, \ldots, \eta\super{0}_{k_1}; 
\eta\super{x_2}_1, \ldots, \eta\super{x_2}_{k_2}; 
\ldots; \eta\super{x_\np}_1, \ldots, \eta\super{x_\np}_{k_\np}; 
\ldots; \eta\super{\infty}_2, \ldots, \eta\super{\infty}_{k_{\np+1}}).
\end{align}
In particular,
given any smooth function $F \colon \chamber_\np \to \bC$, 
taking $k_1 = k_2 = \cdots = k_\np = k_{\np+1} = 1$, 
\begin{align} \label{eq:deftildef}
\begin{split}
f[(\bH; \; & 0, x_2, \ldots , x_\np, \infty; [\eta\super{0}]_{1}, [\eta\super{x_2}]_{1}, \ldots , [\eta\super{x_\np}]_{1}, [\eta\super{\infty}]_{1})]
\\
:= \; & (\eta\super{0}_1)^{-h_{s_1+1}} \prod_{j=2}^\np (\eta\super{x_j}_1)^{-h_{s_j+1}} \times F(0, x_2, \ldots , x_\np) 
\end{split}
\end{align}
where $(s_1,\ldots,s_\np)=\multii \in \bZpos^\np$, 
defines a smooth function $f \in C^\infty (\mathcal{T}\sub{1,1,\ldots,1})$,
which lifts to a smooth function in $C^\infty (\mathcal{T}_\infty)$. 
By virtue of~\cite[Thm.~4]{Dubedat:SLE_and_Virasoro_representations_localization} 
there are $\np+1$ commuting copies of the Virasoro algebra acting 
on the section $f \mu_\zeta \in C^\infty (\mathcal{T}_\infty,\Det_\infty)$, 
with one copy corresponding to each marked point $x_1=0, x_2, \ldots , x_\np, \infty$. 
We denote the generators in the $\Vir$-action associated to the marked point $x_j$ by $L_n\super{x_j}$. 
By construction, this action has central charge $c=1$. 
Furthermore, these representations are in fact highest-weight modules with highest-weight vectors $f \mu_\zeta$:
for each $j \in \llbracket 1,\np \rrbracket$, we have
\begin{align} \label{eq:highest-weight sections}
\begin{split}
L_0\super{x_j}(f \mu_\zeta) = \; & h_{s_j+1} (f \mu_\zeta) , \\
L_n\super{x_j}(f \mu_\zeta) = \; & 0 ,\qquad \textnormal{for all } n > 0.
\end{split}
\end{align}
The representation at $\infty$ has weight zero: 
$L_n\super{\infty}(f \mu_\zeta) = 0$ for all $n \geq 0$. 
Moreover, we have
\begin{align*}
\Delta_\ell\super{x_j} = \mathcal{D}_\ell\super{x_j} + D_\ell\super{x_j},
\end{align*}
where $\Delta_\ell\super{x_j}$ is the partial differential operator~\eqref{eq:defBSAoperator} involving $L_i\super{x_j}$ for all $i$, 
and $\mathcal{D}_\ell\super{x_j} = \mathcal{D}_\ell\super{j}$ is defined in~\eqref{eq:defBSAdifferentialoperator}, 
and $D_\ell\super{x_j}$ is a differential operator which vanishes if $\eta\super{x_j}_i = 0$ for all $i>1$. 
This shows that, if $F$ satisfies the BPZ equations~\eqref{eq:saintaubineq} at each marked point, then
\begin{align} \label{eq:BPZlifted}
\Delta_{s_j+1}\super{x_j} (f \mu_\zeta) = 0 ,
\qquad \textnormal{for all }j \in \llbracket 1, \np\rrbracket .
\end{align}
Conversely, if $f \mu_\zeta$ satisfies 
the \quote{null-vector} equations~\eqref{eq:BPZlifted}, 
then $F$ satisfies the BPZ PDEs~\eqref{eq:saintaubineq}
(see~\cite[Sect.~4]{Dubedat:SLE_and_Virasoro_representations_localization}). 
In conclusion, we have related solutions $F$ to the BPZ PDEs~\eqref{eq:saintaubineq} 
to solutions $f \mu_\zeta$ to Equations~\eqref{eq:BPZlifted}
via the correspondence of Equation~\eqref{eq:deftildef}.

\subsection{Fusion of BPZ PDEs --- proof of Theorem~\ref{thm:intermediatetheorem}}
\label{subsec:proof}

\intermediatetheorem*

\begin{remark}
Our Theorem~\ref{thm:intermediatetheorem} as well as its proof are very closely related to~\cite[Thm.~15]{Dubedat:SLE_and_Virasoro_representations_fusion}.
However,~\cite[Thm.~15]{Dubedat:SLE_and_Virasoro_representations_fusion} 
only applies to irrational central charges, 
because of a certain algebraic result required to carry out the 
argument~\cite[Lem.~1]{Dubedat:SLE_and_Virasoro_representations_fusion}. 
The reason for this is that the structure of highest-weight modules 
of the Virasoro algebra is much more intricate when the central charge is rational. 
Our Lemma~\ref{lem:analoglemma1Dubedat} is an extension of~\cite[Lem.~1]{Dubedat:SLE_and_Virasoro_representations_fusion} 
to the case of unit central charge. 
On the other hand,~\cite[Lem.~12,~13,~14]{Dubedat:SLE_and_Virasoro_representations_fusion}, 
which are used for building the bridge between analytic geometry and algebra, 
do apply to any central charge. 
Therefore, we can use all of them for the proof of Theorem~\ref{thm:intermediatetheorem}.
\end{remark}

\begin{proof}[Proof of Theorem~\ref{thm:intermediatetheorem}]
From $F$ as in the statement, 
we construct $f \in C^\infty (\mathcal{T}\sub{1,1,\ldots,1})$ as in~\eqref{eq:deftildef}:
\begin{align*}
f[(\bH; \; & 0, x_2, \ldots , x_\np, \infty; [\eta\super{0}]_{1}, [\eta\super{x_2}]_{1}, \ldots , [\eta\super{x_\np}]_{1}, [\eta\super{\infty}]_{1})]
\\
:= \; & (\eta\super{0}_1)^{-h_{s_1+1}} (\eta\super{x_k}_1)^{-h_{\ell}} (\eta\super{x_{k+1}}_1)^{-h_{2}} 
\prod_{\substack{2 \leq j \leq \np \\ j \neq k,k+1}} (\eta\super{x_j}_1)^{-h_{s_j+1}} \times F(0, x_2, \ldots , x_\np) ,
\end{align*}
where in the coordinates~\eqref{eq:Teich_coord} on $\mathcal{T}\sub{1,1,\ldots,1}$, we have
the $1$-jets $[\eta\super{\infty}]_1 = -1/z$ and
\begin{align*}
[\eta\super{0}]_1 = \; & \eta\super{0}_1 z , \quad \eta\super{0}_1 \in \bR ,
\\
[\eta\super{x_j}]_1 = \; & \eta\super{x_j}_1 (z-x_j) , \quad \eta\super{x_j}_1 \in \bR ,
\; j \in \llbracket 2,\np \rrbracket .
\end{align*}
The first step of the proof is to derive an asymptotic expansion for the section 
$f \mu_\zeta$ as $x_{k+1} \to x_k$, starting from the assumed asymptotic expansion~\eqref{eq:asyf} of $F$. 
To this end, note that we have $\eta\super{x_{k+1}} = \eta\super{x_k} - \eta\super{x_k}(x_{k+1})$, which implies that
\begin{align*}
\sum_{j \geq 1} \eta\super{x_{k+1}}_j (z-x_{k+1})^j = \sum_{j \geq 1} \eta\super{x_k}_j \big( (z-x_k)^j - (x_{k+1}-x_k)^j) \big).
\end{align*}
Taking the derivative with respect to $z$ and evaluating at $z=x_{k+1}$ yields
\begin{align*}
\eta\super{x_{k+1}}_1 = \sum_{j \geq 1} j \, \eta\super{x_k}_j (x_{k+1} - x_k)^{j-1} .
\end{align*}
Hence, using the expansion~\eqref{eq:asyf}, we infer that
\begin{align*}
f[(\bH; \; & 0, x_2, \ldots , x_\np, \infty; [\eta\super{0}]_{1}, \ldots , [\eta\super{x_k}]_{1}, [\eta\super{x_k} - \eta\super{x_k}(x_{k+1})]_1 , [\eta\super{x_{k+2}}]_{1}, \ldots , [\eta\super{\infty}]_{1})]
\\
:= \; & (\eta\super{0}_1)^{-h_{s_1+1}} (\eta\super{x_k}_1)^{-h_{\ell}} 
\prod_{\substack{2 \leq j \leq \np \\ j \neq k,k+1}} (\eta\super{x_j}_1)^{-h_{s_j+1}} 
\times \Big( \eta\super{x_k}_1 + \sum_{j \geq 2} j \, \eta\super{x_k}_j (x_{k+1} - x_k)^{j-1} \Big)^{-h_{2}} 
\\
\; & \times 
(x_{k+1}-x_k)^{h_{\ell+1}-h_{\ell}-h_{2}} 
\sum_{i\geq 0} f_i(\ldots, x_k,x_{k+2},\ldots) (x_{k+1}-x_k)^i .
\end{align*}
It is crucial to note that, while $k$-jets use \emph{formal} local coordinates, with possibly zero as radius of convergence, 
in order to carry out the fusion argument for the PDEs 
it is necessary to establish a \emph{true series expansion} in genuine local coordinates with a \emph{positive} radius of convergence.
This we obtain for our \emph{explicit} functions from Lemma~\ref{lem:lemmaASY}, which gives~\eqref{eq:asyf}.

The Lagrange inversion theorem now allows us to write 
\begin{align*}
x_{k+1}-x_k = \sum_{i \geq 1} g_i (\eta\super{x_k}(x_{k+1}) - \eta\super{x_k}(x_k))^i,
\end{align*}
where $g_i$ is a rational function of $\{x_k,\eta\super{x_k}_1, \ldots, \eta\super{x_k}_{n_i}\}$ where $n_i$ is a finite integer for all $i\geq 1$ and, 
in particular, $g_1 = 1/\eta\super{x_k}_1$. 
This justifies that we indeed have the expansion
\begin{align} \label{eq:expansiontildef}
\begin{split}
f[(\bH; \; & 0, x_2, \ldots , x_\np, \infty; [\eta\super{0}]_{1}, \ldots , [\eta\super{x_k}]_{1}, [\eta\super{x_k} - \eta\super{x_k}(x_{k+1})]_1 , [\eta\super{x_{k+2}}]_{1}, \ldots , [\eta\super{\infty}]_{1})]
\\
:= \; & \big( \eta\super{x_k}(x_{k+1}) - \eta\super{x_k}(x_{k}) \big)^{h_{\ell+1}-h_{\ell}-h_{2}}
\sum_{i\geq 0} \tilde{f}_i \, \big( \eta\super{x_k}(x_{k+1}) - \eta\super{x_k}(x_{k}) \big)^i ,
\end{split}
\end{align}
where $\tilde{f}_i$ is a smooth function of 
$\{x_2, \ldots , x_\np, \infty\}$ as well as of 
$\{\eta\super{x_j}_1 \;|\; j \neq k,k+1 \}$ and of 
$\{\eta\super{x_k}_1, \ldots , \eta\super{x_k}_{m_i}\}$ for some finite integer $m_i$. 
Note also that the coefficients $\tilde{f}_i$ are smooth because they are products of compositions of smooth functions. 
Moreover, $m_0 = 1$ and 
\begin{align} \label{eq:deftildef0}
\tilde{f}_0 = (\eta\super{0}_1)^{-h_{s_1+1}}  (\eta\super{x_k}_1)^{-h_{\ell+1}}
\prod_{\substack{2 \leq j \leq \np \\ j \neq k,k+1}} (\eta\super{x_j}_1)^{-h_{s_j+1}} 
\times f_0(x_1,\ldots,x_k,x_{k+2},\ldots,x_\np) .
\end{align}
Therefore, $\tilde{f}_0  \mu_\zeta$ is a $\Vir$-highest-weight vector as in~\eqref{eq:highest-weight sections} 
with weight $h_{s_j+1}$ at $x_j$ for $j \neq k,k+1$, 
and weight $h_{\ell+1}$ at $x_k$. 
The other coefficients give smooth sections $\tilde{f}_i \mu_\zeta$.

The final step of the proof is to connect this analytic setting to the algebraic Lemma~\ref{lem:analoglemma1Dubedat}. 
Consider the space $C^\infty (\mathcal{T}_\infty,\Det_\infty) \otimes V_{\alpha,h_{2}}$ of formal series as in Section~\ref{subsec:section5p1p2}, and set 
\begin{align*}
\mathcal Z := t^\alpha \sum_{i\geq 0} (\tilde{f}_i \mu_\zeta) t^i , \qquad
\alpha := h_{\ell+1} - h_{\ell} - h_{2} . 
\end{align*}
We now express the action of the $\np+1$ $\Vir$-copies on $\mathcal Z$ 
in terms of the action of the $\np$ $\Vir$-copies on $f \mu_\zeta$. 
Specifically, 
\cite[Lem.~12]{Dubedat:SLE_and_Virasoro_representations_fusion} 
identifies the $\Vir$-action at $x_k$ with the action of $\hat{L}_n$ in~\eqref{eq:actionhatLn};
\cite[Lem.~13]{Dubedat:SLE_and_Virasoro_representations_fusion} 
identifies the action of the generators $L_{-1}\super{x_{k+1}}$ and $L_{-2}\super{x_{k+1}}$ with $\Tilde{L}_{-1}$~\eqref{eq:actiontildeL1} and $\Tilde{L}_{-2}$~\eqref{eq:actiontildeL2}, respectively;
and~\cite[Lem.~14]{Dubedat:SLE_and_Virasoro_representations_fusion} 
relates the action of $L_n\super{x_j}$ on $f \mu_\zeta$ associated to the \quote{spectator points} $x_j$, with $j \neq k,k+1$,
to the action of $L_n\super{x_j}$ on $\mathcal Z$. 
Also, by construction, the section $f \mu_\zeta$ is $\Vir$-highest-weight as in~\eqref{eq:highest-weight sections} 
with weight $h_{s_j+1}$ at $x_j$ for $j \neq k,k+1$, 
weight $h_{\ell}$ at $x_k$, and weight $h_{2}$ at $x_{k+1}$. 
Since by assumption $F$ in~\eqref{eq:deftildef} 
satisfies the corresponding BPZ equations at those points, 
$f \mu_\zeta$ satisfies~\eqref{eq:BPZlifted}:
\begin{align*}
\begin{cases}
\Delta_{s_j+1}\super{x_j} \, (f \mu_\zeta) =  0 , \qquad j \neq k,k+1 , \\
\Delta_{\ell}\super{x_k} \, (f \mu_\zeta) =  0 , \\
\Delta_{2}\super{x_{k+1}} \, (f \mu_\zeta) =  0 .
\end{cases}
\end{align*}
This implies in particular that 
$\hat{\Delta}_\ell (\mathcal Z) = 0$ and $\tilde{\Delta}_2 (\mathcal Z) = 0$. 
To finish the proof, we just need to apply Lemma~\ref{lem:analoglemma1Dubedat} to infer that $\tilde{f}_0$ in~\eqref{eq:deftildef0} satisfies
\begin{align*}
\begin{cases}
\Delta_{s_j+1}\super{x_j} \,  (\tilde{f}_0 \mu_\zeta) =  0 , \qquad j \neq k,k+1 , \\
\Delta_{\ell+1}\super{x_k} \,  (\tilde{f}_0 \mu_\zeta) =  0 .
\end{cases}
\end{align*}
From this, we conclude that $f_0$ satisfies the asserted BPZ equations~(\ref{eq:jBSA},~\ref{eq:kp1BSA}). 
\end{proof}

\appendix 

\bigskip{}
\section{Expressions for Schur polynomials} 
\label{app:appendixschur}
In this appendix, we recall the Schur polynomials, 
used in particular for the proof of Proposition~\ref{prop:combinatorialformula}. 
Recall first that we consider sets of variables $x_{\summ_k},\ldots ,x_{\summ_{k+1}-1}$, where 
\begin{align*}
\summ_k := 1+ \sum_{j=1}^{k-1} s_j, \qquad k \in \llbracket 1, \np\rrbracket.
\end{align*} 
In particular, we have $\summ_{k+1}-\summ_k = s_k$. 
Moreover, let $\lambda=(\lambda_i)_{i=1}^{s_k}$ be a partition. 
The \emph{Schur polynomial} associated with the partition $\lambda$ admits the bialternant formula 
\begin{align*} 
\Schur_\lambda(x_{\summ_k},\ldots ,x_{\summ_{k+1}-1}) 
= \frac{\det\big( x_i^{\lambda_j+s_k-j}\big)_{q_k \leq i,j \leq q_{k+1}-1}}{\prod_{\summ_k\leq i < j < \summ_{k+1}} (x_i-x_j)}.
\end{align*}
Utilizing the Leibniz formula for the determinant, this can also be written as
\begin{align} \label{eq:defslambda}
\Schur_\lambda(x_{\summ_k},\ldots ,x_{\summ_{k+1}-1}) = \frac{\sum_{\sigma \in \SymGrp_{s_k}} \sign(\sigma) \prod_{i=\summ_k}^{\summ_{k+1}-1} x_{\sigma(i)}^{\lambda_i+s_k-i}}{\prod_{\summ_k\leq i < j < \summ_{k+1}} (x_i-x_j)}.
\end{align}
Equivalently, the Schur polynomial also admits the following combinatorial formula:
\begin{align} \label{eq:combinatorialschur}
\Schur_\lambda(x_{\summ_k},\ldots ,x_{\summ_{k+1}-1}) = \sum_T x_{\summ_k}^{t_{\summ_k}} \cdots  x_{\summ_{k+1}-1}^{t_{\summ_{k+1}-1}},
\end{align}
where the sum is taken over all column-strict tableaux $T$ with shape $\lambda$, with entries in $\{1,\dots, s_k\}$ (and any content), 
and where each $t_i$ is the number of occurrences of the number \quote{$i$} in the tableau $T$. 
In particular, the evaluation of~\eqref{eq:combinatorialschur} 
at $x_i=x_k$ for all $i \in \llbracket \summ_k, \summ_{k+1}-1\rrbracket$ 
leads to
\begin{align} \label{eq:evaluationschur}
\Schur_\lambda(x_k,\ldots ,x_k) = \Schur_\lambda(1,\ldots ,1) \, x_k^{|\lambda|} ,
\end{align}
where $\Schur_\lambda(1,\ldots ,1)$ represents the number of column-strict Young tableaux 
of shape $\lambda$ with entries in $\{1,\dots, s_k\}$ (and any content): 
\begin{align*}
\Schur_\lambda(1,\ldots ,1) = \prod_{1\leq i<j\leq s_k} \frac{\lambda_i-\lambda_j+j-i}{j-i}.
\end{align*}

%\newpage

\newcommand{\changeurlcolor}[1]{\hypersetup{urlcolor=#1}}      
\changeurlcolor{black}

\bibliographystyle{alpha}
\newcommand{\etalchar}[1]{$^{#1}$}

\end{document}